\documentclass[11pt,a4paper]{article}

\usepackage[a4paper,text={150mm,240mm},centering,headsep=10mm,footskip=15mm]{geometry}

\usepackage[english]{babel}
\usepackage[utf8]{inputenc}
\usepackage[T1]{fontenc}
\usepackage{epsfig, float}
\usepackage{appendix}
\usepackage{cancel}

\usepackage{mathcomp}
\usepackage{dsfont}
\usepackage{lmodern}
\usepackage{scrextend}
\usepackage{color}

\usepackage[nottoc]{tocbibind}

\usepackage{amsmath,amssymb,amsfonts,amsthm}
\usepackage{amsfonts}
\usepackage{url}
\usepackage{bbm}
\usepackage{mathrsfs}

\usepackage{cancel}

\usepackage[normalem]{ulem}
\usepackage{todonotes}

\definecolor{miguelscolor}{rgb}{.7,.2,.2}
\definecolor{dirkscolor}{rgb}{.2,.2,.7}
\definecolor{felixscolor}{rgb}{.2,.7,.2}

\newcommand{\norm}[1]{\left\lVert#1\right\rVert}

\newcommand{\R}{\mathbb{R}}
\newcommand{\C}{\mathbb{C}}
\newcommand{\N}{\mathbb{N}}

\renewcommand{\Re}{\operatorname{Re}}
\renewcommand{\Im}{\operatorname{Im}}

\newtheorem{theorem}{Theorem}[section]
\newtheorem{lemma}[theorem]{Lemma}
\newtheorem{proposition}[theorem]{Proposition}
\newtheorem{remark}[theorem]{Remark}
\newtheorem{definition}[theorem]{Definition}
\newtheorem{corollary}[theorem]{Corollary}

\numberwithin{equation}{section} 

\date{\today}

\title{\Large{\textsc{  Analyticity of Resonances and Eigenvalues and Spectral Properties of the massless Spin-Boson Model    }}}

\author{Miguel Ballesteros\thanks{\texttt{miguel.ballesteros@iimas.unam.mx},
    Instituto de Investigaciones en Matem\'aticas Aplicadas y en Sistemas,
    Universidad Nacional Aut\'anoma de M\'exico}, Dirk-Andr\'e
    Deckert\thanks{\texttt{deckert@math.lmu.de}, Mathematisches Institut
    der Ludwig-Maximilians-Universität München}, Felix
    H\"anle\thanks{\texttt{haenle@math.lmu.de}, Mathematisches Institut
der Ludwig-Maximilians-Universität München}}
\begin{document}

\maketitle

\begin{abstract}
We extend the method of  Pizzo  multiscale  analysis for resonances introduced in \cite{bbp} in order to  infer analytic properties of resonances and eigenvalues (and their  eigenprojections)  as well as estimates for the localization of the spectrum of dilated Hamiltonians and norm-bounds for the corresponding resolvent operators,  in  neighborhoods of resonances and eigenvalues. We apply our method to the massless Spin-Boson model  assuming a slight infrared regularization. We prove that the resonance and the ground-state eigenvalue (and their  eigenprojections) are analytic with respect to the dilation parameter and  the coupling constant. Moreover, we prove that the spectrum of the dilated Spin-Boson Hamiltonian in the neighborhood of the resonance and the ground-state eigenvalue is localized in two cones in the complex plane with vertices at the location of the resonance and the ground-state eigenvalue, respectively. Additionally, we provide norm-estimates for the resolvent of the dilated Spin-Boson Hamiltonian near the resonance and the ground-state eigenvalue. The topic of analyticity of eigenvalues and resonances has   let to several  studies and advances in the past. However, to the best of our knowledge,   this is the first time that it is addressed from the perspective of  Pizzo  multiscale analysis.  Once the multiscale analysis is set up our method gives easy access to analyticity: Essentially, it  amounts to proving it  for isolated eigenvalues only and use that uniform limits of analytic functions are analytic.  The type of spectral and resolvent estimates that we prove are needed to control the time evolution including the scattering regime. The latter will be demonstrated in a forthcoming publication. The introduced multiscale  method  to study spectral and resolvent estimates  follows its own inductive scheme and  is independent (and different) from the method we  apply to construct resonances.     
\end{abstract}

\section{Introduction}
\label{sec:introduction}

In this paper, we  analyze the massless Spin-Boson model which is a simple but non-trivial model of quantum field theory. It can be seen as a model of a two-level atom interacting with its second quantized scalar field, and hence, provides a widely employed model for quantum optics.
The unperturbed energies of the two-level atom shall be denoted by real numbers $e_0<e_1$.  After switching on the interaction with a massless scalar field that may induce transitions between the atom levels, the free ground-state energy $e_0$ is shifted to the interacting ground-state energy $\lambda_0$ while the free excited state with energy $e_1$ turns into a resonance with complex energy $\lambda_1$.

Unfortunately, in the massless case, neither $e_0$ nor $e_1$ are isolated points in the spectrum of the free Hamiltonian which is why the interacting ground-state and the resonance  ($\lambda_0$ and $\lambda_1$) cannot be constructed using standard results from regular perturbation theory. 
Several technologies were developed to overcome this difficulty. Two succesful methods that recently received a lot of attention are: The so-called  Pizzo  multiscale method (see e.g.\ \cite{pizzo1,pizzo2,bbp,bach}) and  the  renormalization group method (see e.g.\ \cite{bfs1,bfs2,bfs3,bcfs,bfs100,bbf,fgs100,feshbach,s,f,bffs}).  In both cases, a family of  spectrally dilated Hamiltonians is analyzed since this allows for complex eigenvalues.
In  this paper, we will  employ  the  Pizzo  multiscale technique.  This technique invokes an infrared cut-off which is then removed using an inductive scheme. In each step of the induction, lower and lower boson momenta are added to the interaction and regular perturbation theory is used to construct the respective ground-state and resonance. In order to reach the limit of no infrared cut-off good control over the closing gap is essential.
Note that  such a procedure has  been introduced for the construction of resonances  in the Pauli-Fierz model \cite{bach,bbp}. In this  work we also construct resonances (and ground-state eigenvalues) but this is not our main purpose.  Our main purpose is to prove that resonances are analytic with respect to the dilation parameter and coupling constant, and furthermore, to provide certain  spectral and resolvent estimates that allow for the control of the dynamics including the scattering regime.  In a forthcoming paper,  these estimates are employed to address scattering theory for the model  at hand.  We want to remark that in  \cite{bmw} the time evolution of this model is  studied  using the spectral renormalization method.  Some results derived therein  are similar to some of ours, however,  utilizing different methods, respectively.

What we call resonances and ground-states  multiscale analysis  is an inductive construction of a sequence of Hamiltonians that enjoy infrared cutoffs and satisfy certain properties. As the parameter of the sequence tends to infinity these cutoffs are removed.  
Our mutliscale analysis is the content of Theorem \ref{thm:ind}. Its basic scheme is presented in Section \ref{sec:indscheme}  and the proof of Theorem  \ref{thm:ind} is carried out in Section  \ref{sec:proof-induction}. Our proofs of analyticity and  of  resolvent and spectral estimates are not part of our multiscale analysis employed in the contruction of the resonance and ground-state energy, they only use it as  mathematical  input. The latter results are presented in Section \ref{resolvent-estimates} (spectral and resolvent estimates) and Section \ref{analyticity} (analyticity).  Theorem \ref{thm:ind}  is only an intermediate but necessary step.  As we mention above, the method of multiscale analysis for resonances is introduced in \cite{bach,bbp}. However, the results in \cite{bach,bbp}  cannot be   used directly to prove analyticity because many of the estimations therein consider the dilation parameter, $\theta$, to be purely imaginary whereas analyticity requires estimates that are uniform for $\theta$ in an open set.  Thus, although it does not involve major obstacles, for the sake of analyticity, many of the given calculations  and some of the proofs need to be redone. For the convenience of the reader and in order to keep this  work self-contained  we  provide them in the proof of Theorem \ref{thm:ind}. Note  that Theorem \ref{thm:ind} is applied to the Spin-Boson model  while  \cite{bach,bbp} address  the Pauli-Fierz model. This gives us the opportunity to review the  Pizzo  multiscale technique for a non-trivial but more tractable model. 

In Section \ref{resolvent-estimates} (in particular, in Subsection \ref{resspec}) we introduce a new inductive scheme that is used  to study resolvent and spectral estimates. This scheme is independent and different from the scheme used  in Section \ref{sec:4}  to construct resonances. It allows to localize the spectrum in two cones with  vertices at the location of the resonance and ground-state energy, respectively, and allows for arbitrary small apex angles provided the coupling constant is sufficiently small.  We want to emphasize that such a result requires a  more subtle analysis than localizing the spectrum in cuspidal domains. Additionally, we provide estimates for the resolvent operator in the  vicinity of the cones.   

The study of analytic properties of resonances and  ground-state  eigenvalues in the context of non-relativistic quantum field theory has been the source of several  studies. These papers use the method of spectral renormalization. 
In \cite{gh}  a large class of models of quantum field theory was analyzed and analyticity of the ground-state with respect to the coupling constant was proven under the assumption this ground-state is non-degenerate. The existence of  a  unique ground-state and its analyticity with respect to the coupling constant was shown in \cite{hasler1} for the Spin-Boson model without an infrared regularization, and in \cite{hh101} for the Pauli-Fierz model. Furthermore, in \cite{bffs}, a model describing the interaction of an atom with its quantized electromagnetic field was studied and it was proven,
 that the excited states are analytic functions of the momentum
of the atom and of the coupling constant. Likewise, in \cite{ah}, it is shown that the ground-state energy of the translationally invariant Nelson model is an analytic function of the coupling constant and the total momentum.

Here, to the best of our knowledge, we  give the first extension of the  Pizzo  multiscale method that provides a ready access to analyticity properties that   essentially   amounts to proving it  for isolated eigenvalues only and exploiting that uniform limits of analytic functions are analytic. 

In \cite{bfs1,bfs2,bfs3}, the renormalized group technique was invented and applied in order to construct the ground state and resonances for the confined Pauli-Fierz model. Moreover, resolvent and spectral estimates were obtained therein. Based on this new method, several simplifications and applications were developed in a variety of works \cite{bcfs,bfs100,bffs,feshbach,fgs100,bbf,bffs,s,f,bbf}. The  Pizzo  multiscale analysis was first  invented in \cite{pizzo1,pizzo2} and then adapted in order to gain access to spectral and resolvent estimates and  the construction of ground-states in \cite{bach,bbp}. In \cite{bmw},  resolvent and spectral estimates are derived in order to control the time evolution in the Spin-Boson model and, in \cite{bf}, smoothness of the resolvent and local decay of the photon dynamics for quantum states in a spectral interval  just above the ground state energy was proven.

Next, we will introduce the massless Spin-Boson model together with the well-known mathematical facts and tools in order to present our main results in Section \ref{sec:mainresult}.

\subsection{The Spin-Boson model}
\label{sec:defmodel}

In the rest of this Section 1  we will state preliminary definitions and well-known tools and facts from which we start our analysis. \\

The non-interacting Spin-Boson Hamiltonian is defined as
\begin{align}
\label{h0def}
H_0:=K \otimes \mathbbm 1 + \mathbbm 1 \otimes H_f , \qquad K:= \begin{pmatrix}
e_1 & 0 \\
0 & e_0
\end{pmatrix} ,
\qquad
H_f:=\int \mathrm{d^3}k \, \omega(k) a(k)^* a(k) .
\end{align}
We regard $K$  as an idealized free Hamiltonian of a two-level atom. As already stated in the introduction, its two energy levels are denoted by the real numbers $ 0 = e_0 <e_1$. Furthermore, $H_f$ denotes the free Hamiltonian of a massless scalar field having dispersion relation $\omega(k)=|k|$, and $a,a^*$ are the annihilation and creation operators on  standard Fock space which will be defined below. In the following we will sometimes call $K$ the atomic part, and $H_f$ the free field part of the Hamiltonian.  
The sum of the free two-level atom Hamiltonian $K$ and the free field Hamiltonian $H_f$ will simply be referred to as the ``free Hamiltonian'' $H_0$.
The interaction term reads
\begin{align}
V:= \sigma_1\otimes \left( a(f) + a(f)^*\right) , 
 \qquad \text{where} \qquad \sigma_1= \begin{pmatrix}
0 & 1 \\
1 & 0
\end{pmatrix}, 
\end{align}
 and  the boson form factor  is given by
\begin{align}
f: \R^3 \setminus \{0\}\to \R , \qquad k\mapsto e^{-\frac{k^2}{\Lambda^2}}|k|^{-\frac{1}{2}+\mu} .
\label{eq:f}
\end{align} 
Given the two-point function in \eqref{eq:ccr} the relativistic form factor is $f(k)=(2\pi)^{-\frac{3}{2}}(2|k|)^{-\frac{1}{2}}$ which, however, renders the model ill-defined due to the fact that such an $f$ is not square integrable.
This is referred to as  ultraviolet divergence.  In our case, the Gaussian factor in \eqref{eq:f} acts as an ultraviolet cut-off for $\Lambda>0$ being the ultraviolet cut-off parameter and, in addition, the fixed number 
\begin{align}
\label{const:mu}
\mu \in (0, 1/2)
\end{align}
 implies a regularization of the infrared singularity at $k=0$ which is a technical assumption chosen to keep the proofs of this work more tractable. With a lot of additional work one can also treat the case $\mu=0$  using methods described in \cite{bbkm}.  
 In our notation the missing factor $2^{-\frac{1}{2}}(2\pi)^{-\frac{3}{2}}$ is absorbed in the coupling constant.
The full Spin-Boson Hamiltonian is then defined as
\begin{align}
\label{eq:H}
H:= H_0 + g V
\end{align}
for some  coupling constant $g \in \mathbb{C} $.   The operator $H$ is densely defined on the Hilbert space
\begin{align}
\mathcal H := \mathcal K \otimes \mathcal F\left[ \mathfrak{h}\right] , \qquad \mathcal K:= \mathbb C^2  ,
\end{align}
where 
\begin{align} \label{Fock}
\mathcal F\left[ \mathfrak{h}\right] :=  \bigoplus^\infty_{n=0} \mathcal{F}_n\left[ \mathfrak{h}\right] ,
\qquad
 \mathcal F_n   \left[ \mathfrak{h}\right] := 
\mathfrak{h}^{\odot n},
\qquad
\mathfrak h:= L^2(\mathbb R^3,\C)
\end{align}
denotes the standard bosonic Fock space and superscript $\odot n$ denotes the n-th symmetric tensor product where by convention $\mathfrak{h}^{\odot 0}\equiv\C$.  Throughout this work we will use the notation $\N_0:=\N\cup\{0\}$.  Thanks to the direct sum, an element $\Psi \in \mathcal F[\mathfrak{h}]$
can be represented as a family $(\psi_n)_{n\in\N_0}$ of wave functions $\psi_n \in \mathfrak{h}^{\odot n}$. The state $\Psi$ with $\psi_0=1$ and $\psi_n=0$ for all $n\geq 1$ is called the vacuum and  is denoted by
\begin{align}
\label{Omega}
\Omega:=(1,0,0,\dots)\in \mathcal F\left[ \mathfrak{h}\right] .
\end{align}
For any $h\in \mathfrak{h}$ and $\Psi = (\psi_n)_{n\in\N_0} \in \mathcal F[\mathfrak{h}]$, we define the creation operator 
\begin{align}
\left( a(h)^*\Psi \right)_{n\in\N_0}:=\left( 0,h\odot \psi_0  , \sqrt{2} h\odot \psi_1, \dots \right), \quad  \left( a(h)^*\Psi \right)_{n} = \sqrt{n} h \odot  \psi_{n-1} ,  \quad  \forall n\in\N ,
\end{align}
and the annihilation operator $a(h)$ as the respective adjoint. Occasionally, we shall also use the physics notation 
\begin{align}\label{aastar}
a(h)^*=\int \mathrm{d^3}k \, h(k) a(k)^* , \qquad h\in \mathfrak{h},
\end{align}
 where, formally,  the action of these operators  in the $n$ boson sector of a vector $\Psi = (\psi_n)_{n\in\N_0} \in\mathcal F[\mathfrak{h}]$ can be seen as:
\begin{align}
\left(a(k) \psi  \right)^{(n)}(k_1,...,k_n)&=\sqrt{n+1} \psi^{(n+1)}(k,k_1,...,k_n)  \\
\left(a(k)^* \psi  \right)^{(n)}(k_1,...,k_n)&=\frac{1}{\sqrt{n}}\sum^n_{i=1} \delta^{(3)}(k-k_i) \psi^{(n-1)}(k_1,...,\tilde k_i,...,k_n).
\end{align}
Here, the notation $\tilde \cdot $ means that the corresponding variable is omitted.
Note that $a$ and $a^*$ fulfill the canonical commutation relations:
\begin{align}
\label{eq:ccr}
\left[a(h),a^*(l)   \right]=\left\langle h, l\right\rangle_{\mathfrak{h}} , \qquad \left[a (h),a(l)   \right]=0 , \qquad \left[a^*(h),a^*(l)   \right]=0  \qquad \forall h,l\in\mathfrak{h} .
\end{align}
Let us recall some well-known facts about the introduced model. 
Clearly, $K$ is self-adjoint on $\mathcal K$ and its spectrum 
consists of two eigenvalues $e_0$ and $e_1$. The corresponding eigenvectors are
\begin{align}
\label{varphi}
\varphi_0= \begin{pmatrix}
0 \\ 1
\end{pmatrix} \qquad \text{and} \qquad  \varphi_1=  \begin{pmatrix}
1 \\ 0
\end{pmatrix} \qquad \text{with} \qquad K \varphi_i =e_i \varphi_i , \quad i=0,1.
\end{align}
Moreover, $H_f$ is self-adjoint on its natural domain $\mathcal D(H_f)\subset \mathcal F[\mathfrak{h}]$ and its spectrum  $\sigma (H_f)= [0, \infty )$
 is absolutely continuous (see \cite{reedsimon2}). Consequently, the spectrum of $H_0$ is given by
$\sigma (H_0)=  [e_0, \infty )$, and $e_0,e_1$ are eigenvalues embedded in the absolutely continuous part of the spectrum of $H_0$ (see \cite{reedsimon1}).

Finally, also the closedness  of the full Hamiltonian $H$ is well-known (see e.g.\ \cite{spohnspin}),  however, for the sake of completeness, we give a proof in the Appendix \ref{app:sa}.

\begin{proposition}
\label{thm:Hsa}
The operator $V$ is relatively bounded by $H_0$ with infinitesimal bound and, consequently, $H$ is a closed operator  on the domain $\mathcal D(H) = \mathcal K \otimes \mathcal D(H_f )$.
\end{proposition}

   \begin{remark}\label{R} In this  work we omit spelling out identities whenever unambiguous.     For every vector spaces $V_1$,  $V_2$ and operators $ A_1 $ and $A_2$ defined on $V_1$ and $V_2$, respectively, we identify
\begin{equation}\label{iden}
A_1 \equiv A_1 \otimes \mathbbm 1_{V_2}, \hspace{2cm}  A_2  \equiv \mathbbm 1_{V_1} \otimes A_2 .
\end{equation}
In order to simplify our notation further, and whenever unambiguous, we do not utilize specific  notations for every inner product or norm that we  employ.    
\end{remark}

\subsubsection{Complex dilation}

In the following we introduce the tools necessary  for the complex dilation of Hamiltonians that allows to study resonances as eigenvalues.
\begin{definition}
For every $\theta \in \mathbb R$, we define the unitary transformation
\begin{align}
u_\theta: \mathfrak{h}&\rightarrow \mathfrak{h}
, \qquad \psi(k) \rightarrow e^{-\frac{3\theta}{2}} \psi(e^{-\theta}k) .
\end{align}
Similarly, we define its second quantization 
$U_\theta: \mathcal F [\mathfrak{h}]\rightarrow \mathcal F [\mathfrak{h}]$
by its action on $\mathcal F_n [\mathfrak{h}]$:
\begin{align}
U_\theta \Psi^{(n)}(k_1,...,k_n)=e^{-\frac{3\theta}{2}n}\Psi^{(n)}(e^{-\theta}k_1,...,e^{-\theta}k_n) .
\end{align}
\end{definition}
A straight forward calculation yields, for every $\theta \in \mathbb{R}$, 
\begin{align}
\label{dila}
a^\theta(h): =U_\theta a(h) U^{-1}_\theta= e^{-\frac{3}{2}\theta}a(h_\theta) \qquad \text{where} \qquad h_\theta(k):= h(e^{-\theta}k)
\end{align}
for $h\in \mathfrak{h}$
  and (again, for $\theta \in \mathbb{R}$)
\begin{align} \label{Htheta}
H^\theta :=   U_\theta H U^{-1}_\theta = H_0^{\theta} +  g V^\theta,
\end{align} 
where 
\begin{align}
 H_0^{\theta} : =     K \otimes \mathbbm 1 + \mathbbm 1 \otimes H^\theta_f, \hspace{1cm}
H_f^\theta=  e^{-\theta} H_f , \qquad  V^\theta= \sigma_1 \otimes \left(a(f^{\overline \theta})+ a(f^{ \theta})^*  \right) 
\end{align} 
and 
\begin{align}
 f^\theta: \R^3\to\R , \quad k\mapsto e^{-\theta (1+\mu)} e^{-e^{2\theta}\frac{k^2}{\Lambda^2}}|k|^{-\frac{1}{2}+\mu}. 
\end{align} 
The expressions for $  H_f^{\theta} $ and $  f^{\theta} $ above are well-defined for  $\theta \in \mathbb{C}$. It follows from Appendix \ref{app:sa} that, as long as  
$ f^\theta ,f^\theta/\sqrt{\omega} \in \mathfrak{h}$, $V^{\theta}$ is relatively bounded, with infinitesimal bound, with respect to $H_0^{\theta}$ provided, e.g., $ |\theta| <   \pi /16$. In this case we define $  H^{\theta} $ using the right hand side of Eq.\ \eqref{Htheta}.  Employing similar bounds as in Appendix \ref{app:sa} it turns out that $ H^{\theta} $ is a closed operator and $  \mathcal D (H^{\theta}) = \mathcal D (H_0 )$. Then, it is easy to see that the family    $  \left\lbrace  H^\theta \right\rbrace_{   |\theta |  <  \pi /16 }  $ is an analytic family of type $A$. Notice that the expression in the middle of \eqref{Htheta} does not make sense for non-real $\theta$ because we have not defined $ U_{\theta}  $ for those $\theta$'s (we can define it, but it turns out to be an unbounded operator and therefore the meaning of the middle of \eqref{Htheta} is still unclear).  

From the explicit formula of $ H_{0}^{\theta} $ we deduce that it has only two eigenvalues, namely $e_0$ and $e_1$, and 
\begin{equation}
\sigma ( H_{0}^{\theta} ) = \Big \{ e_i + e^{-\theta} r  \:  :  \:  r \geq 0, i \in \{0, 1 \}  \Big \}. 
\end{equation}      
In this text we use the notation 
\begin{align}
\label{eq:def-disc}
D(x, r):=\left\{ z\in\C : |z-x|<r  \right\}, \qquad  x\in \C , r>0 .
\end{align}

\subsubsection{Infrared cut-offs and definition of the ground-state and the resonance}

In the following, we introduce a family of Hamiltonians $H^{(n),\theta}$  which have two isolated (complex) eigenvalues $\lambda^{(n)}_i$ in small neighborhoods of $e_i$, $ i \in \{ 0, 1 \}$. For every $n$, $H^{(n),\theta}$ enjoys an infared cutoff that is removed as $n$ tends to infinity. 
\begin{definition}
\label{def:Hntheta} We fix a real number $ \boldsymbol{\nu} \in (0,  \pi/16)  $ and for every $\theta \in \mathbb{C}$ we set $ \nu :=\Im \theta   $.  
 We define 
\begin{align}
\label{def:setS}
\mathcal S:=\left\{\theta\in\C: -10^{-3} < \Re \theta <   10^{-3} \quad \text{and} \quad  \boldsymbol{\nu} <  \Im \theta   <\pi/16 \right\} .
\end{align} 
For $\theta\in \mathcal S$ and $n\in \N$, we define:
\begin{enumerate}
\item[(i)]  The sequence of infrared cut-offs $\{\rho_n\}_{n\in\N}$ with $\rho_n:=\rho_0 \rho^n$ for real $ 0 < \rho_0 < \min( 1, e_1/4) $ and $0<\rho<1$. In Definition \ref{sequence} below we specify additional properties of it. 
\item[(ii)]
The cutoff-Hilbert space of one particle, $  \mathfrak{h}^{(n)} $:  
\begin{align}
 \mathfrak{h}^{(n)}:=L^2(\R^3\setminus \mathcal B_{\rho_n}), \quad  \mathcal B_{\rho_n}:=\left\{  x\in\R^3  : |x|<\rho_n \right\} .
\end{align}
The Fock space with one particle sector $  \mathfrak{h}^{(n)} $ is defined as in Eq.\ \eqref{Fock}, and we denote it by $  \mathcal F[\mathfrak{h}^{(n)}]    $. We denote its vacuum state  by $\Omega^{(n)} $.   We set
\begin{align}
\mathcal H^{(n)}:=\mathcal K \otimes \mathcal F[\mathfrak{h}^{(n)}] .
\end{align}
The free boson energy operator with an infrared cutoff  is  defined on   $\mathcal  F[\mathfrak{h}^{(n)}]  $  by Eq.\ \eqref{h0def}, we denote it by $  H_f^{(n), 0} \equiv  H_f^{(n)}$. We set 
\begin{align}\label{Hntheta}
 H_f^{(n), \theta }:= e^{-\theta}   H_f^{(n), 0}. 
\end{align}
For every function $  h \in  \mathfrak{h}^{(n)}   $ we define creation  and  annihilation operators, $a_n(h), \, a_n^*(h) $, on   $  \mathcal F[\mathfrak{h}^{(n)}] $  according to Eq.\ \eqref{aastar}.  We use the same formula for functions   $  h \in  \mathfrak{h}   $, then it is understood that we take the restriction of $h$ to  
$  \R^3\setminus \mathcal B_{\rho_n}  $.

We define the following family of Hamiltonians (densely defined on $  \mathcal H^{(n)}   $ - see Remark \ref{R})
\begin{align} \label{operators111}
 H_0^{(n), \theta} : =     K +   H^{(n), \theta}_f, \hspace{1cm}  V^{(n), \theta} := \sigma_1 \otimes \left(a_n(f^{\overline \theta})+ a_n(f^{ \theta})^*  \right) 
\end{align} 
and 
\begin{align}
 H^{(n), \theta} : =    H_0^{(n), \theta} + g V^{(n), \theta} .
\end{align}
\end{enumerate}
\end{definition}
The Hamiltonians $H^{(n),\theta}$ turn out to have  gaps between the  eigenvalues $\lambda^{(n)}_i$ and the rest of  the spectrum of $H^{(n),\theta}$. This allows us to define Riesz projections, $  P_i^{(n)}  $, corresponding to the eigenvalues $  \lambda_i^{(n)}  $ and use regular perturbation theory for each $n\in\N$. In an inductive scheme, one can obtain explicit estimates on the resolvents and the eigenvalues in each step. Below, we prove that the sequences $ ( \lambda^{(n)}_i )_{n \in \mathbb{N}} $ converge and the interacting ground-state energy $\lambda_0$ and resonance   energy $\lambda_1$ of  $H^\theta$ are the limits
\begin{align}
\lambda_i:= \lim\limits_{n\to\infty}\lambda^{(n)}_i , \qquad i=0,1.
\end{align}
  We define 
\begin{align} \label{pet1tt}
 \mathfrak{h}^{(n,\infty)}:= L^2(\mathcal B_{\rho_n} )  . 
\end{align} 
 We denote  the corresponding Fock  space by $  \mathcal F[\mathfrak{h}^{(n, \infty)}]  $ (it is defined as in \eqref{Fock}), with vacuum state $\Omega^{(n, \infty)}$.  It is straightforward to verify that $ \mathcal H $ is isomorphic to $ \mathcal H^{(n)}\otimes \mathcal F[\mathfrak{h}^{(n,\infty)}] $ and, therefore, we identify
\begin{align} \label{pet2tt}
\mathcal H \equiv \mathcal H^{(n)}\otimes \mathcal F[\mathfrak{h}^{(n,\infty )}]. 
\end{align}   
We prove below that the sequence   $ ( P^{(n)}_i \otimes P_{ \Omega^{(n, \infty)} } )_{n \in \mathbb{N}} $, where $ P_{ \Omega^{(n, \infty)} } $ is the orthogonal projection  onto  the vector space generated by $ \Omega^{(n, \infty)} $,  converges to an eigenprojection corresponding to the eigenvalue $\lambda_i$.

\section{Main results}
\label{sec:mainresult}
Here, we state the main results of our work. All proofs are presented in the next  sections.
 In Proposition \ref{thm:res} below we state the 
 existence of the ground-state eigenvalue and the resonance of $ H^{\theta} $. A similar result, for a more complicated model (Pauli-Fierz), is proved in \cite{bach}. The strategy of proof of Proposition \ref{thm:res} is based  on  the methods introduced in \cite{bach} but it differs from the proof  therein because, here, all our estimates must be independent of $\theta \in \mathcal S$.     
As emphasized earlier, the existence of the resonance and the ground-state is not our  focus but is only provided in order for this work to be self-contained.

 The next proposition is proved in Section \ref{sec:proof-res}. 

\begin{proposition}[Construction of the ground-state and the resonance] 
    \label{thm:res}
   For every $\rho, \rho_0 $ sufficiently small (see Definition \ref{sequence}) there is a constant $ g_0 > 0 $ (that depends on $\rho,  \rho_0$ and $ \boldsymbol{ \nu } $) such that, for every $\theta\in\mathcal S$ (see \eqref{def:setS}) and every $ g \in D(0,   g_0)  $, the (complex) number
\begin{align}
\lambda_i :=\lim\limits_{n\to\infty}\lambda^{(n)}_i , \quad i=0,1
\end{align} 
is an eigenvalue of $H^\theta$ and the range of
\begin{align}
\label{limP}
P_i :=\lim\limits_{n\to\infty}P^{(n)}_i\otimes P_{ \Omega^{(n, \infty)} } , \quad i=0,1
\end{align} 
consists of  eigenvectors corresponding to $\lambda_i$. An explicit  formula for $ g_0 $ is presented in Definition \ref{gzero} below.
\end{proposition}
The  nondegeneracy  of the eigenvalues in Proposition  \ref{thm:res} as well as estimates for the imaginary part of the resonance can be derived from the corresponding results for the Pauli-Fierz model  in  \cite{bach} and \cite{bbp}. Since their proofs do not need  the new features of our multiscale scheme and they are not relevant for our main results, we only state them without proofs and refer to \cite{bach}.

\begin{remark}[Fermi golden rule]
\label{prop:fermi} 

The eigenvalues $\lambda_0$ and $\lambda_1$ are non-degenerate, this follows from Section 6.4.3 in \cite{bach} (we do not repeat the proof here).
The leading order of the imaginary part of the resonance $\lambda_1$ can be explicitly calculated. This is presented in  Theorem 5.6 in \cite{bach} for the Pauli-Fierz model  and, using a different method, in \cite{bfs100}. We do not include a proof here because  it follows,  for the model  at hand, without much change from the proof in \cite{bach}.    

We assume that $ |g| > 0$ is small enough and   define
 \begin{align}
 E_I:=-4\pi^2 (e_1-e_0)^2 |f(e_1-e_0)|^2 .
 \end{align}
  Then,   there is a constant $C_\eqref{eq:impartres1}>0$ and a constant $\epsilon>0$ such that for all  $n\in  \N $ large enough
\begin{align}
\label{eq:impartres1}
\left|  \Im \lambda_1^{(n)} -g^2 E_I \right| \leq g^{2+\epsilon} C_\eqref{eq:impartres1}.
\end{align}
\end{remark}

The next theorems are our main results. We prove analyticity of the resonance and the ground-state,  and the corresponding eigen-projections, with respect to the dilation parameter and  coupling constant.

The next theorem is proved in Section \ref{analyticity} (see Theorem  \ref{thm:anaPp}).
\begin{theorem}[Analyticity with respect to the dilation parameter]
\label{thm:anaP}
  For $\rho, \rho_0 $ sufficiently small and  $g \in D(0,   g_0)$ (see Proposition \ref{thm:res}), the functions 
  \begin{align}
  \mathcal S \ni \theta \mapsto  P_i, \hspace{1cm}  \mathcal S \ni \theta \mapsto  \lambda_i
  \end{align} 
  are  analytic. Moreover, this implies that $ \lambda_i (\theta) \equiv  \lambda_i$ is constant for $\theta\in\mathcal S$ (see \eqref{def:setS}).
\end{theorem}
\begin{remark}
\label{rem:lambda0equal0}
 Our  bounds in the inductive scheme (see Theorem \ref{thm:ind} below) which are used to prove Theorem \ref{thm:anaP} blow up as we take $\nu\to 0$. We study  simultaneously  the cases  $i = 0$ and $i = 1$ and, therefore, our estimations blow up  also for $i=0$. However, it is easy to see from our method that for $ i= 0 $ alone we can take $ \theta $ in a neighborhood of $0$ and prove analyticity  in  this neighborhood. This implies that 
$ \lambda_0  $ is real, because $ H^{\theta} $ is selfadjoint for $\theta = 0$. It is the ground-state energy constructed in \cite{bbkm,spohnspin}. 
\end{remark}
Next theorem is proved in Section \ref{analyticity} (Theorem  \ref{thm:anaPp}).
\begin{theorem}[Analyticity with respect to the coupling constant]
\label{thm:anaPg}
 For every $\rho, \rho_0 $ sufficiently small and   $g \in D(0,   g_0)$,
 the functions  
  \begin{align}
   g \mapsto  P_i,   \hspace{1cm} g \mapsto  \lambda_i
  \end{align} 
  are   analytic.  
\end{theorem}
Our next two theorems provide an estimate for the spectrum of $H^{\theta}$ in neighborhoods of $\lambda_0$ and $\lambda_1$, and resolvent estimates in these neighborhoods. As discussed in the introduction, similar results on spectral estimates can be found in \cite{bfs1,bfs2,bfs3,bfs100} in which the spectrum is located in cuspidal domains using the spectral renormalization method based on the  Feshbach-Schur map. 
Here, we localize the spectrum  in  cones.  For every  $ z\in\C $, we define
\begin{align}
\label{eq:defcone}
\mathcal C_m(z) :=\left\{  z+xe^{-i\alpha} : x\geq 0 , |\alpha-\nu |\leq  \nu/m \right\},
\end{align}   
where we assume that $m \geq 4$, allover this work. 

The next theorem is proved in the proofs of Theorems \ref{resuno} and \ref{resdos} below.   
\begin{theorem}[Resolvent estimates]\label{resolventestimates}
There is a constant $  \bold C   $ (see Definition \ref{C} and \eqref{ggg}) that depends on $ \boldsymbol{\nu} $ but not on $g$ nor in  $\rho $ and $ \rho_0$ such that  for every $ m \geq 4$ and $\rho, \rho_0 $ sufficiently small, there exists $g^{(m)} > 0$  with the following properties: 
for every $ \theta \in \mathcal{S} $ and  $g \in D(0,   g_0)$ (see Proposition \ref{thm:res}) with $|g | \leq g^{(m)}$ , 
\begin{align}
\norm{\frac{1}{ H^{\theta}-z}} \leq 16
  \boldsymbol{   C}^{n+1}  \frac{1}{ {\rm dist} ( z, \mathcal{C}_m(\lambda_i )       )} , 
\end{align} 
for every $ z  \in  B_i^{(1)} \setminus \mathcal C_m \left(\lambda_i - 
2 \rho_n^{1+ \mu/4} e^{-i\nu}\right) $ and
\begin{align}
\norm{\frac{1}{ H^{\theta}-z}} \leq 8
  \boldsymbol{   C}^{n+1}  \frac{1}{ {\rm dist} ( z, \mathcal{C}_m(\lambda_i^{(n)} )       )} , 
\end{align} 
for every $ z  \in  B_i^{(1)} \setminus \mathcal C_m \left(\lambda_i^{(n)}- 
\rho_n^{1+ \mu/4} e^{-i\nu}\right)$. Here, the symbol ${\rm dist }$ denotes the   Euclidean  distance in $\C$  and $B_i^{(1)}$ is defined in \eqref{region:Bi1}. 
\end{theorem}
Explicit  bounds for  $\boldsymbol{C}$, $\rho_0 $ and $\rho$, $g_0$ and $ g^{(m)} $ are given in Definitions \ref{C}, \ref{sequence}, \ref{gzero} and Eq.\ \eqref{ggg}, respectively. We remark that we intentionally do not  provide optimal estimates because  these would    render the proof unnecessary opaque.

The next theorem is proved in the proof of Theorem \ref{resdos} below.
\begin{theorem}[Spectral estimates]\label{spectralestimates} For every $\rho, \rho_0 $ sufficiently small, $ \theta \in \mathcal{S} $ and  $g \in D(0,   g_0)$  with $ |g| \leq  g^{(m)} $, there is a neighborhood  $ 
B_i^{(1)} $ of $ \lambda_i    $ (that depends on $\nu $ but not on $g$) such that  the spectrum of $H^{\theta}$ in $ B_i^{(1)}   $ is contained in $ \mathcal C(\lambda_i)  $ (recall that $\nu$ is the imaginary part of $\theta$). An explicit
formula for $ B_i^{(1)} $ is given in \eqref{region:Bi1}.    
\end{theorem}

\section{Resolvent estimates far away from the spectrum and   detailed analysis of $ H^{(1), \theta} $  } \label{RES}

In this subsection we derive resolvent estimates for  $ H^{(n), \theta} $ and $ H^{ \theta} $  for complex numbers $z$ that are far away from their respective spectra. For the first Hamiltonian,  $ H^{(1), \theta} $,  having an infrared cutoff, we  present resolvent estimates for points that are close to its spectrum. Here, we do not need any restrictions on the sequence $\{\rho_n\}_{n \in \mathbb{N}}$  other than $ 0 < \rho_0 < \min(1, e_1/4) $,  $ 0 < \rho < 1 $.   
In the forthcoming  sections (see Definition \ref{sequence}) we need to assume other properties for the sequence  $\{\rho_n\}_{n \in \mathbb{N}}$. 
We    emphasize that the particular choice of numbers $\rho_n$ does not imply any physical constraint, it only specifies the rate at which the infrared cut-off is removed.

In this section and (in the whole   paper) we denote by $  c > 0 $ any generic (indeterminate) constant (it can change from line to line) that is independent of   
 the parameters $n$, $\theta$, $\rho_0$, $\rho$,  $\theta$,  $\boldsymbol \nu$     and   $g$.  It might  depend on the set $\mathcal S$, as a whole, but not on its elements $ \theta \in \mathcal S$ and  nor either on the parameter   $\boldsymbol \nu$. 
Moreover, by stating that $|g|$ is small enough, we mean that there is a constant such that uniformly for $|g|$ smaller than this constant the referred statement  holds true.  We employ that such a constant does not depend on $\theta$ and $n$ but it depends on the set $\mathcal S$ and on the remaining parameters  ($e_1$, $\rho_0$, $\rho$, $\mu$ and $\Lambda$).

\subsection{Resolvent estimates far away from the spectrum }

We define regions in the complex plane in which we derive resolvent estimates.
\begin{definition}
\label{def:regionsAB}
We set $  \delta: = e_1 - e_0 = e_1 $ and  define the region 
\begin{align}
\label{region:A}
A:&=
A_1\cup A_2\cup A_3  ,
\end{align}
where  
\begin{align}
A_1:&=\left\{  z\in\C : \Re z <e_0-\frac{1}{2}\delta \right\}
\\
A_2:&= \left\{  z\in\C : \Im z >\frac{1}{8}\delta \sin (\nu) \right\}
\\
A_3:&= \left\{  z\in\C : \Re z >e_1+\frac{1}{2}\delta , \Im z \geq -\sin \Big (\frac{\nu}{2}\Big ) \left(\Re (z) -(e_1+\frac{1}{2}\delta)   \right)\right\} ,
\end{align}  
and for $i\in \{0,1\}$,
\begin{align}  
\label{region:Bi1}
B_i^{(1)}:=\left\{  z\in\C : |\Re z-e_i| \leq \frac{1}{2}\delta, -\frac{1}{2}\rho_1   \sin(\nu)\leq \Im z \leq \frac{1}{8}\delta \sin (\nu)  \right\} .
\end{align} 
\begin{figure}[h]
\centering
\includegraphics[width=\textwidth]{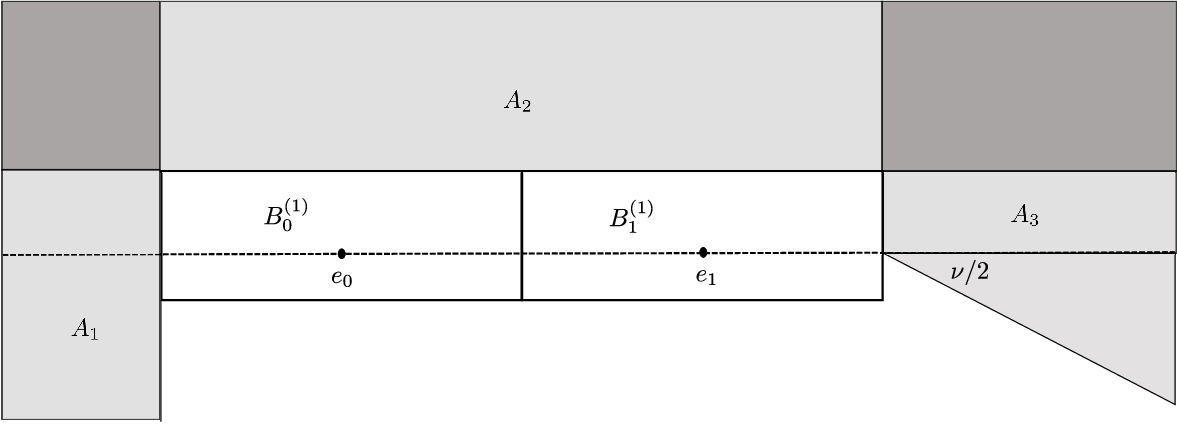}
  \caption{Subsets of the complex plane (see Definition \ref{def:regionsAB})   }
    \label{fig:regionsAB}       
\end{figure}
\end{definition}

In this subsection, we estimate the resolvent of $ H^{(n), \theta} $ and $H^{\theta}$  far away from their spectra, namely in the region $A$  defined in \eqref{region:A}. 
 These  estimates are  applied for the induction basis in  our inductive scheme described in Section \ref{sec:indscheme}.
\begin{lemma}
\label{lemma:resestinA}
 Let $\theta\in\mathcal S$ (see \eqref{def:setS}) and  $n\in\N$.  There is a constant $  C_{\eqref{const:resestinA} }  $ (independent of $\theta $, $n$,  $g$, $\rho_0$, 
  $ \rho $ and $ \boldsymbol \nu  $) such that
for small enough $|g|$ (depending on $  \boldsymbol \nu $), 
 for every  $i\in\{0,1\}$:
 \begin{align}
\label{const:resestinA}
\norm{\frac{1}{H^{(n),\theta}-z}} \leq \frac{  C_{\eqref{const:resestinA} } }{  \sin(\boldsymbol \nu/2 )  } 
 \frac{1}{  |e_i-z|}, \hspace{.5cm}  \norm{\frac{1}{H^{\theta}-z}} \leq
 \frac{  C_{\eqref{const:resestinA} } }{  \sin(\boldsymbol \nu/2 )  }   \frac{1}{|e_i-z|},  \qquad \forall z\in A .
\end{align} 
\end{lemma} 
\begin{proof}
  Let $z\in A$ and $n\in  \N $. Then, arguing as in  Appendix \ref{app:sa}  and using functional calculus, we obtain that 
\begin{align}
\label{eq:prA0}
\norm{V^{(n),\theta}\frac{1}{H_0^{(n),\theta}-z}} 
&\leq \norm{V^{(n),\theta}\frac{1}{(H_0^{(n)}+1)^\frac{1}{2}}} \norm{\frac{H_0^{(n)}+1}{H_0^{(n),\theta}-z}}
\notag \\
&\leq \left( \norm{f^\theta}_2 +2\norm{f^\theta/\sqrt{\omega}}_2  \right) \sup_{y\in [0,\infty),i=0,1}\left| \frac{e_i + y+1}{e_i +e^{-\theta}y-z} \right| 
 .
\end{align} 
Geometrical considerations imply that there is  a   constant $c>0$ such that
\begin{align}
\text{dist}\left(\{e_i+e^{-\theta} y : \,    i=0,1 \},A   \right)\geq \frac{c}{\sin(\boldsymbol \nu/2)}(1+y) \qquad \forall y\geq 0 ,
\end{align}
and hence, there  there is  a  constant $c>0$ such that
\begin{align}
\label{eq:prA000}
\norm{V^{(n),\theta}\frac{1}{H_0^{(n),\theta}-z}} 
&\leq  \frac{c}{\sin(\boldsymbol \nu/2)}, \qquad \forall z\in A
 .
\end{align} 

Then,  we choose $|g|$ small enough such that
\begin{align}
  \norm{ g V^{(n),\theta}\frac{1}{H_0^{(n),\theta}-z}} 
\leq \frac{1}{2} 
\end{align} 
and hence,
\begin{align}\label{nom}
H^{(n),\theta}-z=\left(1+gV^{(n),\theta}\frac{1}{H_0^{(n),\theta}-z} \right) \left( H_0^{(n),\theta}-z \right)
\end{align}
is invertible for all $z\in A$, since $A\cap \sigma (H_0^{(n),\theta})=\emptyset$.   Thanks to the particular geometry, there is a constant $ c > 0 $ such that $|e_{j}+ye^{-\theta}-z|\geq  c \sin   (  \boldsymbol \nu/2   )  |e_i-z|$, for every $z\in A$, every $j \in \{0, 1 \}$, and every positive number $y$. This  and \eqref{nom} imply   
\begin{align}
\label{eq:resestA1}
\norm{\frac{1}{H^{(n),\theta}-z}} \leq 2\norm{\frac{1}{H_0^{(n),\theta}-z}} = \sup_{i=0,1}\sup_{y\geq \rho_n}\frac{2}{|e_i+ye^{-\theta}-z|} \leq  \frac{c}{ \sin(\boldsymbol \nu/2)  |e_i-z|}
\end{align}  
for all $z\in A$, $i=0,1$ and some constant $c>0$.  This completes the proof for the first equation in \eqref{const:resestinA}.  Since the second equation can be shown in a very similar fashion we omit the proof here.
\end{proof}

\subsection{Analysis of $ H^{(1), \theta} $  }
\label{sec:ib}

\begin{lemma}
\label{lemma:ib-reswithoutdisc}
Let $\theta\in\mathcal S$ (see \eqref{def:setS}) and $|g|$  small enough (depending on $\boldsymbol{\nu}$ and $\rho_1$). Then, 
\begin{align}
\norm{\frac{1}{H^{(1),\theta}-z}} \leq 2\norm{\frac{1}{H_0^{(1),\theta}-z}}  \qquad \forall z\in  E^{(1)}_i , i=0,1 ,
\end{align} 
where  
\begin{align}\label{E1i}
 E^{(1)}_i:=B^{(1)}_i\setminus D\left(e_i, \frac{\rho_1 
 \sin (\nu) }{8}\right) . 
\end{align}  
\end{lemma}
\begin{proof}
 Let $z\in E^{(1)}_i$ and $i=0,1$. Then, we have, arguing as in Appendix \ref{app:sa},
\begin{align} \label{pin1}
&\norm{V^{(1),\theta}\frac{1}{H_0^{(1),\theta}-z}} \leq \norm{V^{(1),\theta}\frac{1}{(H_0^{(1)}+1)^\frac{1}{2}}} \norm{\frac{H_0^{(1)}+1}{H_0^{(1),\theta}-z}}
\notag \\
&\leq \left( \norm{f^\theta}_2 +2\norm{f^\theta/\sqrt{\omega}}_2  \right) \sup_{y\in \{0\} \cup [\rho_1,\infty),i=0,1}\left| \frac{e_i+y+1}{e_i+e^{-\theta}y-z} \right|  .
\end{align} 
Take $y\in \{0\} \cup [\rho_1,\infty)$ and $  i \in \{ 0,1  \}$. It follows that
\begin{align}\label{pin2}
\left| \frac{e_i+y+1}{e_i+e^{-\theta}y-z} \right| 
 =  & \left| \frac{e_i+ e^{\theta}( e_i+e^{-\theta}y-z )+1 -  e^{\theta}( e_i -z )    }{e_i+e^{-\theta}y-z}   \right| 
   \\ \notag 
\leq &  |e^{\theta}| + c \frac{1}{  |  e_i+e^{-\theta}y-z   |} \leq\frac{c}{\rho_1 \sin( \nu)},
\end{align}
where the last inequality  is due to the considered geometry.

From \eqref{pin1} and \eqref{pin2} we obtain that there is  a finite constant $c>0$ such that  
\begin{align}
\label{eq:lemma-reswithoutdisc-standtardest1}
\norm{V^{(1),\theta}\frac{1}{H_0^{(1),\theta}-z}}  \leq \frac{c}{\rho_1 \sin( \nu)}  
\leq \frac{c}{\rho_1 \sin(\boldsymbol \nu)}   .
\end{align} 
For $|g|$ small enough (depending on $ \rho_1 $ and $\boldsymbol{\nu}$), we arrive at
\begin{align}
\label{eq:lemma-reswithoutdisc-standtardest}
\norm{gV^{(1),\theta}\frac{1}{H_0^{(1),\theta}-z}} 
&\leq \frac{1}{2} ,
\end{align} 
and hence,
\begin{align}
H^{(1),\theta}-z=\left(1+gV^{(1),\theta}\frac{1}{H_0^{(1),\theta}-z} \right) \left( H_0^{(1),\theta}-z \right)
\end{align}
is invertible for all $z\in E^{(1)}_i$, since $E^{(1)}_i\cap \sigma (H_0^{(1),\theta})=\emptyset$. Then, we obtain
\begin{align}
\norm{\frac{1}{H^{(1),\theta}-z}} \leq 2\norm{\frac{1}{H_0^{(1),\theta}-z}},
\end{align} 
which completes the proof.
\end{proof}
\begin{definition}
\label{def_ib-proj}
We define the projections
\begin{align}
\label{eq:ib-proj}
 P^{(1)}_i: =-\frac{1}{2\pi i}\int_{\hat \gamma^{(1)}_i}\mathrm{d}z \, \frac{1}{H^{(1),\theta}-z}
\end{align}
and
\begin{align}
\label{eq:ib-projat}
 P^{(1)}_{\text{at},i}: =-\frac{1}{2\pi i}\int_{\hat \gamma^{(1)}_i}\mathrm{d}z \, \frac{1}{H_0^{(1),\theta}-z}=P_{e_i}\otimes P_{\Omega^{(1)}}
\end{align}
where 
\begin{align}
\label{eq:ib-gamma}
\hat \gamma^{(1)}_i: [0,2\pi ] \to \C , \quad
t \mapsto e_i +\frac{1}{4}\rho_1  \sin (\nu)  e^{it}, 
\end{align}
 $P_{e_i}$ is the projection onto the eigenspace space corresponding to  $e_i$ of the Hamiltonian $K$ and  $P_{\Omega^{(1)}}$  is the projection onto the vector space generated by the vacuum,  $\Omega^{(1)}  $, of   $\mathcal F[\mathfrak{h}^{(1)}]$.
\end{definition}
\begin{remark}\label{remp}
The right-hand side of Eq.\   \eqref{eq:ib-projat} follows from the following computation:
\begin{align}
&-\frac{1}{2\pi i}\int_{\hat \gamma^{(1)}_i}\mathrm{d}z \, \frac{1}{H_0^{(1),\theta}-z}  
=
-\frac{1}{2\pi i}\int_{\hat \gamma^{(1)}_i}\mathrm{d}z \, \frac{1}{H_0^{(1),\theta}-z} (P_{e_i}\otimes P_{\Omega^{(1)}} +\overline{P_{e_i}\otimes    
P_{\Omega^{(1)}}
})
\notag \\
&=P_{e_i}\otimes P_{\Omega^{(1)}} -\frac{1}{2\pi i}\int_{\hat \gamma^{(1)}_i}\mathrm{d}z \, \frac{1}{H_0^{(1),\theta}-z}\overline{P_{e_i}\otimes   
P_{\Omega^{(1)}}    } ,
\end{align}
 where
\begin{align}
\overline{P_{e_i}\otimes P_{\Omega^{(1)}}   }=\overline{P_{e_i}}\otimes 1+P_{e_i}\otimes\overline{ P_{\Omega^{(1)}}   }
\end{align}
implies that $ -\frac{1}{2\pi i}\int_{\hat \gamma^{(1)}_i}\mathrm{d}z \, \frac{1}{H_0^{(1),\theta}-z}\overline{P_{e_i}\otimes P_{\Omega^{(1)}}     }=0$.
\end{remark}
\begin{lemma}
\label{lemma:ib-proj}
   Let $\theta\in\mathcal S$ (see \eqref{def:setS}) and let $|g|$ be small enough (depending on $ \boldsymbol{\nu} $ and $ \rho_1 $). Take $i\in\{ 0,1 \}$. Then,  there is a constant $C_{\eqref{eq:pconst}}>0$ 
   (independent of $\theta $, $n$,  $g$, $\rho_0$, 
  $ \rho $ and $ \boldsymbol \nu  $)  such that
\begin{align}
\label{eq:pconst}
\norm{ P_i^{(1)}-P^{(1)}_{\text{at},i}} \leq | g|  \frac{  C_{\eqref{eq:pconst}} }{\rho_1 \sin(\boldsymbol \nu)}  <1 \qquad \text{and} \qquad \norm{ P_i^{(1)}}\leq 1+  |g| \frac{  C_{\eqref{eq:pconst}} }{\rho_1 \sin(\boldsymbol \nu)}  <2 , 
\end{align} 
where $\hat \gamma^{(1)}_i$, $ P_i^{(1)}$ and $P^{(1)}_{\text{at},i}$ are introduced  in Definition \ref{def_ib-proj}. 
\end{lemma}
\begin{proof}
 First, we observe that
\begin{align}
\norm{ P_i^{(1)}-P^{(1)}_{\text{at},i}} 
&\leq \frac{1}{2\pi }\norm{ \int_{\hat \gamma^{(1)}_i}\mathrm{d}z \, \left( \frac{1}{H^{(1),\theta}-z} -\frac{1}{H_0^{(1),\theta}-z} \right)} .
\end{align} 
Note that $\hat \gamma^{(1)}_i\subset E^{(1)}_i$ (see \eqref{E1i}).   Eq.\  \eqref{eq:lemma-reswithoutdisc-standtardest1} implies that  there is a finite constant $c>0$ such that for every $ z$ in the (image of the) curve $   \hat \gamma^{(1)}_i $
\begin{align}
\norm{gV^{(1),\theta}\frac{1}{H_0^{(1),\theta}-z}} 
&< |g| \frac{c}{\rho_1 \sin(\boldsymbol \nu)}   \leq \frac{1}{2} ,
\end{align} 
for $|g|$ small enough (depending on $  \boldsymbol \nu $ and $\rho_1$). Next, we obtain 
\begin{align}
 & \norm{ P_i^{(1)}-P^{(1)}_{\text{at},i}} 
\leq \frac{1}{2\pi }\norm{ \int_{\hat \gamma^{(1)}_i}\mathrm{d}z \,\frac{1}{H_0^{(1),\theta}-z} \sum^\infty_{l=1} \left( -gV^{(1),\theta}\frac{1}{H_0^{(1),\theta}-z} \right)^l }  
\notag \\
&\leq    \frac{\rho_1 \sin (\nu)}{4} \sup_{  |z  - e_i | =   \frac{\rho_1 \sin (\nu)}{4}   }  
\Big (   \Big \|  \frac{1}{H_0^{(1),\theta}-z} \Big \|   \Big )
  | g| \frac{c}{\rho_1 \sin(\boldsymbol \nu)}
 \sum^\infty_{l=0}  \Big(\frac{1}{2}\Big ) ^l 
< |g| \frac{c}{\rho_1 \sin(\boldsymbol \nu)} . 
\end{align}  
This proves the first part of the lemma. Furthermore, it follows from \eqref{eq:ib-projat} that $\norm{P^{(1)}_{\text{at},i}}=1$,
and hence,
\begin{align}
\norm{ P_i^{(1)}} \leq \norm{P^{(1)}_{\text{at},i}} +\norm{\hat P_i^{(1)}-P^{(1)}_{\text{at},i}} \leq 1+  |g| \frac{c}{\rho_1 \sin(\boldsymbol \nu)} <2,
\end{align}
for sufficiently small $|g|$.
This proves proves the second part of the lemma.
\end{proof}
\begin{remark} \label{missing} 
  Let $\theta\in\mathcal S$ (see \eqref{def:setS}) and $i\in\{ 0,1 \}$.
 Suppose that $|g|$ is small enough (depending on $ \boldsymbol{\nu} $ and $ \rho_1 $). Then, it follows from Lemma   \ref{lemma:ib-proj} together with the fact that $ P^{(1)}_{\text{at},i}$ is a rank-one projection that also $P_i^{(1)}$ is a rank-one projection.  Lemma  \ref{lemma:ib-reswithoutdisc}  implies that  $  H^{(1), \theta} $ has no spectral points in $E_i^{(1)} =  B^{(1)}_i\setminus D\left(e_i, \frac{\rho_1 
 \sin (\nu) }{8}\right) $. Since the contour of integration for the projection 
  $P_i^{(1)}$  is contained in $ B^{(1)}_i $ and its interior contains  $  D\left(e_i, \frac{\rho_1 
 \sin (\nu) }{8}\right) $, we conclude that  
 there is a unique  spectral point  
 $ \lambda^{(1)}_i $ of $  H^{(1), \theta} $ in $B_i^{(1)}$, it  is a simple eigenvalue  and it is contained in $  D\left(e_i, \frac{\rho_1 
 \sin (\nu) }{8}\right) $.
\end{remark}
Lemma \ref{lemma:ib-reswithoutdisc} together with Lemma \ref{lemma:ib-proj}   yield a resolvent estimate in the whole region $B_i^{(1)}\setminus \{\lambda_i^{(1)}\}$ by making use of the maximum modulus principle of complex analysis. 
\begin{lemma}
\label{lemma:ib-res}
 Let $\theta\in\mathcal S$ (see \eqref{def:setS}) 
  let $|g|$ be small enough (depending on $ \boldsymbol{\nu} $ and $ \rho_1 $). Take $i\in\{ 0,1 \}$. Then,  there is a constant $C_{\eqref{eq.ib-res}}>0$
   (independent of $\theta $, $n$,  $g$, $\rho_0$, 
  $ \rho $ and $ \boldsymbol \nu  $)  
  such that 
\begin{align}
\label{eq.ib-res}
\norm{\frac{1}{H^{(1),\theta}-z}\overline{ P_i^{(1)}}} 
\leq   \frac{  C_{\eqref{eq.ib-res}} }{\rho_1 \sin ( \nu) }
  \leq   \frac{  C_{\eqref{eq.ib-res}} }{\rho_1 \sin (\boldsymbol \nu) }, 
\qquad \forall z\in  B_i^{(1)}.
\end{align} 
\end{lemma}
\begin{proof}
Note that the function
\begin{align}
\overline{D\left(e_i, \frac{1}{8}\rho_1 \sin (\nu) \right)} \ni z \mapsto G_{\phi,\psi}(z):=\left\langle \phi, \frac{1}{H^{(1),\theta}-z}\overline{ P_i^{(1)}} \psi  \right\rangle 
\end{align}
is continuous, and furthermore,  analytic on $D\left(e_i, \frac{1}{8}\rho_1  \sin (\nu) \right)$, for all  $\phi, \psi \in  \mathcal H$. Then, it follows from the maximum modulus principle that this function attains its maximum on the boundary of its domain. This together with the Cauchy-Schwarz inequality and Lemma \ref{lemma:ib-reswithoutdisc} and \ref{lemma:ib-proj}  implies that  there is a finite constant $c>0$ such that
\begin{align}
\label{eq:lemma-ib-res-maxmod}
\left| G_{\phi,\psi}(z) \right|\leq \frac{c}{\rho_1 \sin ( \nu) }  \| \phi  \|  \psi \| , \qquad \forall z\in\overline{D\left(e_i, \frac{1}{8}\rho_1 \sin (\nu) \right)}.
\end{align}
Consequently,  there is a finite constant $c>0$ such that
\begin{align}
\label{eq:lemma-ib-res-est1}
\norm{\frac{1}{H^{(1),\theta}-z}\overline{ P_i^{(1)}}} 
&\leq \frac{c}{\rho_1 \sin ( \nu) }, \qquad \forall z\in\overline{D\left(e_i, \frac{1}{8}\rho_1 \sin (\nu) \right)} ,
\end{align} 
and moreover, Lemma \ref{lemma:ib-reswithoutdisc} (again)  guarantees that there is a finite constant $c>0$ such that 
\begin{align}
\label{eq:lemma-ib-res-est2}
\norm{\frac{1}{H^{(1),\theta}-z}} &\leq 2\norm{\frac{1}{H_0^{(1),\theta}-z}}
\leq   \frac{c}{\rho_1 \sin ( \nu) },  
\end{align} 
for all $z\in  B^{(1)}_i \setminus D\left(e_i, \frac{1}{8}\rho_1 \sin (\nu) \right) $.
This together with \eqref{eq:lemma-ib-res-est1} completes the proof.
\end{proof}

Applying Lemma \ref{lemma:ib-res} to our particular geometry allows to formulate the following corollary.
\begin{corollary}
We define  $  q^{(1)}_i :=  \lambda^{(1)}_i + 
\frac{1}{4} \rho_1 e^{-i \nu}     $   and recall Eq.\ \eqref{eq:defcone}. 
\label{coro:ib-res}
 Let $\theta\in\mathcal S$ ( see \eqref{def:setS}) 
  let $|g|$ be small enough (depending on $ \boldsymbol{\nu} $ and $ \rho_1 $). Take $i\in\{ 0,1 \}$. Then,  there is a constant $C_{\eqref{coroeq.ib-res}}>0$
   (independent of $\theta $, $n$,  $g$, $\rho_0$, 
  $ \rho $ and $ \boldsymbol \nu  $)  
  such that 
\begin{align}
\label{coroeq.ib-res}
\norm{\frac{1}{H^{(1),\theta}-z}\overline{ P_i^{(1)}}} 
\leq   \frac{1}{ \sin(\nu/m) }\frac{  C_{\eqref{eq.ib-res}} }{{ \rm dist }( z, \mathcal{C}_m(q_1^{(1)}) )  }, \qquad \forall z \in  B^{(1)}_i \setminus \mathcal{C}_m(q_i^{(1)}) .
\end{align} 
\end{corollary}
\begin{lemma}
\label{lemma:ib-dist-energy}
 Let $\theta\in\mathcal S$ (see \eqref{def:setS}) and let $|g|$ be small enough (depending on $ \boldsymbol{\nu} $ and $ \rho_1 $). Then, there is a constant $C_{\eqref{constgggg}}>0$     (independent of $\theta $, $n$,  $g$, $\rho_0$, 
  $ \rho $ and $ \boldsymbol \nu  $)    such that for every $i \in \{0, 1 \}$
\begin{align}
\label{constgggg}
\left| \lambda_i^{(1)} -e_i \right| \leq   |g| C_{\eqref{constgggg}} .
\end{align} 
\end{lemma}
\begin{proof}
 It follows from Lemma \ref{lemma:ib-proj} that $\Big |\left\langle  \varphi_i \otimes \Omega,  P^{(1)}_i \varphi_i \otimes \Omega \right\rangle \Big |>\frac{1}{2}$ for $|g|$ small enough   (depending on $ \boldsymbol{\nu} $ and $ \rho_1 $).
  We calculate
\begin{align}
\lambda_i^{(1)}
&= \frac{\left\langle  \varphi_i \otimes  \Omega^{(1)}, H^{(1),\theta}  P^{(1)}_i \varphi_i \otimes  \Omega^{(1)} \right\rangle}{\left\langle  \varphi_i \otimes  \Omega^{(1)}  ,  P^{(1)}_i \varphi_i \otimes   \Omega^{(1)}  \right\rangle}
=e_i+ g\frac{\left\langle  V^{(1),\overline{\theta}}  \varphi_i \otimes  \Omega^{(1)}  ,    P^{(1)}_i \varphi_i \otimes   \Omega^{(1)}  \right\rangle}{\left\langle  \varphi_i \otimes   \Omega^{(1)} ,  P^{(1)}_i \varphi_i \otimes   \Omega^{(1)}  \right\rangle}.
\end{align} 
 Let now $z\in\C$ such that $|e_i-z|=\frac{1}{4}\rho_1 \sin (\nu) $.   Eq.\ \eqref{eq:pconst}   (which requires  $|g|$  to be small enough    -depending on $ \boldsymbol{\nu} $ and $ \rho_1 $) then allows to obtain 
\begin{align}
&\left| \lambda_i^{(1)} -e_i \right| \leq 4  \norm{gV^{(1),\overline{\theta}}  \varphi_i \otimes  \Omega^{(1)}  }\leq 
4  \left| e_i-z  \right| \norm {gV^{(1),\overline{\theta}}\frac{1}{H^{(1),\overline \theta}_0-z}} 
\leq |g| c,
\end{align}
 for some constant $c$  (independent of $\theta $, $n$,  $g$, $\rho_0$, 
  $ \rho $ and $ \boldsymbol \nu  $) .
Here, we have used \eqref{eq:lemma-reswithoutdisc-standtardest1} from Lemma \ref{lemma:ib-reswithoutdisc} in the last step. Notice that in this  work we assume that the imaginary part of $\theta$,  $ \nu $, is positive. Then, strictly speaking, we do not have the right to use our results for $\overline \theta$.  
 However, the restriction we impose by assuming that $ \nu  $ is  not negative is irrelevant. This is assumed only for convenience in order to simplify our notation. Of course, the same results hold true if we take $ - \pi/16  < \nu < - \boldsymbol{\nu}  $. 
In this case the spectrum of $H_0^{(1), \theta}$ is just mirrored  with respect to  the real line. Hence, one has to mirror also the definition of $B_i^{(1)}$  with respect to  the real line in order to obtain the same estimates. Afterwards, the proofs are just the same. 
\end{proof}
\begin{remark} 
\label{lemma:ib-proj2}
 Let $\theta\in\mathcal S$ (see \eqref{def:setS}) and let $|g|$ be small enough (depending on $ \boldsymbol{\nu} $ and $ \rho_1 $). Then 
\begin{align}
P_i^{(1)}=-\frac{1}{2\pi i}\int_{ \gamma^{(1)}_i}\mathrm{d}z \,  \frac{1}{H^{(1),\theta}-z},
\end{align} 
 where $\gamma^{(1)}_i: [0,2\pi ] \to \C , \quad
t \mapsto  \lambda^{(1)}_i  +\frac{1}{4}\rho_1 \sin (\nu)  e^{it}$. This follows from Remark \ref{missing}, because, for small enough $|g|$,  $\gamma^{(1)}_i ,\hat \gamma^{(1)}_i \subset B_i^{(1)}\setminus \{\lambda_i^{(1)}\} $, see Lemma 
\ref{lemma:ib-dist-energy}. 
\end{remark}
\begin{lemma}
\label{lemma:ib-resprima}
 Let $\theta\in\mathcal S$ (see \eqref{def:setS}) 
 and let $|g|$ be small enough (depending on $ \boldsymbol{\nu} $ and $ \rho_1 $). Take $i\in\{ 0,1 \}$. Then,  there is a constant $C_{\eqref{eq.ib-resprima}}>0$
   (independent of $\theta $, $n$,  $g$, $\rho_0$, 
  $ \rho $ and $ \boldsymbol \nu  $)  
  such that
\begin{align}
\label{eq.ib-resprima}
\norm{\frac{1}{H^{(1),\theta}-z}\overline{ P_i^{(1)}}} 
\leq  
\frac{C_{\eqref{eq.ib-resprima}}}{\sin (\boldsymbol \nu)} \frac{1}{\rho_1+\left| \lambda^{(1)}_i-z \right|}, 
\qquad \forall z\in  B_i^{(1)}.
\end{align} 
\end{lemma}
\begin{proof}
We use Lemma \ref{lemma:ib-res}  and calculate, for $  | z - e_i  | \leq  \rho_1 $, 
\begin{align}
 \frac{ 1 }{\rho_1 \sin ( \nu) }  =  \frac{\rho_1+|\lambda_i^{(1)}-z|}{  \rho_1 \sin ( \nu) } \frac{1}{\rho_1+|\lambda_i^{(1)}-z|}  \leq    
 \frac{ c }{  \sin (\boldsymbol \nu)  } \frac{1}{\rho_1+|\lambda_i^{(1)}-z|},
\end{align}
where we use Lemma \ref{lemma:ib-dist-energy} and choose $|g|$ small enough. For $  | z - e_i  | >  \rho_1 $ we use Lemma
\ref{lemma:ib-reswithoutdisc}. The spectral theorem and the explicit 
form of the spectrum of non-interacting Hamiltonian $ H_0^{(1),\theta} $ allow us to estimate the norm of its resolvent. Then, similar estimates as above imply the desired result. 
\end{proof}
\begin{lemma}
\label{lemma:resestinAprima}
 Let $\theta\in\mathcal S$ (see \eqref{def:setS}) and      $n\in\N$.  There is a constant $  C_{\eqref{const:resestinAprima} }  $ (independent of $\theta $, $n$, 
  $g$, $\rho_0$, 
  $ \rho $ and $ \boldsymbol \nu  $) such that
for small enough $|g|$ (depending on $  \boldsymbol \nu $), 
 for every  $i\in\{0,1\}$:
\begin{align}
\label{const:resestinAprima}
\norm{\frac{1}{H^{(n),\theta}-z}} \leq \frac{C_{\eqref{const:resestinAprima}}}{\sin (\boldsymbol \nu/2 )} \frac{1}{\rho_l+\left| \widetilde{\lambda}_i-z \right|}, \hspace{.5cm}  \norm{\frac{1}{H^{\theta}-z}} \leq\frac{C_{\eqref{const:resestinAprima}}}{\sin (\boldsymbol \nu/2)} \frac{1}{ \rho_l+\left| \widetilde{\lambda}_i-z \right|} , 
\end{align} 
for every $  z\in A \cup  \Big ( B^{(1)}_1  - [0, \infty) e^{-i \nu}  \Big) \setminus B^{(1)}_1 ,  $ every $l \in \mathbb{N}$ and every $  \widetilde{\lambda}_i \in D( \lambda^{(1)}_i, 3  g ) .$

Moreover, 
\begin{align}
\label{coroeq.ib-resT}
\norm{\frac{1}{H^{(n),\theta}-z}\overline{ P_i^{(1)}}} 
\leq   \frac{1}{ \sin(\nu/m) }\frac{  C_{\eqref{coroeq.ib-resT}} }{{ \rm dist }( z, \mathcal{C}_m(q_1^{(n)}) )  },
\end{align} 
for every $  z\in \Big ( B^{(1)}_1  - [0, \infty) e^{-i \nu}  \Big) \setminus B^{(1)}_1$. 
\end{lemma} 
\begin{proof} We take $z \in A$. Lemma \ref{lemma:ib-dist-energy} implies that 
$  | \widetilde{\lambda}_i - e_i | \leq  |g| ( C_{\eqref{constgggg}} + 3 )  $. 
Notice that 
\begin{align}
\frac{1}{ |  e_i - z    | }    \leq    \frac{ \rho_m+ | \widetilde{\lambda}_i-z | }{ |  e_i - z    | } 
  \frac{1}{\rho_m+ |  \widetilde{\lambda}_i-z |} \leq &  
    \frac{ \rho_m+ |  \widetilde{\lambda}_i - e_i| + | e_i  -z | }{ |  e_i - z    | } 
  \frac{1}{\rho_m+\left|  \widetilde{\lambda}_i-z \right|} \\ \notag 
  \leq & c \frac{1}{\rho_m+\left|  \widetilde{\lambda}_i-z \right|}, 
\end{align}
since $   |  e_i - z    |  $ is bounded from below uniformly for  $z  \in A$. Then, the result follows from Lemma \ref{lemma:resestinA}. The result for 
$ z \in \Big ( B^{(1)}_1  - [0, \infty) e^{-i \nu}  \Big) $  can be found similarly which is why we omit the proof. The proof of \eqref{coroeq.ib-resT} follows  from a similar argument as in Corollary  \ref{coro:ib-res},and therefore, it is also omitted. 
\end{proof}
\begin{definition} \label{compile}
In this subsection (Section \ref{RES}) we assumed a finite number of times that $ |g|  $ is small enough (depending on $\boldsymbol{\nu}$ and $\rho_1$). We set $ \boldsymbol{g} > 0 $  such that  for every  $  |g| \leq  \boldsymbol{g}$ all results of this section  hold true. Similarly, in the statements of our results we use a finite number of times constants (that are independent of $\theta $, $n$,  $g$, $\rho_0$, 
  $ \rho $ and $ \boldsymbol \nu  $) in order to estimate from above  norms of various entities. We denote by $ \boldsymbol{c} \geq 1 $ a, fixed, upper bound of the set of all  these constants.  We additionally take $\boldsymbol{g}$ small enough such that
$$
\norm{ P_i^{(1)}-P^{(1)}_{\text{at},i}} < 10^{-3}, 
$$
see Lemma \ref{lemma:ib-proj}. 
\end{definition}

\section{Resonance and ground-state multiscale analysis}
\label{sec:4}

\subsection{Notation: the sequence $(\rho_{n})_{n \in \mathbb{N} \cup \{ 0  \}}$ and the coupling constant $g$ }\label{notation}

Next, we introduce a constant, $ \boldsymbol{D} $, that includes all constants involved in estimations for our multiscale construction. This constant does not depend on $\theta \in \mathcal{S} $, $g$, $\boldsymbol{\nu} $, $n$, $\rho$ and $\rho_0$. Our bounds do depend on  $\boldsymbol{\nu}$. They blow up as $ \boldsymbol{\nu}  $ tends to zero with a rate that is not worse than  $\sin( \boldsymbol{\nu}/2 )^{-3} $. 
The constant $  \boldsymbol{D} $ does not depend on $ \boldsymbol{\nu} $. The dependence on  $ \boldsymbol{\nu}  $ in our bounds is reflected in a function 
$\boldsymbol{C} \equiv \boldsymbol{C}( \boldsymbol{\nu} )$ that is bounded from below by the constant $\boldsymbol{D}$ multiplied by the  factor  $\sin( \boldsymbol{\nu}/2 )^{-3} $.  As already explained, the constant $  \boldsymbol{D}$ and the blow-up rate
$\sin( \boldsymbol{\nu}/2 )^{-3} $ are intentionally not optimal but often estimated from above to increase the readability of the proofs.

\begin{definition}[The function  $\boldsymbol{C} \equiv \boldsymbol{C}( \boldsymbol{\nu} ) $]
\label{C}
We fix a  constant, $\boldsymbol{D}$,  that does not depend on  $g\in D(0,\boldsymbol g)$,  $\theta \in \mathcal{S} $,  and  $ \boldsymbol{\nu} $. The only property that it must satisfy is the following (see Definition \ref{compile})
\begin{align}
\boldsymbol{D} \geq 10^6 + 10 \boldsymbol{c}. 
\end{align} 
Next, we fix  a function $\boldsymbol{C} \equiv \boldsymbol{C}( \boldsymbol{\nu} ) $ 
satisfying 
\begin{align}\label{dorm1}
\boldsymbol{C} \equiv \boldsymbol{C}( \boldsymbol{\nu} ) \geq  \boldsymbol{D} \sin( \boldsymbol{\nu}/2 )^{-3} . 
\end{align} 
\end{definition}
The sequence $(\rho_{n})_{n \in  \N_0}$ that we introduce above is defined in the following manner. 

\begin{definition}[  Parameters $\rho_0$ and $\rho$] \label{sequence}
The parameters $\rho_0$ and $\rho$ in the definition of the sequence $\rho_n=\rho_0\rho^n , n\in\N_0$, see Definition \ref{def:Hntheta}(i), have to fulfill the following constraints:  
\begin{align}\label{dorm2}
\boldsymbol{C}^8 \rho_0^{\mu} \leq 1, \qquad \boldsymbol{C}^4 \rho^{\mu} \leq 1/4.   
\end{align}
\end{definition}
 We recall that in this work we require  $ |g|$ to be small enough.   The next definition  summarizes all requirements that it must satisfy. 

\begin{definition}[The coupling constant $g$] \label{gzero} We set a constant $g_0 \leq \boldsymbol{g}$ satisfying the following (see Definition \ref{compile}):
\begin{align}\label{dorm3}
g_0 \leq \frac{\rho_1\sin(  \boldsymbol{\nu} /2)^2}{10^4 \boldsymbol{c}}. 
\end{align}
 Henceforth, we always  require  $|g| \leq g_0 $.  
\end{definition}  
\begin{remark} \label{recuerda}
The selection of $\boldsymbol{C} $, the sequence $(\rho_{n})_{n \in \N_0}$ and the constant $g_0$   will later allow to set up the infrared induction scheme, and is therefore, rather involved. This remark is intended to help the reader to understand this procedure.   Below (in this section), we   use boldface fonts   whenever we use the properties of $\boldsymbol{C}$ and $g$ (or $g_0$) that we specified above.

 The requirements for $(\rho_{n})_{n \in \N_0}$ are only present in order to satisfy  the last inequalities in Eqs.\ \eqref{eq:P1} and \eqref{eq:P3} below. Then, it  will turn out that it is only necessary to assume that $ \boldsymbol{C}^4 \rho_0^{\mu} \leq 1 $ and 
 $\boldsymbol{C}^2 \rho^{\mu} \leq 1/2$ in order to close our induction. We assume stronger conditions  again, for notational convenience, and because it implies a faster convergence rate  in \eqref{eq:P3}, which will be used in a forthcoming paper (see Remark \ref{scattering}).   
\end{remark}

\subsection{Induction scheme and the strategy of our multiscale construction}
\label{sec:indscheme}

   We denote by  
\begin{align} \label{pet1p}
 \mathfrak{h}^{(n,n+1)}:= L^2(\mathcal B_{\rho_n}\setminus \mathcal B_{\rho_{n+1}})  
\end{align} 
the Hilbert space  of one-particle (boson) states with energies  in the interval $ [ \rho_{n+1}, \rho_n ). $ We denote  the corresponding Fock  space by $  \mathcal F[\mathfrak{h}^{(n,n+1)}]  $ (it is defined as in \eqref{Fock}).   Note that $ \mathcal H^{(n+1)} $ is isomorphic to $ \mathcal H^{(n)}\otimes \mathcal F[\mathfrak{h}^{(n,n+1)}] $, and, therefore, we identify
\begin{align} \label{pet2p}
\mathcal H^{(n+1)}\equiv \mathcal H^{(n)}\otimes \mathcal F[\mathfrak{h}^{(n,n+1)}]. 
\end{align}   
For $i=0,1$, we inductively (and simultaneously) construct  sequences $\{\lambda^{(n)}_i\}_{n\in\N_0}$ of complex numbers,  sequences $\{B^{(n)}_i\}_{n\in\N}$  of subsets of the complex plane and  sequences $\{P^{(n)}_i\}_{n\in\N_0}$ of operators that satisfy the properties listed below.   
\begin{enumerate}
\item[($\mathcal P 1$)]  We set $\lambda^{(0)}_i\equiv \lambda^{(1)}_i $. For $n\in\N$, $\lambda^{(n)}_i$ is a simple eigenvalue of $H^{(n),\theta}$ and 
\begin{align}
\label{eq:P1}
\left|  \lambda^{(n)}_i-\lambda^{(n-1)}_i  \right|< |g| \bold C^{n+1}(\rho_{n-1})^{1+\mu}\leq |g|  \left(\frac{1}{2}\right)^{n-1}\rho_{n-1} .
\end{align}
The second inequality follows from Definition \ref{sequence}.
\item[($\mathcal P 2$)] For $n\in\N$, we  define (recall that $ \nu = \Im \theta  $) 
\begin{align}
\label{eq:P21}
B_i^{(n)}:= B^{(1)}_i\setminus \left\{  z\in\C : \Im z <\Im \lambda^{(n)}_i - \frac{1}{4} \rho_n \sin (\nu)  \right\} .
\end{align}
$\lambda^{(n)}_i$ is the only point in the spectrum of $H^{(n),\theta}$ intersected with $B_i^{(n)}$.
\item[($\mathcal P 3$)]  We set $P^{(0)}_i\equiv P^{(1)}_i $. For $n\in\N$, we define
\begin{align}
\label{eq:projin}
P^{(n)}_i=-\frac{1}{2\pi i}\int_{\gamma^{(n)}_i}\mathrm{d}z \, \frac{1}{H^{(n),\theta}-z},
\end{align}
where 
\begin{align}
\label{eq:gammain}
\gamma^{(n)}_i: [0,2\pi ] \to \C , \quad
t \mapsto \lambda^{(n)}_i +\frac{1}{4}\rho_n \sin (\nu)  e^{it}.
\end{align}
 The projections $P^{(n)}_i$ satisfy
\begin{align}
\label{eq:P3}
\norm{P^{(n)}_i - P^{(n-1)}_i\otimes P_{\Omega^{(n-1,n)}}} \leq \frac{|g|}{\rho}\bold C^{2n+2} \rho^\mu_{n-1}\leq \frac{|g|}{\rho} \left( \frac{1}{2}
\right)^{n-1}, 
\end{align}
where $P_{\Omega^{(n-1,n)}}$ is the projection onto the vacuum vector $\Omega^{(n-1,n)}\in \mathcal F[\mathfrak{h}^{(n-1,n)}]$ (see \eqref{pet1p}-\eqref{pet2p}).  In  \eqref{eq:P3} we omit the tensor product for $n=1$. The second inequality follows from Definition \ref{sequence}.

\item[($\mathcal P 4$)]  Set $n\in\N$.  For any $z\in B_i^{(n)}$, we have that 
\begin{align}
\label{eq:P5}
\norm{\frac{1}{H^{(n),\theta}-z}\overline{P^{(n)}_i}}\leq \frac{\bold C^{n+1}}{\rho_n+\left| z-\lambda^{(n)}_i \right|} ,
\end{align}
where $\overline{P^{(n)}_i}=\mathbbm 1_{\mathcal H^{(n)}}-P^{(n)}_i$.
\end{enumerate}
\begin{theorem}[Multiscale analysis for resonances and ground state eigenvalues]
    \label{thm:ind}
  For every $ i \in \{ 0, 1  \} $ and $ \theta \in \mathcal{S} $ (see  \eqref{def:setS}),  there exist sequences of complex numbers $\{\lambda^{(n)}_i\}_{n\in\N_0}$, subsets of the complex plane $\{B^{(n)}_i\}_{n\in\N}$ and projection operators $\{P^{(n)}_i\}_{n\in\N_0}$  satisfying Properties ($\mathcal P 1$)-($\mathcal P 4$). Recall that we assume that $ |g | \leq g_0 $.  
\end{theorem}
The proof for this theorem  is given in Section \ref{sec:proof-induction}.

Similar results, for the  Pauli-Fierz model, are derived in \cite{bbp}.
 In the present paper we need uniform estimates with respect to $\theta \in \mathcal S$ and $ g \in D(0, g_0) $, in order to obtain uniform convergence with respect to these parameters (which is an important  ingredient for the proof of analyticity). This is not the case in \cite{bbp} where analyticity is not an issue at stake.

 \begin{remark}
\label{rem.ib}
Note that $(\mathcal P1)$ and $(\mathcal P3)$ hold true for $n=1$, by definition.  Remark \ref{missing} implies  that $(\mathcal P2)$ holds true for $n=1$. Moreover, $(\mathcal P4)$, for $n=1$,   follows from Lemma  \ref{lemma:ib-resprima} (Recall  Definitions \ref{compile}, \ref{C} and \ref{gzero}).
\end{remark}

\subsection{Proof of Theorem~\ref{thm:ind}}
\label{sec:proof-induction}

We recall that in the remainder of this work  we always assume that $| g| \leq
g_0  $ (see Definition \ref{gzero}) and $ \theta \in \mathcal{S} $ (see  \eqref{def:setS}).  
In Section \ref{sec:indkey}, we prove some key ingredients which are then used in Section \ref{sec:is} in order to conclude the induction step.

\subsubsection{Key estimates for the induction step}
\label{sec:indkey}
In this section, we assume that ($\mathcal P 1$)-($\mathcal P 4$) hold true for all $m\leq n \in\N$ and we derive some key estimates which we apply in the next section in order to show the induction step in the proof of Theorem \ref{thm:ind}.

 By Eq.\ \eqref{h0def}, we define free boson energy operator  restricted to  $\mathcal F[\mathfrak{h}^{(n, n+1)}]  $    and denote it by $  H_f^{(n, n+1), 0} \equiv   H_f^{(n, n+1)}$ (see  \eqref{pet1p}-\eqref{pet2p}). We set 
\begin{align}\label{Hnthetaprima}
 H_f^{(n, n+1), \theta }:= e^{-\theta}   H_f^{(n, n+1), 0}. 
\end{align}
For every function $  h \in  \mathfrak{h}^{(n, n+1)}$, we denote the creation and annihilation operators, $a_{n, n+1}(h), \, a_{n, n+1}^*(h) $, on  $ \mathcal F[\mathfrak{h}^{(n,n+1)}] $  according to Eq.\ \eqref{aastar}.  We use the same notation for functions   $  h \in  \mathfrak{h}   $  but then understand the argument as $h$ restricted to $\mathfrak{h}^{(n,n+1)}$.

\noindent  Furthermore, we fix the following operator (defined on $\mathcal K \otimes  \mathcal F[\mathfrak{h}^{(n,n+1)}]$, and hence, on $  \mathcal H^{(n + 1)}   $ - see Remark \ref{R})
\begin{align} \label{operators222}
 V^{(n, n+1), \theta} := \sigma_1 \otimes \left(a_{n, n+1}(f^{\overline \theta})+ a_{n, n+1}(f^{\theta})^*  \right).
\end{align}

\noindent   In this notation we obtain (see Remark \ref{R}):
\begin{align} \label{HnHnm1}
H^{(n+1),\theta}= H^{(n),\theta}+H^{(n,n+1),\theta}_f+gV^{(n,n+1),\theta}.
\end{align}

\begin{lemma}
\label{lemma:is-normproj}
 Suppose that ($\mathcal P 1$)-($\mathcal P 4$) hold true for all $m\in\N$ such that $m\leq n$. Then,
\begin{align}
\norm{P^{(n)}_i} \leq 2+\frac{2|g|}{\rho} \leq 3 , \qquad i=0,1 
\end{align}
 \textbf{(notice that $|g| \leq \rho / 2$, see Definition \ref{gzero}- recall Remark \ref{recuerda})}
and 
\begin{align}\label{obv}
|\lambda_i^{(n)} - \lambda_i^{(1)} |   \leq 2  |g|.
\end{align} 
\end{lemma}
\begin{proof} 
Eq.\ \eqref{obv} is a  consequence of Property ($\mathcal P 1$). 
We estimate
\begin{align}
\norm{P^{(n)}_i}
&\leq \norm{  P^{(1)}_i}+ \sum^n_{j=2}\norm{P^{(j)}_i-P^{(j-1)}_i\otimes P_{\Omega^{(j-1,j)}}} 
\notag \\
& \leq  2+\frac{|g|}{\rho} \sum^{n-1}_{j=0}\left( \frac{1}{2}
\right)^{j}
\leq 2+\frac{2|g|}{\rho},
\end{align}
where  we  apply the induction hypothesis ($\mathcal P 3$) for $j\leq n$ and use Definition \ref{compile}  \textbf{and Definition \ref{gzero}}. 
\end{proof}
\begin{definition}
\label{def:regionM}
Let $n\in\N$ and $i\in\{0,1\}$.  We define the region 
\begin{align}
M_i^{(n)}:=B_i^{(n)}\setminus\left\{ z\in\C : \Im (z) \in \left(-\infty,  \Im (\lambda^{(n)}_i) - \frac{2}{5}\rho_{n+1} \sin (\nu) \right)  \right\}.
\end{align}
\end{definition}
\begin{lemma}
\label{lemma:is-keyest1}
 Suppose that ($\mathcal P 1$)-($\mathcal P 4$) hold true for all $m\in\N$ such that $m\leq n$. 
Then, for $i\in\{0,1\}$:
\begin{align}
\norm{\frac{1}{H^{(n),\theta}+H_f^{(n,n+1),\theta}-z}\overline{P^{(n,n+1)}_i}}\leq  \frac{24+  4 \bold C^{n+1}}{\sin (\boldsymbol \nu)}   \frac{1}{\rho_{n+1}+\left|  z- \lambda^{(n)}_i \right|}  ,
\end{align}
for all $z\in M_i^{(n)}$,  where
we have used the notation $P^{(n,n+1)}_i:=P^{(n)}_i \otimes P_{\Omega^{(n,n+1)}}$.
\end{lemma}
\begin{proof} 
Let $z\in M_i^{(n)}$.  
 Note that (see Remark \ref{R})
\begin{align}
\label{eq:lemma-is-Pident}
\overline{P^{(n)}_i}+P^{(n)}_i\otimes \overline{P_{\Omega^{(n,n+1)}}} 
&= 1-P^{(n)}_i+P^{(n)}_i\otimes \left( 1- P_{\Omega^{(n,n+1)}}\right)= 1-P^{(n)}_i \otimes P_{\Omega^{(n,n+1)}}
\notag \\
&=\overline{P^{(n,n+1)}_i}, 
\end{align}
and consequently, we obtain from  functional calculus (notice that $\left[ H^{(n),\theta} , H_f^{(n,n+1),\theta} \right]=0$)  that
\begin{align}
\label{eq:is-keyest10}
&\norm{\frac{1}{H^{(n),\theta}+H_f^{(n,n+1),\theta}-z}\overline{P^{(n,n+1)}_i}}
\notag \\
&\leq \norm{\frac{1}{H^{(n),\theta}+H_f^{(n,n+1),\theta}-z}\overline{P^{(n)}_i}}
+\norm{\frac{1}{H^{(n),\theta}+H_f^{(n,n+1),\theta}-z}P^{(n)}_i\otimes \overline{P_{\Omega^{(n,n+1)}}}}
\notag \\
&=\sup_{s\in \{0\} \cup [\rho_{n+1},\infty)}
\norm{\frac{1}{H^{(n),\theta}-(z-e^{-\theta}s)}\overline{P^{(n)}_i}} 
+\sup_{s\in  [\rho_{n+1},\infty)}\frac{\norm{P^{(n)}_i}}{|\lambda^{(n)}_i-(z-e^{-\theta}s)|} .
\end{align} 

Lemma \ref{lemma:resestinAprima}, Definition \ref{compile} and induction hypothesis ($\mathcal P4$), together with Lemma 
\ref{lemma:is-normproj}  \textbf{and the Definition of $\boldsymbol{C}$, see Remark \ref{recuerda}, in Definition \ref{C}     (notice that 
$  \boldsymbol{C} \geq   \frac{4\boldsymbol{c}}{ \sin( \boldsymbol{\nu} /2  ) }  \geq    \frac{4 C_{\eqref{const:resestinAprima}    }  }{ \sin( \boldsymbol{\nu} /2  ) }  $ and $  \| \overline{P^{(n)}_i}  \| \leq 4$)}, imply that 
\begin{align}\label{malech1}
\norm{\frac{1}{H^{(n),\theta}-(z-e^{-\theta}s)}\overline{P^{(n)}_i}} \leq  \frac{\bold C^{n+1}}{\rho_n+\left| \lambda^{(n)}_i - (z-e^{-\theta}s)  \right|}  , 
\end{align}
for every $s \in  \{0\} \cup [\rho_{n+1},\infty)$.

From the definitions of the sets 	$   M^{(n)}_i $ and 
$ \mathcal{S} $, it follows that 
\begin{align}
\label{eq:is-firstdist}
|\lambda^{(n)}_i-(z-e^{-\theta}s)| \geq \frac{1}{4} \rho_{n+1}  \sin (\nu)
\end{align}
for all $z\in M^{(n)}_i$ and $s\in [\rho_{n+1},\infty)$. 
Moreover, we define the sets
\begin{align}
\label{region:Gni}
G^{(n)}_i:= \left\{ z\in M^{(n)}_i :\Re (z) \geq \Re (\lambda^{(n)}_i)  \right\}, \quad i=0,1 ,
\end{align}
and for $d \geq 0$
\begin{align}
\label{region:Lndi}
L^{(n),d}_i:= \left\{ \lambda^{(n)}_i+e^{-\theta}(x+id):x\in\R \right\}, \quad i=0,1 .
\end{align}
Furthermore, we define
\begin{align}
\label{region:Lni}
L^{(n)}_i:=  \bigcup_{d\geq 0} L^{(n),d}_i \cap  G^{(n)}_i, \quad i=0,1 .
\end{align}
\begin{figure}[h]
\centering
\includegraphics[width=0.6\textwidth]{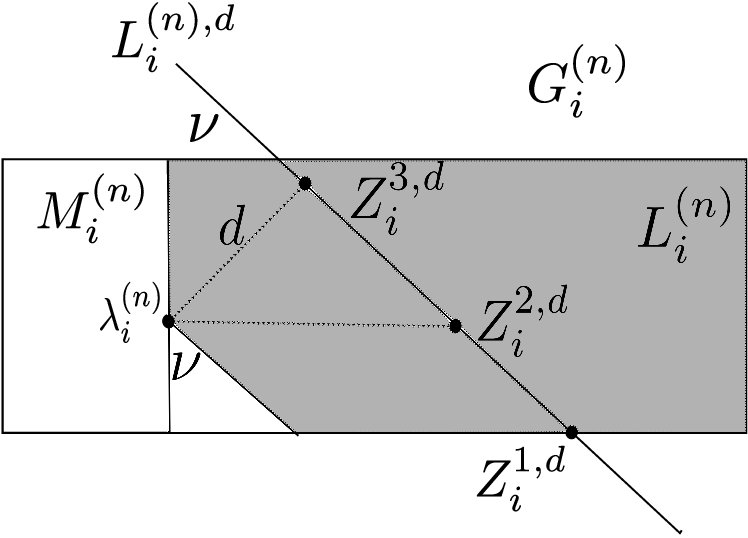}
  \caption{Subsets of $M_i^{(n)}$ }
    \label{fig:subsetbox}       
\end{figure}
Note that, by construction, we have
\begin{align}
 \text{dist}\left(L^{(n),d}_i, \lambda^{(n)}_i\right)= e^{- \Re \theta} d ,   \quad i=0,1,
\end{align}
and, by definition of  the sets 	$   M^{(n)}_i $ and 
$ \mathcal{S} $,   it follows that  
\begin{align}
\label{eq:estwithoutL}
\left|  z- \lambda^{(n)}_i \right|\leq \left|  z- \lambda^{(n)}_i -e^{-\theta}s\right|, \qquad \forall z\in M^{(n)}_i\setminus L^{(n)}_i, \quad \forall s\in [\rho_{n+1},\infty)  ,
\end{align}
where we have used the factor $ \frac{2}{5} $ in the definition of  $  M^{(n)}_i  $.
Let $Z^{1,d}_i$ and $Z^{2,d}_i$ be the intersections of $L^{(n),d}_i$ with the lines
\begin{align}
 \lambda^{(n)}_i-i \frac{2}{5} \rho_{n+1}\sin(\nu)+\R \quad \text{and} \quad \lambda^{(n)}_i+\R, 
\end{align}
respectively. Furthermore, we define $Z^{3,d}_i:=\lambda^{(n)}_i+de^{i\frac{\pi}{2}-\theta}$ and recall that $\nu <\pi/16$. Then,  we obtain
\begin{align}
\label{eq:regionest1}
&\sup_{z\in L^{(n),d}_i\cap G^{(n)}_i}\left|  z- \lambda^{(n)}_i \right|^2 
=\left| Z^{1,d}_i  - \lambda^{(n)}_i \right|^2 =e^{-2\Re \theta}d^2 + \left| Z^{3,d}_i  - Z^{1,d}_i   \right|^2
\notag \\
&=e^{-2\Re \theta}d^2 + \left(\left| Z^{3,d}_i  - Z^{2,d}_i   \right|+\left| Z^{2,d}_i  - Z^{1,d}_i   \right|\right)^2
=e^{-2\Re \theta}d^2 + \left(\frac{e^{-\Re \theta}d}{\tan (\nu)}+ 
 \frac{2}{5}\rho_{n+1}\right)^2 .
\end{align}
This yields the bound
\begin{align}
\label{eq:regionest2}
&\frac{\left|  z- \lambda^{(n)}_i \right|}{\left|  z- \lambda^{(n)}_i-e^{-\theta}s \right|}
\notag \\
&\leq \left[  \frac{e^{-2\Re \theta}d^2}{\left|  z- \lambda^{(n)}_i-e^{-\theta}s \right|^2}+\left(  \frac{e^{-\Re \theta}d}{ \tan (\nu) \left|  z- \lambda^{(n)}_i-e^{-\theta}s \right|}+ \frac{ 2 \rho_{n+1}}{  5\left|  z- \lambda^{(n)}_i-e^{-\theta}s \right|} \right)^2 \right]^\frac{1}{2}
\end{align}
for $s\in [\rho_{n+1},\infty)$ and $z\in L^{(n),d}_i\cap G^{(n)}_i$.
Note that  $\left|  z- \lambda^{(n)}_i-e^{-\theta}s \right|\geq e^{-\Re \theta}d$ for all $z\in L^{(n),d}_i$ and together with \eqref{eq:is-firstdist} we obtain
\begin{align}
\label{eq:regionest3}
\frac{\left|  z- \lambda^{(n)}_i \right|}{\left|  z- \lambda^{(n)}_i-e^{-\theta}s \right|}
\leq \left[  1+\left(  \frac{\cos (\nu)}{\sin (\nu) }+  \frac{8}{  5\sin (\nu) }\right)^2 \right]^\frac{1}{2}\leq \frac{4}{\sin ( \boldsymbol \nu)},
\end{align}
for all $z\in L^{(n)}_i   \cap G^{(n)}_i   $. This and \eqref{eq:estwithoutL} guarantees
\begin{align}
\label{eq:regionestwithL}
\frac{1}{\left|  z- \lambda^{(n)}_i-e^{-\theta}s \right|}
\leq  \frac{4}{\sin (\boldsymbol \nu)}\frac{1}{\left|  z- \lambda^{(n)}_i \right|} \qquad \forall z\in M^{(n)}_i, \quad \forall s\in [\rho_{n+1},\infty)  ,\quad i=0,1.
\end{align}
Now, if $\left|  z- \lambda^{(n)}_i \right|\geq \rho_{n+1}$, we use \eqref{eq:regionestwithL} and  compute
\begin{align}
\label{eq:finregest1}
\frac{1}{\left|  z- \lambda^{(n)}_i-e^{-\theta}s \right|}
&\leq  \frac{4}{\sin (\boldsymbol \nu)}\frac{\rho_{n+1}+\left|  z- \lambda^{(n)}_i \right|}{\left|  z- \lambda^{(n)}_i \right|} \frac{1}{\rho_{n+1}+\left|  z- \lambda^{(n)}_i \right|}
\notag \\
&\leq \frac{8}{\sin (\boldsymbol \nu)}\frac{1}{\rho_{n+1}+\left|  z- \lambda^{(n)}_i \right|} ,
\end{align}
and if $\left|  z- \lambda^{(n)}_i \right|< \rho_{n+1}$, we use \eqref{eq:is-firstdist} and find
\begin{align}
\label{eq:finregest2}
\frac{1}{\left|  z- \lambda^{(n)}_i-e^{-\theta}s \right|}
&\leq  \frac{4}{\sin (\boldsymbol \nu)}\frac{\rho_{n+1}+\left|  z- \lambda^{(n)}_i \right|}{\rho_{n+1}} \frac{1}{\rho_{n+1}+\left|  z- \lambda^{(n)}_i \right|}
\notag \\
&\leq \frac{8}{\sin (\boldsymbol \nu)}\frac{1}{\rho_{n+1}+\left|  z- \lambda^{(n)}_i \right|} .
\end{align}
We conclude from \eqref{eq:finregest1} and \eqref{eq:finregest2} that for $i=0,1$
\begin{align}
\label{eq:finregest0}
\frac{1}{\left|  z- \lambda^{(n)}_i-e^{-\theta}s \right|}
\leq \frac{8}{\sin (\boldsymbol \nu)}\frac{1}{\rho_{n+1}+\left|  z- \lambda^{(n)}_i \right|} \qquad \forall z\in M^{(n)}_i, \quad \forall s\in [\rho_{n+1},\infty)  
\end{align}
holds true.
Eqs.\   \eqref{eq:is-keyest10},  \eqref{malech1},  \eqref{eq:regionestwithL}, together with Lemma \ref{lemma:is-normproj} and Eq.\ \eqref{eq:finregest0}  
 yield
\begin{align}
\norm{\frac{1}{H^{(n),\theta}+H_f^{(n,n+1),\theta}-z}\overline{P^{(n,n+1)}_i}}
\leq   \frac{24+  4 \bold C^{n+1}}{\sin (\boldsymbol \nu)}   \frac{1}{\rho_{n+1}+\left|  z- \lambda^{(n)}_i \right|}  .
\end{align}
This completes the proof.
\end{proof}
\begin{lemma}
\label{lemma:is-keyest2-help}
 For all $z\in M_i^{(n)}\setminus \{\lambda_i^{(n)}\}$,  all  $ 0 \leq  r \leq \left|z-\lambda_i^{(n)}\right| $ and every $i\in\{0,1\}$:
\begin{align}
\norm{\frac{H_f^{(n,n+1)}+r}{H_f^{(n,n+1),\theta}-\left(z-\lambda_i^{(n)}\right)}}\leq \frac{ 10  }{\sin (\boldsymbol \nu)} . 
\end{align}
\end{lemma}
\begin{proof}
We calculate:
\begin{align}
\label{eq:helplemma1}
&\norm{\frac{H_f^{(n,n+1)}+r}{H_f^{(n,n+1),\theta}-(z-\lambda_i^{(n)})}}
= \sup_{y\in \{0\} \cup [\rho_{n+1},\infty)} \left| \frac{y+r}{e^{-\theta}y+\lambda_i^{(n)}-z}   \right|
\notag \\
&\leq    |e^{\theta}| +  |e^{\theta}|   \sup_{y\in \{0\} \cup [\rho_{n+1},\infty)} \left| \frac{e^{-\theta}r-\lambda_i^{(n)}+z}{e^{-\theta}y+\lambda_i^{(n)}-z}   \right|
\leq   |e^{\theta}|  +  (1 +   |e^{\theta}| ) \sup_{y\in \{0\} \cup [\rho_{n+1},\infty)} \frac{\left| z- \lambda_i^{(n)}\right| }{\left| e^{-\theta}y+\lambda_i^{(n)}-z\right| }  
\notag \\
&\leq  |e^{\theta}|  + \frac{  4 (1 +   |e^{\theta}| )  }{\sin (  \boldsymbol{\nu})}\leq  \frac{  10}{\sin (\boldsymbol \nu)},
\end{align}
where we have used \eqref{eq:regionestwithL} in the second last step.
\end{proof}
\begin{lemma}
\label{lemma:is-keyest2}
 Suppose that ($\mathcal P 1$)-($\mathcal P 4$) hold true for all $m\in\N$ such that $m\leq n$. 
Then,
\begin{align}
\norm{V^{(n,n+1),\theta}\frac{1}{H^{(n),\theta}+H_f^{(n,n+1),\theta}-z}}\leq   \frac{ 2500}{\rho \sin (\boldsymbol \nu)^2} \bold C^{n+1}\rho_n^{\mu}
\end{align}
for all $z\in M_i^{(n)}\setminus \{ \lambda_i^{(n)}\}$ such that $\left|  z- \lambda_i^{(n)}\right| \geq \frac{1}{10}\rho_{n+1}  \sin(\nu)$ and for all $i\in\{0,1\}$.
\end{lemma}
\begin{proof}
Set $r =\left|  z- \lambda_i^{(n)}\right| \geq \frac{1}{10}\rho_{n+1}\sin(\nu)$. We observe that
\begin{align}
\label{eq:keyest20}
&\norm{V^{(n,n+1) ,\theta}\frac{1}{H^{(n),\theta}+H_f^{(n,n+1),\theta}-z}}
\leq \norm{V^{(n,n+1) ,\theta}\frac{1}{H_f^{(n,n+1)}+r}}
\norm{\frac{H_f^{(n,n+1)}+r}{H_f^{(n,n+1),\theta}-(z-\lambda_i^{(n)})}}
\notag \\
& \hspace{4.4cm} \cdot \norm{\left( H_f^{(n,n+1),\theta}-(z-\lambda_i^{(n)})\right)\frac{1}{H^{(n),\theta}+H_f^{(n,n+1),\theta}-z}} .
\end{align}
Lemma \ref{lemma:is-keyest2-help} yields
\begin{align}
\label{eq:keyest21}
\norm{\frac{H_f^{(n,n+1)}+r}{H_f^{(n,n+1),\theta}-(z-\lambda_i^{(n)})}}\leq \frac{  10}{\sin (\boldsymbol \nu)} ,
\end{align}
and furthermore, we obtain from  functional calculus that
\begin{align}
\label{eq:similar}
&\norm{\left( H_f^{(n,n+1),\theta}-(z-\lambda_i^{(n)})\right)\frac{1}{H^{(n),\theta}+H_f^{(n,n+1),\theta}-z}} 
\leq \sup_{y\in\{0\}\cup [\rho_{n+1},\infty)}\norm{\frac{e^{-\theta}y+\lambda_i^{(n)}-z}{H^{(n),\theta}+e^{-\theta}y-z}}
\notag \\
&\leq \norm{ P_i^{(n)}}+\sup_{y\in\{0\}\cup [\rho_{n+1}, \infty)}\norm{\frac{\lambda_i^{(n)}-(z-e^{-\theta}y)}{H^{(n),\theta}-(z-e^{-\theta}y)}\overline{P_i^{(n)}}} \leq3+\bold C^{n+1}  \leq 4 \bold C^{n+1}.
\end{align}
In the last step, we use  Lemma \ref{lemma:is-normproj} for the first term.  For the second term, we utilize  Lemma \ref{lemma:resestinAprima}, Definition \ref{compile} and induction hypothesis ($\mathcal P4$), together with Lemma 
\ref{lemma:is-normproj}  \textbf{and the Definition of $\boldsymbol{C}$ in Definition \ref{C} - see Remark \ref{recuerda}   }  (notice that 
$  \boldsymbol{C} \geq   \frac{\boldsymbol{4 c}}{ \sin( \boldsymbol{\nu} /2  ) }  \geq    \frac{4 C_{\eqref{const:resestinAprima}    }  }{ \sin( \boldsymbol{\nu} /2  ) }  $ and $  \| \overline{P^{(n)}_i}  \| \leq 4$).

Using the proofs in Appendix \ref{app:sa}, we obtain
\begin{align}
\label{eq:keyest22}
&\norm{V^{(n,n+1) ,\theta}\frac{1}{H_f^{(n,n+1)}+r}}
\notag \\
& \leq \frac{1}{\sqrt{r}}\left( \norm{a_{n, n+1}(f^{\overline\theta}) \left( H_f^{(n,n+1)}+r\right)^{-\frac{1}{2}}}
+\norm{{a_{n, n+1}(f^\theta)}^*\left( H_f^{(n,n+1)}+r\right)^{-\frac{1}{2}}}\right) 
\notag \\
&\leq    \frac{1}{r} \norm{f^\theta}_{\mathfrak{h}^{(n,n+1)}} + \frac{2}{\sqrt{r}} \norm{f^\theta/\sqrt{\omega}}_{\mathfrak{h}^{(n,n+1)}}.  
\end{align}
We estimate
\begin{align}\label{tuto1}
\norm{f^\theta}_{\mathfrak{h}^{(n,n+1)}}
&=   \sqrt{\int_{\mathcal B_{\rho_n} \setminus \mathcal B_{\rho_{n+1}}} \mathrm{d}^3k \, |f^\theta(k)|^2}
=   |e^{ - \theta (1+ \mu)} |   \sqrt{4\pi\int^{\rho_{n}}_{\rho_{n+1}} \mathrm{d}u\,  u^{1+2\mu}|e^{-2e^{2\theta}\frac{u^2}{\Lambda^2}}}|
\notag \\
&\leq    |e^{ - \theta (1+ \mu)} |    \sqrt{4\pi}   \rho_n^{\mu} \rho_n,
\end{align}
and similarly, 
\begin{align} \label{tuto2}
\norm{f^\theta/\sqrt \omega}_{\mathfrak{h}^{(n,n+1)}}
\leq   |e^{ - \theta (1+ \mu)} |    \sqrt{4\pi}\rho_n^{\mu}\rho_n^{\frac{1}{2}}. 
\end{align}
From our choice of $r$, it follows that   $\sqrt{\frac{\rho_n}{r}}\leq \frac{\sqrt{10}}{\sqrt{\rho \sin (\nu)}}  $  and, consequently, we obtain
\begin{align}
\label{eq:keyest22'}
\norm{V^{(n,n+1) ,\theta}\frac{1}{H_f^{(n,n+1)}+r}}
&\leq    |e^{ - \theta (1+ \mu)} |   \sqrt{4\pi}  \left(  \frac{\rho_n}{r}  +2\sqrt{ \frac{\rho_n}{r}}\right)\rho_n^{\mu}
\leq   |e^{ - \theta (1+ \mu)} |   \frac{60}{\rho \sin (\boldsymbol \nu)}\rho_n^{\mu} . 
\end{align}
Plugging \eqref{eq:keyest21}, \eqref{eq:similar} and \eqref{eq:keyest22'} into \eqref{eq:keyest20} yields (we recall that $ \mu \in (0, 1/2)$)
\begin{align}
\norm{V^{(n,n+1) ,\theta}\frac{1}{H^{(n),\theta}+H_f^{(n,n+1),\theta}-z}}
\leq    \frac{  2500}{\rho \sin (\boldsymbol \nu)^2}\bold C^{n+1}\rho_n^{\mu}.
\end{align}
This completes the proof.
\end{proof}

\subsubsection{Induction step}
\label{sec:is}
In this section, we apply the results from Section \ref{sec:indkey} in order to  show the induction step, i.e., we assume that ($\mathcal P 1$)-($\mathcal P 4$) hold true for all $m\leq n \in\N$, and prove that ($\mathcal P 1$)-($\mathcal P 4$) hold true for  $ n+1$. This together with Remark \ref{rem.ib} completes the proof of Theorem \ref{thm:ind}. 

We  first employ the estimates of  Section \ref{sec:indkey} in order to prove Property ($\mathcal P 2$) and ($\mathcal P 3$). After this, we prove    ($\mathcal P 1$). Finally,  ($\mathcal P 4$)  follows again from the results of Section \ref{sec:indkey}   together with the maximum modulus principle.    \\

\textbf{Proof  of  ($\mathcal P 2$) and ($\mathcal P 3$)   :}
\begin{proposition}
\label{prop:is-resestwithoutdisc}
Suppose that ($\mathcal P 1$)-($\mathcal P 4$) hold true for all $m\in\N$ with $m\leq n$. 
Then, 
\begin{align}
\norm{\frac{1}{H^{(n+1),\theta}-z}}\leq\frac{168+ 8 \bold C^{n+1}}{\sin (\boldsymbol \nu)}   \frac{1}{\rho_{n+1}+\left|  z- \lambda^{(n)}_i \right|} ,
\end{align}
for all $z\in M_i^{(n)}\setminus \{ \lambda_i^{(n)}\}$  (see Definition \ref{def:regionM}) such that $\left|  z- \lambda_i^{(n)}\right| \geq \frac{1}{10}\rho_{n+1}  \sin(\nu)$ and for all $i\in\{0,1\}$.
\end{proposition}
\begin{proof}
Let $z\in M_i^{(n)}\setminus \{ \lambda_i^{(n)}\}$ such that $\left|  z- \lambda_i^{(n)}\right| \geq \frac{1}{10}\rho_{n+1} \sin(\nu)$ and $i\in\{0,1\}$. Then, it follows from Lemma \ref{lemma:is-keyest2} that
\begin{align} \label{sss}
&\norm{V^{(n,n+1),\theta}\frac{1}{H^{(n),\theta}+H_f^{(n,n+1),\theta}-z}}
\leq    \frac{ 2500}{\rho \sin (\boldsymbol \nu)^2} \bold C^{n+1}\rho_n^{\mu}   .
\end{align}
 \textbf{Our assumption on $g$  in Definition \ref{gzero} together with 
\eqref{sss}  yield  that (notice that Definition \ref{sequence} implies that  
$    C^{n+1}\rho_n^{\mu}  \leq 1  $ and Definition \ref{gzero} implies that 
$ \frac{  g 2500}{\rho \sin (\boldsymbol \nu)^2} \leq \frac{1}{2}    $, see also Remark \ref{recuerda})} 
\begin{align}
\label{eq:is-prop-relboundg}
\norm{gV^{(n,n+1),\theta}\frac{1}{H^{(n),\theta}+H_f^{(n,n+1),\theta}-z}}
&  \leq \frac{1}{2}  .
\end{align}
This and Lemma \ref{lemma:is-keyest1} imply that
\begin{align}
H^{(n+1),\theta}-z=\left(1+ gV^{(n,n+1),\theta} \frac{1}{H^{(n),\theta}+H_f^{(n,n+1),\theta}-z  }    \right)\left(  H^{(n),\theta}+H_f^{(n,n+1),\theta}-z  \right)
\end{align}
is invertible, and we estimate
\begin{align}
\label{eq-is-propwithoutdisc}
\norm{\frac{1}{H^{(n+1),\theta}-z}}
&\leq 2 \norm{\frac{1}{H^{(n),\theta}+H_f^{(n,n+1),\theta}-z }}
\notag \\
&\leq 2\frac{\norm{P^{(n,n+1)}_i}}{\left|  z-\lambda^{(n)}_i    \right|}
+2 
\norm{\frac{1}{H^{(n),\theta}+H_f^{(n,n+1),\theta}-z}\overline{P^{(n,n+1)}_i}}
\notag \\
&\leq 2\frac{\norm{P^{(n)}_i}}{\left|  z-\lambda^{(n)}_i    \right|}
+
\frac{48+  8 \bold C^{n+1}}{\sin (\boldsymbol \nu)}   \frac{1}{\rho_{n+1}+\left|  z- \lambda^{(n)}_i \right|}, 
\end{align}
where we apply Lemma \ref{lemma:is-keyest1}. Moreover, Lemma \ref{lemma:is-normproj}  implies that $\norm{P^{(n)}_i}\leq 3$ and  it follows from $\left|  z- \lambda_i^{(n)}\right| \geq \frac{1}{10}\rho_{n+1} \sin(\nu)$ that
\begin{align}
\frac{1}{\left|  z-\lambda^{(n)}_i  \right|}\leq \frac{20}{\sin(\boldsymbol \nu)}\frac{1}{\rho_{n+1}+\left|  z-\lambda^{(n)}_i  \right|} .
\end{align}
Altogether, we obtain
\begin{align}
\norm{\frac{1}{H^{(n+1),\theta}-z}}
&\leq 
\frac{168+ 8  \bold C^{n+1}}{\sin (\boldsymbol \nu)}   \frac{1}{\rho_{n+1}+\left|  z- \lambda^{(n)}_i \right|}  ,
\end{align}
and thereby,  complete the proof.
\end{proof}
\begin{lemma}
\label{lemma:is-p3}
 Suppose that ($\mathcal P 1$)-($\mathcal P 4$) hold true for all $m\in\N$ with $m\leq n$.  We define
\begin{align}
\label{eq:Phatn+1}
\hat P^{(n+1)}_i := -\frac{1}{2\pi i} \int_{\hat \gamma^{(n)}_i}\mathrm{d}z \, \frac{1}{H^{(n+1),\theta}-z},
\end{align}
where 
\begin{align}
\label{eq:gammainhat}
\hat \gamma^{(n)}_i: [0,2\pi ] \to \C , \quad
t \mapsto \lambda^{(n)}_i +  \frac{1}{8} \rho_{n+1}  \sin (\nu)  e^{it} .
\end{align}
Then,
\begin{align}
\norm{\hat P^{(n+1)}_i - P^{(n)}_i\otimes P_{\Omega^{(n,n+1)}}} \leq \frac{|g|}{\rho}\bold C^{2(n+1)+2} \rho^\mu_{n}\leq \frac{|g|}{\rho} \left( \frac{1}{2}
\right)^{n} .
\end{align}
(The last inequality follows from Definition \ref{sequence}.)
\end{lemma}
\begin{proof}
Recall that the definition of $P^{(n)}_i$ is introduced in Eq.\ \eqref{eq:projin}. We notice that the function  
\begin{align}
\label{eq:p3map}
  B^{(n)}_i\setminus \{ \lambda^{(n)}_i  \}   \ni z \mapsto \frac{1}{H^{(n),\theta}-z}
\end{align}
is analytic as an operator valued function and the region between $\hat \gamma^{(n)}_i$ and $ \gamma^{(n)}_i$  is contained in the domain of \eqref{eq:p3map}. We  obtain from  the Cauchy integral theorem that
\begin{align}
 P^{(n)}_i  = -\frac{1}{2\pi i} \int_{ \gamma^{(n)}_i}\mathrm{d}z \, \frac{1}{H^{(n),\theta}-z}= -\frac{1}{2\pi i} \int_{ \hat \gamma^{(n)}_i}\mathrm{d}z \, \frac{1}{H^{(n),\theta}-z} .
\end{align}

\noindent  As in Remark \ref{remp}, it turns out that (see Remark \ref{R})
\begin{align}
P^{(n)}_i\otimes P_{\Omega^{(n,n+1)}} =  -\frac{1}{2\pi i} \int_{ \hat \gamma^{(n)}_i}\mathrm{d}z \, \frac{1}{H^{(n),\theta} + H_f^{(n,n+1),\theta}  -z}  .
\end{align}
We calculate 
\begin{align} \label{caraj1}
&\big \| \hat P^{(n+1)}_i  -  P^{(n)}_i\otimes P_{\Omega^{(n,n+1)}} \big \|     = 
   \frac{1}{2\pi } \Big \| \int_{ \hat \gamma^{(n)}_i}\mathrm{d}z \,
    \frac{1}{H^{(n+1),\theta}   -z}  - 
     \frac{1}{H^{(n),\theta} + H_f^{(n, n+1),\theta}  -z} \Big \|   
   \notag  \\
&   = 
   \frac{1}{2\pi } \Big \| \int_{ \hat \gamma^{(n)}_i}\mathrm{d}z \,
    \frac{1}{H^{(n+1),\theta}   -z}  gV^{(n,n+1),\theta} 
     \frac{1}{H^{(n),\theta} + H_f^{(n, n+1),\theta}  -z} \Big \|  .
\end{align}

\noindent Furthermore, Lemma \ref{lemma:is-keyest2} implies  that for $z$ in the curve $ \hat \gamma^{(n)}_i  $
\begin{align}
\label{caraj2}
\norm{ gV^{(n,n+1),\theta}\frac{1}{H^{(n),\theta}+H_f^{(n,n+1),\theta}- z}}\leq 
 |g| \frac{ 2500}{\rho \sin (\boldsymbol \nu)^2} \bold C^{n+1}\rho_n^{\mu}
,
\end{align}
and  Proposition \ref{prop:is-resestwithoutdisc}  ensures that 
\begin{align}
\label{caraj3}
\norm{\frac{1 }{H^{(n+1),\theta}- z}} 
&\leq\frac{168+ 8\bold C^{n+1}}{\sin (\boldsymbol \nu)}   \frac{1}{\rho_{n+1}  } .
\end{align} 
Eqs \eqref{caraj1}-\eqref{caraj3} imply 
\begin{align}
\norm{\hat P^{(n+1)}_i - P^{(n)}_i\otimes P_{\Omega^{(n,n+1)}}} 
\leq |g|  
\frac{ 2500}{\rho \sin (\boldsymbol \nu)^2} \bold C^{n+1}\rho_n^{\mu}  \frac{168+ 8\bold C^{n+1}}{\sin (\boldsymbol \nu)} \leq  \frac{ |g| 2500}{\rho \sin (\boldsymbol \nu)^2} \bold C^{2n+2}\rho_n^{\mu}  \frac{200 }{\sin (\boldsymbol \nu)},    
\end{align}
 \textbf{which together with the definition of  $\boldsymbol{C}$  in Definitions \ref{C}   imply the desired result (Definition \ref{C} imply that    
 $    \frac{   500 \, 000}{  \sin (\boldsymbol \nu)^3}    \leq  \bold C^2 $, see also Remark \ref{recuerda}). }
\end{proof}

\begin{proposition}[Proof of Properties ($\mathcal P2$) and ($\mathcal P3$)]
\label{rem:1}
 Suppose that ($\mathcal P 1$)-($\mathcal P 4$) hold true for all $m\in\N$ with $m\leq n$, then  ($\mathcal P 2$) and ($\mathcal P 3$) hold true for $n+1$. 
\end{proposition}
\begin{proof}
Lemma \ref{lemma:is-p3} implies that $\norm{ \hat P^{(n+1)}_i - P^{(n)}_i\otimes P_{\Omega^{(n,n+1)}}} < 1$  \textbf{(see Definition \ref{gzero} and recall Remark \ref{recuerda})}. From the induction hypothesis it follows that $P^{(n)}_i\otimes P_{\Omega^{(n,n+1)}}$ is a rank-one projection. Therefore, $\hat P^{(n+1)}_i $ is also a rank-one projection.  
Proposition \ref{prop:is-resestwithoutdisc} implies that $ H^{(n+1),\theta}  $ 
has no spectral points in  $ M_i^{(n)}\setminus  D \Big (\lambda_i^{(n)} ,  \frac{1}{10}\rho_{n+1}  \sin(\nu)  \Big )$. Since  the contour of integration for 
$\hat P^{(n+1)}_i $ is contained in  $  M_i^{(n)} $ and its interior  contains   $  D \Big (\lambda_i^{(n)} ,  \frac{1}{10}\rho_{n+1}  \sin(\nu)  \Big )  $, we obtain that 
 there is only one point in $  M^{(n)}_i  $  contained in the spectrum of $H^{(n+1)}_i$.  
This point is the eigenvalue  $\lambda^{(n+1)}_i$ that we introduced above. Lemma \ref{prop:is-resestwithoutdisc} implies that $  |  \lambda^{(n+1)}_i - \lambda^{(n)}_i  | \leq \frac{1}{10} \rho_{n+1} \sin(\nu) $. This in turn implies that  $  B_i^{(n+1)} \subset M_i^{(n)} $. Then, $  \lambda^{(n+1)}_i $ is the only spectral point of $ H_i^{(n+1), \theta} $ in $   B_i^{(n+1)}   $, which is Property ($\mathcal P 2$). A deformation in the integration contour in the definitions of 
$ \hat P^{(n+1)}_i $ and $   P^{(n+1)}_i $ implies that these projections coincide and, therefore, Property ($\mathcal P 3$) is a consequence of Lemma  \ref{lemma:is-p3}.  
\end{proof}

\textbf{Proof of Property ($\mathcal P 1$):}

\begin{proposition}[Proof of Property ($\mathcal P1$)]
\label{prop:is-p1}
 Suppose that ($\mathcal P 1$)-($\mathcal P 4$) hold true for all $m\in\N$ with $m\leq n$. 
Then,  we obtain for $i=0,1$ that
\begin{align}
\left|  \lambda^{(n+1)}_i - \lambda^{(n)}_i \right|
\leq |g| \bold C^{(n+1)+1}\rho_n^{1+\mu} \leq |g| \left( \frac{1}{2} \right)^n \rho_n
\end{align}
holds true.
Notice that the last inequality follows from Definition \ref{sequence}.  
\end{proposition}
\begin{proof}
 In this proof we explicitly emphasize the dependence of $ P^{(n)}_{i}  $ on $\theta$ and write $ P^{(n)}_{i} \equiv  P^{(n),\theta}_{i}$.  We define $\Psi^{(n),\theta}_{i}:=   P^{(n),\theta}_{i}\varphi_i\otimes \Omega^{(n+1)} $, see Remark \ref{R}. Proposition  \ref{rem:1}, Property $\mathcal P 3$, Definition \ref{compile} (see Remark \ref{lemma:ib-proj2})  \textbf{and the restrictions for $g$ in Definition \ref{gzero}  imply that  $  \| \Psi^{(n),\theta}_{i} -   \varphi_i\otimes \Omega^{(n + 1)}  \|  \leq \frac{1}{10^2}  $. This  guarantees that }
\begin{align}
\label{eq:p1proof}
\norm{\Psi^{(n),\theta}_{i}}\leq 2 \qquad \text{and} \qquad  \Big | \left\langle \Psi^{(n),\overline \theta}_{i},  P^{(n+1),\theta}_{i} \Psi^{(n),\theta}_{i}\right\rangle \Big | \geq \frac{1}{2} .
\end{align}
Notice that in this work we assume that the imaginary part of $\theta$,  $ \nu $, is positive. Then, strictly speaking, we do not have the right to use the symbol  $\Psi^{(n), \overline \theta}_{i}:=   P^{(n),\overline \theta}_{i}\varphi_i\otimes \Omega^{(n+1)} $. However, the restriction we impose by assuming that $ \nu  $ is  not negative is irrelevant. This is assumed only for convenience in order to simplify our notation. Of course, the same results hold true if we take $ - \pi/16  < \nu < - \boldsymbol{\nu}  $ (we use this fact in the present proof as well as $    P^{(n),\overline \theta}_{i} = \big ( P^{(n),\theta}_{i}  \big )^*$).       
 Then, we obtain
\begin{align}
\label{eq:distlambdas}
\lambda_i^{(n+1)}
&= \frac{\left\langle \Psi^{(n),\overline \theta}_{i}, H^{(n+1),  \theta} P^{(n+1),\theta}_i \Psi^{(n),\theta}_{i}\right\rangle}{\left\langle \Psi^{(n),\overline \theta}_{i},  P^{(n+1),\theta}_{i} \Psi^{(n),\theta}_{i}\right\rangle}
= \frac{\left\langle  H^{(n+1),\overline \theta}\Psi^{(n),\overline \theta}_{i},  P^{(n+1),\theta}_i \Psi^{(n),\theta}_{i}\right\rangle}{\left\langle \Psi^{(n),\overline \theta}_{i},  P^{(n+1),\theta}_{i} \Psi^{(n),\theta}_{i}\right\rangle}
\notag \\
&=\lambda_i^{(n)}+ g\frac{\left\langle V^{(n,n+1),\overline \theta} \Psi^{(n),\overline \theta}_{i},  P^{(n+1),\theta}_{i} \Psi^{(n),\theta}_{i}\right\rangle}{\left\langle \Psi^{(n),\overline \theta}_{i},  P^{(n+1),\theta}_{i} \Psi^{(n),\theta}_{i}\right\rangle} .
\end{align} 
Now we choose  $z\in\C$ such that $|z-\lambda_i^{(n)}|=\frac{\rho_{n+1} \sin (\nu)}{10}$. We get that 
\begin{align}
&\left| \lambda_i^{(n+1)} -\lambda_i^{(n)}\right| 
\leq   2 \norm{gV^{(n,n+1)  ,\overline \theta}  P^{(n),\overline \theta}_{i} \varphi_i \otimes \Omega^{(n+1)}}   \Big \| P^{(n+1),\theta}_{i} \Psi^{(n),\theta}_{i} \Big \|  \notag \\ &\leq  54 \left| z-\lambda_i^{(n)} \right| \norm {gV^{(n,n+1),\overline \theta}  \frac{1}{H^{(n),\overline \theta} +H_f^{(n,n+1),\overline \theta}   -z}} 
\notag \\
&\leq  g 54 \frac{\rho_{n+1} \sin (\nu) }{10} \frac{2500}{\rho \sin (\boldsymbol{ \nu})^2} \bold C^{n+1}\rho_n^{\mu} 
  \leq    |g|  \bold C^{(n+1)+1}\rho_n^{1+\mu} \leq |g|  \left( \frac{1}{2} \right)^n \rho_n ,
\end{align}  
where we use  Lemmas \ref{lemma:is-keyest2} and \ref{lemma:is-normproj}   \textbf{and the definition of $ \boldsymbol{C} $ in Definition  \ref{C} (it implies that $ 54 \frac{2500}{\sin(\boldsymbol{\nu})^2}  \leq \boldsymbol{C} $, see also Remark \ref{recuerda}).}
\end{proof}

\textbf{Proof of Property ($\mathcal P 4$):}

\begin{lemma}
\label{prop:is-resest-fin}
 Suppose that ($\mathcal P 1$)-($\mathcal P 4$) hold true for all $m\in\N$ with $m\leq n$. 
Then,  for $i\in\{0,1\}$:
\begin{align}
\label{eq:is-resest-fin}
\norm{\frac{1}{H^{(n+1),\theta}-z}\overline{P^{(n+1)}_i}}\leq 3 \frac{168+ 8\bold C^{n+1}}{\sin (\boldsymbol \nu)}   \frac{1}{\rho_{n+1}+\left|  z- \lambda^{(n)}_i \right|},  
\qquad \forall z\in  M_i^{(m)},
\end{align}
and  hence, for all $z \in B_i^{(n+1)}$ (recall   $  B_i^{(n+1)} \subset  M_i^{(n)}$ from the proof of Proposition \ref{rem:1}).  
\end{lemma}

\begin{proof} 
Let $z\in  M_i^{(n)} $ such that $\left|  z- \lambda^{(n)}_i \right|\geq \frac{\rho_{n+1}}{10} \sin(\nu)$ and $i\in\{0,1\}$. Then,
\eqref{eq:is-resest-fin} follows from 
 Proposition \ref{prop:is-resestwithoutdisc} and the fact that 
  $\big \| \overline{ P^{(n+1)}_i} \big \| \leq 3$ (see the proof of Lemma \ref{lemma:is-normproj} and Proposition \ref{rem:1}).  Furthermore, we observe  that $ M_i^{(n)} \ni z \mapsto\frac{1}{H^{(n+1),\theta}-z}\overline{P^{(n+1)}_i}$ is analytic (see the proof of Proposition \ref{rem:1}), and hence, \eqref{eq:is-resest-fin} follows for $\left|  z- \lambda^{(n)}_i \right|\leq \frac{\rho_{n+1}}{10} \sin(\nu)$ from the maximum modulus principle of complex analysis.
\end{proof}

\begin{proposition}[Proof of Property ($\mathcal P 4$)]
\label{prop:is-p5}
 Suppose that ($\mathcal P 1$)-($\mathcal P 4$) hold true for all $m\in\N$ with $m\leq n$ and take $i\in\{0,1\}$.
Then, 
\begin{align}
\label{eq:is-resest-p5}
\norm{\frac{1}{H^{(n+1),\theta}-z}\overline{P^{(n+1)}_i}}
&\leq 4 \frac{168+ 8\bold C^{n+1}}{|\sin (\boldsymbol \nu)|}   \frac{1}{\rho_{n+1}+\left|  z- \lambda^{(n+1)}_i \right|} 
\leq   \frac{\bold C^{(n+1)+1}}{\rho_{n+1}+\left|  z- \lambda^{(n+1)}_i \right|} 
\end{align}
for all $z\in B^{(n+1)}_i$. 
\end{proposition}
\begin{proof}
Let  $i\in\{0,1\}$ and $z\in B^{(n+1)}_i$. It follows from Proposition \ref{prop:is-resest-fin} that
\begin{align}
\label{eq:p5-est0}
\norm{\frac{1}{H^{(n+1),\theta}-z}\overline{P^{(n+1)}_i}}
&\leq 3 \frac{168+ 8C(\nu)^{n+1}}{\sin (\boldsymbol \nu)}   \frac{1}{\rho_{n+1}+\left|  z- \lambda^{(n)}_i \right|} 
\end{align}
holds true.
Lemma \ref{prop:is-resestwithoutdisc} implies that $  |  \lambda^{(n+1)}_i - \lambda^{(n)}_i  | \leq \frac{1}{10} \rho_{n+1} \sin(\nu)  \leq \frac{1}{10} \rho_{n+1}   $ .  Therefore, 
\begin{align}
\frac{1}{\rho_{n+1}+\left|  z- \lambda^{(n)}_i \right|} \leq \frac{10}{9}\frac{1}{\rho_{n+1}+\left|  z- \lambda^{(n+1)}_i \right|} .
\end{align}
This together with \eqref{eq:p5-est0} yields
\begin{align}
\label{eq:p5-est0'}
\norm{\frac{1}{H^{(n+1),\theta}-z}\overline{P^{(n+1)}_i}}
&\leq 4 \frac{168+  8\bold C^{n+1}}{\sin (\boldsymbol \nu)}   \frac{1}{\rho_{n+1}+\left|  z- \lambda^{(n+1)}_i \right|} 
\leq  \frac{\bold C^{(n+1)+1} }{\rho_{n+1}+\left|  z- \lambda^{(n+1)}_i \right|} ,
\end{align}
where in the last line  \textbf{we use the definition of $ \boldsymbol C $ in Definition \ref{C} (it implies that $ 4 \frac{168+  8 }{\sin (\boldsymbol \nu)} \leq \bold C    $, see also Remark \ref{recuerda}).} 
\end{proof}

\section{Resolvent and spectral  estimates} 
\label{resolvent-estimates}

In this section we prove Theorems \ref{resolventestimates}, \ref{spectralestimates} and Proposition \ref{thm:res}.  The resolvent and spectral estimates that we provide are essentially different from the ones presented in Section \ref{sec:4} and \cite{bbp}. The reason is the following: in \cite{bbp} the construction of the resonance is based in a sequence of infrared cut-off 
Hamiltonians. As  the parameter, $n$, of the sequence tends to infinity the cut-off is removed. Each cut-off Hamiltonian has a resonance that is isolated from the rest of the spectrum and they tend to the resonance of the Hamiltonian without cut-off.  The delicate point is to estimate spectra of the cut-off Hamiltonians in such a way that these estimates do not require conditions in the coupling constant that depend on $n$. This implies a selection of spectral regions to be analyzed at each step $n$. In \cite{bbp}  these regions are chosen in neighborhoods of the resonances and far away from the rest of the spectrum, because the interest lies in constructing the resonance of the full Hamiltonian. Here, we need more subtle estimates in regions that are not only close to the resonances but to other parts of the spectrum of the cut-off Hamiltonians. Then, we get resolvent estimates  in terms of the distance to the spectrum  rather than the distance to the resonance, as it is done in \cite{bbp}.  The regions that we choose are complements of cones with vertexes in neighborhoods of the resonances.  Some parts of the cones are closer to the resonances than to the rest of the spectrum and other parts of the cones are closer to other spectral points.  This makes our analysis harder than in \cite{bbp}. Our analysis requires a geometric  construction that controls spectra and resolvents outside  cones at step $n$ using the same information for the exterior of cones at step  $n-1$. In Section \ref{resspec} we analyze the infrared cut-off Hamiltonians and prove spectral and resolvent estimates about them 
(Theorem \ref{ResPrinc}). Geometric aspects of the cones are presented in Lemmas 
\ref{lemma:subset} and \ref{coness} below.  In Lemma \ref{mana0} we give resolvent (and hence spectral) estimates of a Hamiltonian that is obtained by adding to the Hamiltonian at step $n$ the free energy of step $ n+1 $, using the information we have at step $n$. From this last Hamiltonian we obtain the Hamiltonian at step $n+1$ by adding the interacting energy at step $n+1$, the analysis of this is presented in Lemma \ref{mana}.  Theorem \ref{ResPrinc} is a consequence of Lemma \ref{mana}. The study of the full Hamiltonian is carried out in Theorems \ref{resuno} and \ref{resdos} in Section \ref{sec:resest100},  using Theorem \ref{ResPrinc}, in a similar manner as in Section \ref{resspec}.  First, we add the full free energy to the Hamiltonian at step $n$ in Lemma \ref{EL}, and  then, we add the full interacting energy, using Lemma \ref{lemma:is-keyesttilde}, in  Theorems \ref{resuno} and \ref{resdos}. These theorems imply Theorems \ref{resolventestimates} and \ref{spectralestimates}. The proof of Proposition  \ref{thm:res} is not difficult and it  is presented in Section 
\ref{sec:proof-res}.

We start with introducing some notation. 
In this section we assume that Definitions \ref{C}, \ref{sequence} and \ref{gzero} hold true. We additionally assume that $ \boldsymbol{ C} \geq  \boldsymbol{D}\sin(\nu/m)^{-1} $,  in order to freely apply  Corollary \ref{coro:ib-res} and Eq. \eqref{coroeq.ib-resT}. 
We  fix the Hamiltonians   (see Remark \ref{R})
\begin{align}
\tilde H^{(n),\theta} := H^{\theta}_0 +gV^{(n),\theta},
\end{align}
which are densely defined on the Hilbert space $\mathcal H$. 
We recall that we already defined
\begin{align} \label{pet1}
 \mathfrak{h}^{(n,\infty)}= L^2(\mathcal B_{\rho_n} )   
\end{align} 
 and  the corresponding Fock  space  $  \mathcal F[\mathfrak{h}^{(n, \infty)}]  $ (it is defined in \eqref{Fock}), with vacuum state $\Omega^{(n, \infty)}$. We identify, as above,
\begin{align} \label{pet2}
\mathcal H \equiv \mathcal H^{(n)}\otimes \mathcal F[\mathfrak{h}^{(n,\infty )}]. 
\end{align}   
We define the free boson energy operator  on  $\mathcal  F[\mathfrak{h}^{(n, \infty)}]  $  by restricting the definition in Eq.\ \eqref{h0def} accordingly and denote it by the symbol $  H_f^{(n, \infty), 0} \equiv  H_f^{(n, \infty)} $. We set 
\begin{align}\label{Hnthetaprima111}
 H_f^{(n, \infty), \theta}:= e^{-\theta}   H_f^{(n, \infty), 0}. 
\end{align}
For every function $  h \in  \mathfrak{h}^{(n, \infty)}$, we define the creation and annihilation operators, $a_{n, \infty}(h)$ and $a_{n, \infty}^*(h) $, on  $ \mathcal  F[\mathfrak{h}^{(n,\infty)}] $   according to Eq.\ \eqref{aastar}.   Again, we use the same notation also for   $  h \in  \mathfrak{h}   $ but then understand $h$ as its restriction to $\mathfrak{h}^{(n,n+1)}$.

\noindent We fix the following operator (defined on $\mathcal K \otimes  \mathcal F[\mathfrak{h}^{(n,\infty)}] $, and hence, on $  \mathcal H   $ - see Remark \ref{R})

\begin{align} \label{operators333}
 V^{(n, \infty), \theta} := \sigma_1 \otimes \left(a_{n, \infty}(f^{\overline \theta})+ a_{n, \infty}(f^{ \theta})^*  \right),
\end{align}  
and further, we obtain (see Remark \ref{R}):
\begin{align} \label{HnHnmnn}
H^{\theta}= H^{(n),\theta}+H^{(n,\infty),\theta}_f+gV^{(n,\infty),\theta}
= \tilde  H^{(n),\theta}  +gV^{(n,\infty),\theta} .
\end{align}

\subsection{Resolvent and spectral estimates multi-scale analysis}
\label{resspec}
In this section  we analyze the infrared cut-off Hamiltonians and prove spectral and resolvent estimates about them 
(Theorem \ref{ResPrinc}). Geometric aspects of the cones are presented in Lemmas 
\ref{lemma:subset} and \ref{coness} below.  In Lemma \ref{mana0} we give resolvent (and hence spectral) estimates of a Hamiltonian that is obtained by adding to the Hamiltonian at step $n$ the free energy of step $ n+1 $, using the information we have at step $n$. From this last Hamiltonian we obtain the Hamiltonian at step $n+1$ by adding the interacting energy at step $n+1$, the analysis of this is presented in Lemma \ref{mana}.  Theorem \ref{ResPrinc} is a consequence of Lemma \ref{mana}. 
\begin{lemma} 
\label{lemma:subset}
Suppose that  $|g| \leq  \rho \frac{1}{10} \sin(\nu / 2 m)$.  
We define for $i=0,1$
\begin{align}
   q^{(n)}_i :=  \lambda^{(n)}_i + 
\frac{1}{4} \rho_n e^{-i \nu},  \hspace{1cm}   q^{(n, n+1)}_i :=  \lambda^{(n)}_i + 
\Big (\frac{2}{5}  - \frac{1}{100} \Big ) \rho_{n+1} e^{-i \nu} . 
\end{align}
It follows that
\begin{align} \label{ter0}
| \lambda_i -  \lambda_i^{(n)} | \leq    & 2  |g|    \rho_{n}^{1+ \mu/2}
\end{align}
and 
\begin{align}\label{ter1}
\mathcal{C}_m ( q^{(n)}_i )  \subset \mathcal{C}_m ( q^{(n, n+1)}_i )    \subset  \mathcal{C}_m ( q^{(n+ 1)}_i   ),
\end{align}
where the set $\mathcal C_m(\cdot)$ is defined in \eqref{eq:defcone} (see Figure \ref{fig:cones}). 
Moreover, 
\begin{align}\label{ladist}
{\rm dist} \Big (  \mathcal{C}_m ( q^{(n, n+1)}_i ) , \mathbb{C} \setminus    \mathcal{C}_m ( q^{( n+1)}_i ) \Big )  \geq  \sin( \nu / 2 m )   \frac{1}{10}\rho_{n+ 1}.   
\end{align} 
and 
\begin{align}\label{ladistf}
{\rm dist} \Big (  \mathcal{C}_m ( q^{(n)}_i ) , \mathbb{C} \setminus    \mathcal{C}_m ( q^{(n, n+1)}_i ) \Big )  \geq  \sin( \nu  /  m )   \frac{1}{10}\rho_{n+ 1}.   
\end{align}
\end{lemma}
\begin{figure}[h]
\centering
\includegraphics[width=0.7\textwidth]{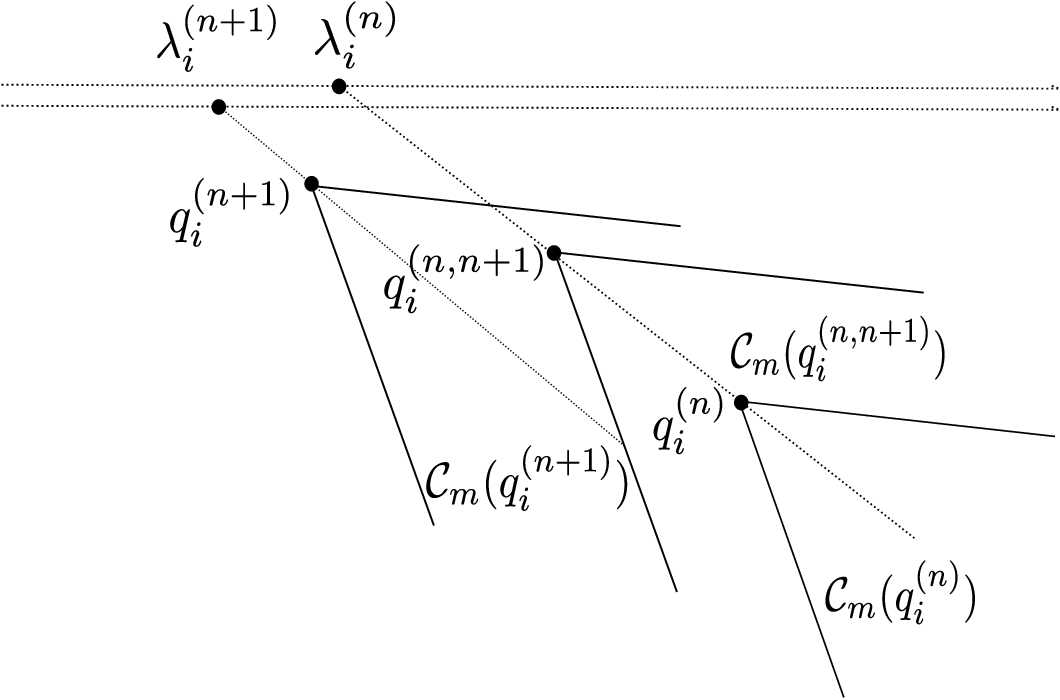}
  \caption{Cones }
    \label{fig:cones}       
\end{figure}
\begin{proof}
That $   \mathcal{C}_m ( q^{(n)}_i )  \subset \mathcal{C}_m ( q^{(n, n+1)}_i )  $ is immediate.
From Theorem \ref{thm:ind} (Property $\mathcal{P}1$) and Definition \ref{sequence} it follows that
\begin{align}\label{ter2}
| \lambda_i^{(n+1)} -  \lambda_i^{(n)} | \leq |g|  \big (\bold C^4 \rho_0^{ \mu  } \big )^{1/2}
\Big (  \bold (C^2 \rho^{\mu})^{n} \Big )^{1/2} \rho_{n}^{1+ \mu/2} \leq |g| \frac{1}{2 ^{n}}   \rho_{n}^{1+ \mu/2}.  
\end{align}
This and a geometric series argument  prove \eqref{ter0}.  We write 
\begin{align}
q_i^{(n, n+1)} = q_{i}^{(n+1)} + \xi_1 e^{-i \nu} + \xi_2 i e ^{-i\nu}. 
\end{align}
Eq.\ \eqref{ter2} implies that 
\begin{align}\label{ter3}
|\xi_2| \leq |g|  \rho_{n},   \hspace{2cm} \xi_1 \geq \Big (\frac{2}{5} - \frac{1}{100} - \frac{1}{4}  \Big )  \rho_{n+1}  -   |g|  \rho_{n}  >    \frac{1}{10}\rho_{n+ 1}.
\end{align}
The last step follows for $g>0$ sufficiently small (see Definition \ref{gzero}).
To prove that $  \mathcal{C}_m ( q^{(n, n+1)}_i )  \subset  \mathcal{C}_m ( q^{(n+ 1)}_i   )  $, it is enough to show that  $  q^{(n, n+1)}_i   \in \mathcal{C}_m ( q^{(n+ 1)}_i   )  $. We shall prove that
\begin{align}\label{ter4}
|\xi_2| /  \xi_1 < \tan(\nu / 2 m),
\end{align}
which holds true if  $ | g | \leq   \rho \frac{1}{10}  \sin(\nu / 2 m) \leq \rho \frac{1}{10} \tan(\nu / 2 m)    $. 

  Eq.\ \eqref{ladistf} is  is implied by the particular geometry considered here   because both cones have the same axis, see also Definition \ref{sequence}.

Eq. \eqref{ter4} implies that the angle between the axis of the cone $ \mathcal{C}_m ( q^{(n+ 1)}_i )     $ and the complex number  $ q_i^{(n, n+1)} - q_{i}^{(n+1)} $ is smaller than $  \nu / 2 m $ and, therefore, the angle between this complex number and the closest edge of the cone must be larger than $ \nu / 2 m $. Then, the distance  between the referred complex number and the edge is larger than 
$$ | q_i^{(n, n+1)} - q_{i}^{(n+1)} | \sin( \nu / 2 m )   \geq \xi_1  \sin( \nu / 2 m )   \geq  \sin( \nu / 2 m )   \frac{1}{10}\rho_{n+ 1},  $$ 
this implies \eqref{ladist}. 
\end{proof}
\begin{lemma} \label{mana0}
Assume  that for all 
$ z \in B^{(1)}_i \setminus \mathcal{C}_m( q_i^{(n)} )  $
\begin{align}
\Big \| \frac{1}{ H^{(n), \theta} - z}  \overline{P^{(n)}_i} \Big \| \leq \boldsymbol{C}^{n+1} \frac{1}{\rm dist \Big ( z ,   \mathcal{C}_m( q_i^{(n)} ) \Big )   }, 
\end{align}
then 
 \begin{align}\label{haha1}
\Big \| \frac{1}{   H^{(n), \theta}  + H_{f}^{(n, n+1), \theta}  - z}  \overline{P^{(n, n + 1)}_i} \Big \| \leq 4 \boldsymbol{C}^{n+1} \frac{1}{\rm dist \Big ( z,   \mathcal{C}_m( q_i^{(n, n+1 )} ) \Big )    }   ,  
\end{align}
and 
 \begin{align}\label{haha2}
\Big \| \frac{1}{   H^{(n), \theta}  + H_{f}^{(n, n+1), \theta}  - z}  H_{f}^{(n, n+1)}   \overline{P^{(n, n + 1)}_i }\Big \| 
\leq  \frac{100}{\sin(\nu/m)}  \boldsymbol{C}^{n+1} 
,  
\end{align} 
for all   $ z \in B^{(1)}_i \setminus \mathcal{C}_m( q_i^{(n, n+1)} )  $. 
\end{lemma}
\begin{proof}
  Take $ z \in B^{(1)}_i \setminus \mathcal{C}_m( q_i^{(n, n+1)} )  $. 
We use the spectral theorem  and that  (see \eqref{eq:lemma-is-Pident})
\begin{align}
\label{eq:lemma-is-Pident-R}
\overline{P^{(n)}_i}+P^{(n)}_i\otimes \overline{P_{\Omega^{(n,n+1)}}} 
= \overline{P^{(n,n+1)}_i}, 
\end{align}
to calculate
\begin{align}
\label{eq:is-keyest10-R}
&\norm{\frac{1}{H^{(n),\theta}+H_f^{(n,n+1),\theta}-z}\overline{P^{(n,n+1)}_i}}
\notag \\
&\leq \norm{\frac{1}{H^{(n),\theta}+H_f^{(n,n+1),\theta}-z}\overline{P^{(n)}_i}}
+\norm{\frac{1}{H^{(n),\theta}+H_f^{(n,n+1),\theta}-z}P^{(n)}_i\otimes \overline{P_{\Omega^{(n,n+1)}}}}
\notag \\
&=\sup_{s\in \{0\} \cup [\rho_{n+1},\infty)}
\norm{\frac{1}{H^{(n),\theta}-(z-e^{-\theta}s)}\overline{P^{(n)}_i}} 
+\sup_{s\in  [\rho_{n+1},\infty)}\frac{\norm{P^{(n)}_i}}{|\lambda^{(n)}_i-(z-e^{-\theta}s)|} .
\end{align} 
 Thanks to the geometry, for all $s \geq 0$, we have
\begin{align}\label{coco1}
{\rm dist} \Big ( z  -  e^{- \theta } s,   \mathcal{C}_m( q_i^{(n)} ) \Big ) \geq 
{\rm dist} \Big ( z,   \mathcal{C}_m( q_i^{(n)} ) \Big )   \geq {\rm dist} \Big ( z,   \mathcal{C}_m( q_i^{(n, n+1)} ) \Big ) .
\end{align}
Eq.\ \eqref{coco1}, our hypothesis, Lemma \ref{lemma:resestinAprima} and Definitions \ref{compile} and \ref{C} imply that for $  s \geq 0 $
\begin{align}\label{coco2}
\norm{\frac{1}{H^{(n),\theta}-(z-e^{-\theta}s)}\overline{P^{(n)}_i}} 
\leq  \boldsymbol{C}^{n+1} \frac{1}{\rm dist \Big ( z,   \mathcal{C}_m( q_i^{(n, n+1 )} ) \Big )    }. 
\end{align}
Notice that for $ s \geq \rho_{n+1}$, 
\begin{align}\label{coco3}
z  -   e^{ -  \theta} s  \notin \mathcal{C}_{m}(   \lambda_i^{(n)}  )
\end{align}
and, therefore, 
\begin{align}\label{coco4}
|\lambda^{(n)}_i-(z-e^{-\theta}s)| \geq {\rm dist} \Big ( z  -  e^{- \theta } s,   \mathcal{C}_m( \lambda_i^{(n)} ) \Big ) = & {\rm dist} \Big ( z ,  
 \mathcal{C}_m(\lambda_i^{(n)} +     e^{- \theta } s ) \Big ) 
 \\ \notag   \geq  & {\rm dist} \Big ( z ,  
 \mathcal{C}_m( q_i^{(n, n+1)}  ) \Big ). 
\end{align}
Eqs.\ \eqref{coco1}, \eqref{coco2} and \eqref{coco4},  and Lemma \ref{lemma:is-normproj} together with Definition \ref{C} imply Eq.\ \eqref{haha1}.

Now, we prove Eq.\ \eqref{haha2}. As in \eqref{eq:is-keyest10-R} and \eqref{coco2}, we have  
\begin{align}
\label{coco5}
& \norm{\frac{1}{H^{(n),\theta}+H_f^{(n,n+1),\theta}-z}  H_f^{(n,n+1)}  \overline{P^{(n,n+1)}_i}}
\notag \\
& \hspace{.5cm} \leq \sup_{s\in \{0\} \cup [\rho_{n+1},\infty)}\boldsymbol{C}^{n+1} \frac{s  }{{\rm dist} \Big ( z - e^{-\theta} s ,   \mathcal{C}_m(  q_i^{(n, n+1 )}  ) \Big )    } 
+\sup_{s\in  [\rho_{n+1},\infty)}\frac{\norm{P^{(n)}_i}  s}{|\lambda^{(n)}_i-(z-e^{-\theta}s)|}  .  
\end{align} 

Notice that
for $z \notin   \mathcal{C}_m(  q_i^{(n, n+1 )}  )  ,$
\begin{align}\label{pisa1}
{\rm dist} \Big ( z - e^{-\theta} s ,   \mathcal{C}_m(  q_i^{(n, n+1 )}  ) \Big ) \geq \frac{1}{2}  s \sin(\nu/m) . 
\end{align}
 Now we argue as in \eqref{coco4} and obtain, for $ s \geq \rho_{n+1} $, 
\begin{align}\label{pisa3}
|\lambda^{(n)}_i-(z-e^{-\theta}s)| \geq {\rm dist} \Big ( z  -  e^{- \theta } s,   \mathcal{C}_m( \lambda_i^{(n)} ) \Big ) = & {\rm dist} \Big ( z  - \frac{1}{10} e^{- \theta } s,  
 \mathcal{C}_m(\lambda_i^{(n)} +   \frac{9}{10}  e^{- \theta } s ) \Big ) 
 \\ \notag   \geq  & {\rm dist} \Big ( z    - \frac{1}{10} e^{- \theta } s ,  
 \mathcal{C}_m( q_i^{(n, n+1)}  ) \Big )
  \\ \notag   \geq  & \frac{1}{2 0 }  s \sin(\nu/m) . 
\end{align}
Eqs. \eqref{coco5}, \eqref{pisa1} and \eqref{pisa3} together with  Lemma \ref{lemma:is-normproj} imply Eq. \eqref{haha2}. 
\end{proof}
\begin{lemma}\label{coness}
Let $ \mathcal{C}^{(1)} ,  \mathcal{C}^{(2)},    \mathcal{C}^{(3)} $ be cones in $\C$, such that $C^{(1)}  \subsetneqq   \mathcal{C}^{(2)} \subsetneqq \mathcal{C}^{(3)} $, of the form \eqref{eq:defcone} - with the same $m$. Assume that  
\begin{align}
\max_{ y \in \partial   \mathcal{C}^{(2)}}{\rm dist} \Big ( y,   \mathcal{C}^{(1)}    \Big ) \leq \frac{1}{2}
{\rm dist} \Big ( \C \setminus \mathcal{C}^{(3)},   \mathcal{C}^{(1)}    \Big ).
\end{align}
Then, for every  $z \notin \mathcal{C}^{(3)}$: 
\begin{align}
{\rm dist} \Big ( z ,   \mathcal{C}^{(2)}    \Big ) \geq \frac{1}{2}{\rm dist} \Big ( z ,   \mathcal{C}^{(1)}    \Big ) .
\end{align}
\end{lemma}
\begin{proof}
We take $z \notin \mathcal{C}^{(3)}$, $y \in \partial  \mathcal{C}^{(2)}$,
and $x \in \mathcal{C}^{(1)}$  such that $ |y- x| = {\rm dist} \Big ( y,   \mathcal{C}^{(1)}    \Big )    $. We calculate
\begin{align}
|  z - y | \geq |z - x| - |x - y| & =  |z- x| -  {\rm dist} \Big ( y,   \mathcal{C}^{(1)}    \Big )
\\ & \geq  |z- x| - \frac{1}{2}
{\rm dist} \Big ( \C \setminus \mathcal{C}^{(3)},   \mathcal{C}^{(1)}    \Big )
. 
\end{align} 
Next, we use that 
\begin{align}
{\rm dist} \Big ( \C \setminus \mathcal{C}^{(3)},   \mathcal{C}^{(1)}    \Big )
\leq |z - x| 
\end{align}
to obtain:
\begin{align}
|  z - y | \geq  \frac{1}{2} |z- x|  \geq \frac{1}{2}{\rm dist} \Big ( z ,   \mathcal{C}^{(1)}    \Big ),
\end{align} 
and therefore, 
\begin{align}
{\rm dist} \Big ( z ,   \mathcal{C}^{(2)}    \Big ) \geq \frac{1}{2}{\rm dist} \Big ( z ,   \mathcal{C}^{(1)}    \Big ) .
\end{align}
\end{proof}
\begin{lemma}
\label{mana}
Assume that   $  |g| \leq \frac{ \sin(\nu/2 m)^3  \rho }{10^8}  $ and $\rho \leq 10^{-3} \sin(\nu/m)e_1 $   and that for all 
$ z \in B^{(1)}_i \setminus \mathcal{C}_m( q_i^{(n)} )  $
\begin{align}\label{difim1}
\Big \| \frac{1}{ H^{(n), \theta} - z}  \overline{P^{(n)}_i} \Big \| \leq \boldsymbol{C}^{n+1} \frac{1}{\rm dist \Big ( z,   \mathcal{C}_m( q_i^{(n)} ) \Big )   },  
\end{align}
then $ \Big ( B^{(1)}_i \setminus  \mathcal{C}_{m}(q_i^{( n+1)})
\Big ) \setminus \{ \lambda_i^{(n + 1)} \} $ 
  is contained in the resolvent set of $  H^{(n+1), \theta} $ and 
\begin{align}\label{difim2}
\Big \|  \frac{ 1 }{ H^{(n+ 1), \theta}  - z } \overline{P_i^{(n+1)}}\Big \|  \leq &  \frac{ 10^5}{   \sin(\nu/m)^2   }  \boldsymbol{C}^{n+1} \frac{1}{\rm dist \Big ( z,   \mathcal{C}_m( q_i^{(n, n+1 )} ) \Big ) .   },
\end{align}
for every $ z \in  B^{(1)}_i \setminus  \mathcal{C}_{m}(q_i^{( n+1)})  $.  
Moreover, assuming that $ \boldsymbol{C} \geq   \frac{10^5}{ \sin(\nu/m)^2 }$, 
for all 
$ z \in B^{(1)}_i \setminus \mathcal{C}_m( q_i^{(n + 1)} )  $
\begin{align}\label{difim3}
\Big \| \frac{1}{ H^{(n + 1), \theta} - z}  \overline{P^{(n)}_i} \Big \| \leq \boldsymbol{C}^{n+2} \frac{1}{\rm dist \Big ( z,   \mathcal{C}_m( q_i^{(n + 1)} ) \Big )   }.   
\end{align}
\end{lemma}
\begin{proof}
Eq.\ \eqref{difim3} is a consequence of \eqref{difim2} and \eqref{ter1} together with $\boldsymbol{C} \geq  \frac{10^5}{ \sin(\nu/m)^2 } $  . 
We fix the cones:
\begin{align}
\mathcal{C}^{(1)} & = \mathcal{C}_m(q_i^{(n, n+1)}), \hspace{.5cm}
\mathcal{C}^{(2)}  = \mathcal{C}_m(q_i^{(n, n+1)} - \rho_{n+ 1} e^{-i\nu}), 
\hspace{.5cm}  \\ \notag
\mathcal{C}^{(3)} & = \mathcal{C}_m(q_i^{(n, n+1)} - 2 \frac{1}{\sin(\nu/m)} \rho_{n+ 1} e^{-i\nu}).
\end{align}
Note that the cones we just defined fulfill  the hypothesis of Lemma \ref{coness}. They
satisfy the following properties (see Lemma \ref{coness}). For all $z \notin C^{(3)}$ and for all $s \geq 0$: 
\begin{align}\label{perra0}
\lambda_i^{(n)} \in \mathcal{C}^{(2)}, \hspace{.5cm} |(z - s e^{-\theta} ) - \lambda_i^{(n)}|  \geq 
{\rm dist}( (z - s e^{-\theta} ) , \mathcal{C}^{(2)}) \geq {\rm dist}( z   , \mathcal{C}^{(2)}) \geq \frac{ 1 }{2 } {\rm dist}(z, \mathcal{C}^{(1)}), 
\end{align}
where we use that $ z - s e^{-\theta} \notin C^{(3)}  $. 
We define $ z_1 = x^{(1)}_1 + i x^{(1)}_2 $ ($x^{(1)}_1,  x^{(1)}_2 \in \mathbb{R} $) to be the point in  the intersection of $ q_{i}^{(n, n+ 1)} + \frac{1}{100}
\rho_{n+1} e^{-i\nu} + \R $ with $ \partial  \mathcal{C}^{(1)}$ with smaller $x^{(1)}_1$,  and  similarly, $z_2 =x^{(2)}_1 + i x^{(2)}_2  $  the point in the intersection of $ q_{i}^{(n, n+ 1)} + \frac{1}{100}
\rho_{n+1} e^{-i\nu} +  i \R $ with $\mathcal{C}^{(1)}$ with bigger $x^{(2)}_2$ . We recall that 
\begin{align}\label{perra1}
  q_{i}^{(n, n+ 1)} + \frac{1}{100}
\rho_{n+1} e^{-i\nu} =  \lambda_i^{(n)} +  \frac{2}{5}
\rho_{n+1} e^{-i\nu}  \in M^{(n)}_i,
\end{align}
see Definition \ref{def:regionM}, and therefore, 
\begin{align}\label{perra2}
 z_1,  z_2   \in M^{(n)}_i   
\end{align} 
 (the factor $\frac{1}{100}$ is chosen for this reason).   
Now, we set
\begin{align}\label{perra3}
\mathcal{U} : = \Big ( \overline{\mathcal{C}^{(3)} \setminus \mathcal{C}^{(1)}} \Big ) \cap
\bigcup_{t \in [0,1]}  \big \{  t z_1 + (1- t)z_2 + e^{-i \nu} \mathbb{R} \big \}. 
\end{align}
Our restrictions on $\rho$ together with \eqref{perra1} and \eqref{perra2} imply that 
\begin{align}\label{perra4}
\mathcal{U} \subset  M^{(n)}_i  .
\end{align}
It follow from the particular considered  geometry at hand that the distance between the boundary of 
$ \mathcal{U} $ and $\lambda_i^{(n)}$ is bigger or equal than the distance between the point $ z_2 $ and the line $q_i^{(n, n+1)} + \mathbb{R} e^{- i \nu} $, which equals $\tan( \nu/m) \iota $, where $ \iota $ is the distance between $ q_i^{(n, n+1)}  $
and  the intersection of the line $  z_2 + i e^{-i \nu}\R $ with $q_i^{(n, n+1)} + \mathbb{R} e^{- i \nu} $. Then $\iota$ is  bigger or equal to  the  distance between $ q_i^{(n, n+1)}  $ and  the intersection of the line $  \tilde z_2 + i e^{-i \nu}\R $ with $q_i^{(n, n+1)} + \mathbb{R} e^{- i \nu} $, where $  \tilde z_2  $
is the intersection of $  z_2 + i \R $ with $  q_i^{(n, n+1)}  + \R  $. Then, $\iota \geq \frac{1}{100}\rho_{n+1} \cos(\nu)\cos(\nu)  $. We obtain that
\begin{align}\label{distu}
{\rm dist}(\partial \mathcal U, \lambda_i^{(n)} ) \geq \frac{1}{100}\rho_{n+1} \cos(\nu)\cos(\nu) \tan(\nu/m) \geq \frac{1}{200} \sin(\nu/m) \rho_{n+1},
\end{align}
where we use that $\theta \in \mathcal S $.  

For every $ z \in  \overline{  \mathcal{C}^{(3)} \setminus \mathcal{C}^{(1)}  \setminus \mathcal{U}  }$ and $s \geq 0$, we have that
\begin{align}\label{perra5}
| \lambda_{i}^{(n)} - (z - s e^{-\theta} )  | \geq \frac{1}{200} \sin(\nu/m)\rho_{n+1}
\end{align}
and 
\begin{equation}\label{perra6}
{\rm dist}( z   ,  \mathcal{C}^{(1)}   ) \leq 2 \frac{ 1}{ \sin(\nu/m) }\rho_{n+1}.  
\end{equation}
It follows form  \eqref{perra5}, \eqref{perra6} together with \eqref{perra0} that  
\begin{align}\label{perra7}
\frac{  {\rm dist}( z   ,  \mathcal{C}^{(1)}   )    }{| \lambda_{i}^{(n)} - (z - s e^{-\theta} )  |} \leq  400 \frac{1}{  \sin(\nu/m)^2 },
\end{align}
for every $  z \in  \Big (B_i^{(1)} \setminus  \mathcal C^{(1)} \Big )  \setminus \mathcal{U}   $.  
This implies, we also use Lemma \ref{mana0} and the spectral theorem  that
(actually we only need $s = 0 $ above), 
 \begin{align}\label{perra8}
\Big \| \frac{1}{   H^{(n), \theta}  + H_{f}^{(n, n+1), \theta}  - z}   \Big \| \leq   \frac{ 10^4}{   \sin(\nu/m)^2    }  \boldsymbol{C}^{n+1} \frac{1}{\rm dist \Big ( z,   \mathcal{C}_m( q_i^{(n, n+1 )} ) \Big )    }   ,  
\end{align}
and for every positive number $r$ 
 \begin{align}\label{perra9}
\Big \| \frac{1}{   H^{(n), \theta}  + H_{f}^{(n, n+1), \theta}  - z} 
( H_{f}^{(n, n+1)} + r) \Big \| 
\leq &  \frac{ 100}{\sin(\nu/m)}  \boldsymbol{C}^{n+1}  \\ \notag & +  \frac{ 10^4}{   \sin(\nu/m)^2    }  \boldsymbol{C}^{n+1} \frac{ r }{\rm dist \Big ( z,   \mathcal{C}_m( q_i^{(n, n+1 )} ) \Big )    }
,  
\end{align} 
where we use  that $ H_f^{(n, n+1)}  P_i^{(n, n+1)} = 0  $,
for every $  z \in  \Big (B^{(1)}_i \setminus  \mathcal C^{(1)} \Big )  \setminus \mathcal{U}   $. Choosing $r = \rho_{n+1} $, and additionally, $z \notin   
\mathcal{C}_m(q_i^{(n+1)}) $, we get from \eqref{ladist} and \eqref{perra9} that
 \begin{align}\label{perra10}
\Big \| \frac{1}{   H^{(n), \theta}  + H_{f}^{(n, n+1), \theta}  - z} 
( H_{f}^{(n, n+1)} + r) \Big \| 
\leq &     \boldsymbol{C}^{n+1}    \frac{10^6}{\sin(\nu/2m)^3  } 
. 
\end{align} 
We observe that
\begin{align}
\label{difi1}
&\norm{V^{(n,n+1) ,\theta}\frac{1}{H^{(n),\theta}+H_f^{(n,n+1),\theta}-z}}
\leq \norm{V^{(n,n+1) ,\theta}\frac{1}{H_f^{(n,n+1)}+r}}
\notag \\
& \hspace{6.4cm} \times \norm{\left(H_f^{(n,n+1)}+r \right)\frac{1}{H^{(n),\theta}+H_f^{(n,n+1),\theta}-z}} .
\end{align}

Then, we have  (see also Eq.\ \eqref{eq:keyest22'})
\begin{align}
\label{difi4}
&\big \| g V^{(n,n+1) ,\theta}\frac{1}{H^{(n),\theta}+H_f^{(n,n+1),\theta}-z} \Big \|
\leq   |g|  \frac{10^6}{\sin(\nu/2m)^3  }  \boldsymbol{C}^{n+1}\norm{V^{(n,n+1) ,\theta}\frac{1}{H_f^{(n,n+1)}+r}} 
\notag \\   
&\leq  |g|  \frac{10^6}{\sin(\nu/2m)^3  }   \boldsymbol{C}^{n+1}   |e^{ - \theta (1+ \mu)} |   \sqrt{4\pi}  \left( \frac{\rho_n}{r}  +2\sqrt{ \frac{\rho_n}{r}}\right)\rho_n^{\mu} 
\leq      \frac{  10^8 |g| }{ 2 \sin(\nu/2 m)^3 \rho }  \boldsymbol{C}^{n+1}    \rho_n^{\mu}
\leq \frac{1}{2}, 
\end{align} 
because Definition \ref{sequence} implies that $  \boldsymbol{C}^{n+1}    \rho_n^{\mu} \leq 1 $ (we use as well our restrictions in $|g|$).

Eq.\ \eqref{difi4} and a Neumann series argument implies that  
$ \Big (B_i^{(1)} \setminus \mathcal{C}_m(q_i^{(n+1)}) \Big ) \setminus \mathcal{U}  $  is contained in  the
resolvent set of $  H^{(n+1), \theta} $ and for all $z$ in this set (see  also \eqref{perra8})
\begin{align}\label{difi5}
\Big \|   \frac{1}{H^{(n + 1), \theta} - z}  \Big \| \leq & 2 \frac{ 10^4}{   \sin(\nu/m)^2    }  \boldsymbol{C}^{n+1} \frac{1}{\rm dist \Big ( z,   \mathcal{C}_m( q_i^{(n, n+1 )} ) \Big )    }. 
\end{align} 
Lemma \ref{prop:is-resest-fin}  ensures that $\lambda^{(n+1)}_i$ is the only spectral point of $H^{(n+ 1), \theta} $ in  $ M^{(n)}_i $. Hence,  the function
\begin{align}\label{difi6}
\mathcal U \ni z \mapsto \frac{ 1 }{ H^{(n+ 1), \theta}  - z } \overline{P_i^{(n+1)}}
\end{align}
is analytic. The maximum modulus principle implies that it attains its maximum on the boundary of $\mathcal U$, then we have (see  Definition \ref{C} and Lemma \ref{prop:is-resest-fin}) 
\begin{align} \label{difi7}
\Big \|  \frac{ 1 }{ H^{(n+ 1), \theta}  - z } \overline{P_i^{(n+1)}}\Big \| \leq &
\boldsymbol{C}^{n+1} \frac{1}{\rho_{n+1}}
\end{align} 
for every $z \in \mathcal{U}$. 
Next, notice that, for  $z \in \mathcal{U},$  $ {\rm dist} \Big ( z,   \mathcal{C}_m( q_i^{(n, n+1 )} ) \Big ) \leq\frac{2}{\sin(\nu/m)}  \rho_{n+1}  $.
Then, we obtain 
\begin{align} \label{difi8}
\Big \|  \frac{ 1 }{ H^{(n+ 1), \theta}  - z } \overline{P_i^{(n+1)}}\Big \|  \leq &  
 \frac{1}{\sin(\nu/m)}  \boldsymbol{C}^{n+1} \frac{2}{ {\rm dist} \Big ( z,   \mathcal{C}_m( q_i^{(n, n+1 )} ) \Big )   }  \\ \notag \leq &  \frac{1}{\sin(\nu/m)}  \boldsymbol{C}^{n+1} \frac{2}{ {\rm dist} \Big ( z,   \mathcal{C}_m( q_i^{( n+1 )} ) \Big )   },
\end{align} 
Eqs.\ \eqref{difi5} and \eqref{difi8}  together with Lemma \ref{lemma:is-normproj} imply the desired  result. 
\end{proof}
The next theorem is proved inductively using Corollary \ref{coro:ib-res} and Lemma \ref{mana}. This is the main theorem of the present subsection. 
\begin{theorem}
\label{ResPrinc}
Assume that   $  |g| \leq \frac{ \sin(\nu/2 m)^3  \rho }{10^8}  $, $\rho \leq 10^{-3} \sin(\nu/m) e_1 $    and $ \boldsymbol{C} \geq \frac{10^5}{ \sin(\nu/m)^2 } $. Then, for all  $n\in\N$   
and for all 
$ z \in B^{(1)}_i \setminus \mathcal{C}_m( q_i^{(n )} )  $:
\begin{align}\label{difim3prima}
\Big \| \frac{1}{ H^{(n ), \theta} - z}  \overline{P^{(n)}_i} \Big \| \leq \boldsymbol{C}^{n+1} \frac{1}{\rm dist \Big ( z,   \mathcal{C}_m( q_i^{(n )} ) \Big )   }.  
\end{align}
\end{theorem}

\subsection{Resolvent estimates}
\label{sec:resest100}

 In this section we study the spectrum and resolvent of the full Hamiltonian, it is carried out in Theorems \ref{resuno} and \ref{resdos},  using Theorem \ref{ResPrinc}, in a similar manner as in Section \ref{resspec}.  First we add the full free energy to the hamiltonian at step $n$ in Lemma \ref{EL}, and  then, we add the full interacting energy, using Lemma \ref{lemma:is-keyesttilde}, in  Theorems \ref{resuno} and \ref{resdos}. These theorems imply Theorems \ref{resolventestimates} and \ref{spectralestimates}.

In this section we assume, in addition to Definitions \ref{C} \ref{sequence} and \ref{gzero} (and  $ \boldsymbol{C} \geq \boldsymbol{D} \sin(\nu/m)^{-1} $), that   
\begin{align}\label{ggg}
 |g| \leq \frac{ \sin(\nu/2 m)^3  \rho }{10^8}, \hspace{1cm} \rho \leq 10^{-3} \sin(\nu/m), \hspace{1cm} \boldsymbol{C} \geq   \frac{10^5}{ \sin(\nu/m)^2 } .     
\end{align}
\begin{lemma}
\label{mmta}
 Let  $z\notin \mathcal{C}_m(\lambda_i^{(n)})  $ and   $ 0 \leq r  \leq  |z-
\lambda_i^{(n)}| $, $s \geq 0$. It follows that
\begin{align}\label{nose1}
\Big |  \frac{ s+r}{ {\rm dist } \Big ( z  - e^{-\theta}s,  \mathcal{C}_m(\lambda_i^{(n)})  \Big )} \Big |
& \leq 2 \frac{  1   }{  \sin     (\nu/m) }
+ \frac{r}{    {\rm dist } \Big ( z ,  \mathcal{C}_m(\lambda_i^{(n)})  \Big )  }
, \\ \notag 
\Big |  \frac{ s+r}{ e^{-\theta} s  -\left(z-
\lambda_i^{(n)}\right)} \Big |
& \leq \frac{  6  }{  \sin    (\nu/m ) }.
\end{align} 
\end{lemma}
\begin{proof}
We use coordinates in $\mathbb{C} \equiv \mathbb{R}^2$ with origin at $ \lambda^{(n)}_{i}  $, the first coordinate axis with direction $ e^{-i \nu} $ and the second coordinate axis with direction $   i e^{-i \nu}  $. Notice that 
for every point $ z = \lambda^{(n)}_{i} +   \xi_1 e^{- i\nu} +  \xi_2 i  e^{- i\nu} \notin \mathcal{C}_m(\lambda_i^{(n)})   $ and every $ s \geq 0 $, the following facts are  implied by the considered geometry:
\begin{align}
\xi_1 \leq 0 & \Longrightarrow \Big | \lambda_i^{(n)}  - (  z-  s e^{- \theta} )           
\Big | \geq  |  \lambda_i^{(n)} -z  \Big |, \label{pit} \\ 
\xi_1 > 0  &  \Longrightarrow  | \xi_2 | \geq    | \xi_1 |  \tan (\nu/  m     ) .   \label{pito}
\end{align}  
Eq.\ \eqref{pito} implies that for $ \xi_1 > 0 $
\begin{align} \label{pitos}
 | z  -  \lambda_i^{(n)} | \leq |\xi_2| \sqrt{ 1 + \tan(\nu / m)^{-2}}, 
\end{align}
 and because $  | \lambda_i^{(n)}  - (  z-  s e^{- \theta} )           
\Big |  \geq |\xi_2|   $, we obtain that (we also use \eqref{pit}) 
\begin{align}\label{pitosn}
\frac{2}{ \sin(\nu /m)} | \lambda_i^{(n)}  - (  z-  s e^{- \theta} )           
\Big | \geq  |  \lambda_i^{(n)} -z  \Big |, 
\end{align}
for every $z\notin \mathcal{C}_m(\lambda_i^{(n)})  $ and every $ s \geq 0  $. 
 Take  $z\notin \mathcal{C}_m(\lambda_i^{(n)})  $ to obtain  
\begin{align}
\label{eq:helplemma111}
 \left| \frac{s+r}{e^{-\theta}s+\lambda_i^{(n)}-z}    \right|
&\leq    |e^{\theta}| +  |e^{\theta}|    \left| \frac{e^{-\theta}r-\lambda_i^{(n)}+z}{e^{-\theta}s+\lambda_i^{(n)}-z}   \right|
\leq   |e^{\theta}|  +  (1 +   |e^{\theta}| )  \frac{\left| z- \lambda_i^{(n)}\right| }{\left| e^{-\theta}s+\lambda_i^{(n)}-z\right| }  
\notag \\
&\leq  |e^{\theta}|  + \frac{  2 (1 +   |e^{\theta}| )  }{  \sin (  \nu/ m )}\leq  \frac{  6  }{  \sin ( \nu/ m)},
\end{align}
which proves the second inequality in \eqref{nose1}. The first inequality  of the claim is again ensured thanks to our considered geometry that implies:
\begin{align}
{\rm dist } \Big ( z  - e^{-\theta}s,  \mathcal{C}_m(\lambda_i^{(n)})  \Big ) \geq \max \Big [ \frac{\sin(\nu/ m)}{2} s ,  {\rm dist } \Big ( z ,  \mathcal{C}_m(\lambda_i^{(n)})  \Big )  \Big ]  ;
\end{align}
recall that $\theta \in \mathcal{S}$. 
\end{proof}

\begin{lemma} \label{EL}
\label{mmtaa}
 For $i=0,1$, the set  $   B^{(1)}_i \setminus \mathcal C_m(\lambda_i^{(n)} )  $ 
 is contained  in the resolvent set  of
 $  \tilde H^{(n), \theta} $ and  for all $z  $ this set:
\begin{align}
\norm{ \frac{1}{\tilde H^{(n),\theta}-z}}
\leq   4
  \boldsymbol{   C}^{n+1}  \frac{1}{ {\rm dist} ( z, \mathcal{C}_m(\lambda_i^{(n)})       )}. 
\end{align}
\end{lemma}
\begin{proof}
The spectral theorem, Lemma \ref{lemma:resestinAprima}, Definitions \ref{compile} and \ref{C},  and Theorem \ref{ResPrinc} imply that (we also use Lemma \ref{lemma:is-normproj}, which is valid for every $n$ because Theorem 
\ref{thm:ind} is proved above)
\begin{align}
\norm{ \frac{1}{\tilde H^{(n),\theta}-z}}
= &  \sup_{ r  \geq  0} \norm{ \frac{1}{ H^{(n),\theta}  + e^{-\theta } r -z} 
\big ( P_i^{(n)}  + \overline{P_i^{(n)}} \big) }  \\ \notag \leq &  
\boldsymbol{   C}^{n+1}  \frac{1}{ {\rm dist} ( z-  e^{-\theta } r, \mathcal{C}_m(q_i^{(n)})       )}  + 3 \frac{1}{  | \lambda_i^{(n)} + e^{-\theta } y - z   | }
\\ \notag \leq &  4
  \boldsymbol{   C}^{n+1}  \frac{1}{ {\rm dist} ( z-  e^{-\theta } r, \mathcal{C}_m(\lambda_i^{(n)})       )}
  \leq  4
  \boldsymbol{   C}^{n+1}  \frac{1}{ {\rm dist} ( z, \mathcal{C}_m(\lambda_i^{(n)})       )}, 
\end{align}
were we use that $\lambda_i^{(n)} \in  \mathcal{C}_m(\lambda_i^{(n)}) $  and the geometrical  fact that $   {\rm dist} ( z-  e^{-\theta } r, \mathcal{C}_m(\lambda_i^{(n)})       ) \geq  {\rm dist} ( z, \mathcal{C}_m(\lambda_i^{(n)})       )    $. 
\end{proof} 
\begin{lemma}
\label{lemma:is-keyesttilde}
 For every     $ z \in   B^{(1)}_i \setminus \mathcal C_m(\lambda_i^{(n)} -  e^{-i \nu} \rho_n^{1+ \mu/4} )  $ 
the following inequality holds true 
\begin{align}
\norm{V^{(n,\infty),\theta}\frac{1}{\tilde H^{(n),\theta}-z}}
\leq  \frac{10^3}{\sin (  \nu/m)^2 }\bold C^{n+1}  \rho_n^{\frac{3\mu}{4}}  .
\end{align}  
\end{lemma}
\begin{proof} 
 We compute (see Remark \ref{R}):
\begin{align}
\label{eq:is-keyesttilde0}
\norm{V^{(n,\infty),\theta}\frac{1}{\tilde H^{(n),\theta}-z}}
  \leq &       \norm{V^{(n,\infty),\theta}\frac{1}{ H^{(n,\infty )}_f+r}} 
 \norm{( H_f^{(n, \infty )}  + r   )  \frac{1}{\tilde H^{(n),\theta}-z}}  
\end{align}  
where  we take 
\begin{align}
r  =    {\rm dist } \Big (  \mathbb C  \setminus \mathcal{C}_m(\lambda_i^{(n)} - e^{- i \nu}  \rho_n^{1+ \mu/4} )  ,  \mathcal{C}_m(\lambda_i^{(n)})  \Big )
=  \sin(\nu/m) \rho_n^{1 + \mu/4}  \ \leq \left|  z- \lambda_i^{(n)}\right|.
\end{align}

As in Eq.\ \eqref{eq:keyest22'} we obtain
\begin{align}
\label{teto1}
\norm{V^{(n,\infty) ,\theta}\frac{1}{H_f^{(n, \infty )}+r}}   
&\leq    |e^{ - \theta (1+ \mu)} |   \sqrt{4\pi}  \left(  \frac{\rho_n}{r}  +2\sqrt{ \frac{\rho_n}{r}}\right)\rho_n^{\mu}
 \leq  \frac{50}{\sin(\nu/m)}   \rho_n^{ \frac{3\mu}{4}} .  
\end{align}
The spectral theorem and Theorem \ref{ResPrinc} and Lemma  \ref{mmta}  imply that
\begin{align}\label{teto2}
 \norm{\left( H_f^{(n, \infty ),\theta}   +r   \right)\frac{1}{\tilde H^{(n),\theta}-z}}  \leq & 4
  \boldsymbol{   C}^{n+1}  \sup_{s \geq 0}
 \Big |  \frac{ s+r}{ {\rm dist } \Big ( z  - e^{-\theta}s,  \mathcal{C}_m(\lambda_i^{(n)})  \Big )} \Big |
\\ \notag \leq & 4
  \boldsymbol{   C}^{n+1}  \left[   \frac{  2   }{  \sin    (\nu/m) }
+ \frac{r}{    {\rm dist } \Big ( z ,  \mathcal{C}_m(\lambda_i^{(n)})  \Big )  }
\right] 
\\ \notag \leq &  4
  \boldsymbol{   C}^{n+1}  \left[  \frac{  2   }{  \sin     (\nu/m) }
+ \frac{r}{   \sin(\nu/m) \rho_n^{1+ \mu/4}  } \right] \leq \frac{ 12  \boldsymbol{   C}^{n+1}  }{  \sin     (\nu/m) } .
\end{align}
We conclude the desired result by \eqref{eq:is-keyesttilde0} together with \eqref{teto1} and \eqref{teto2}.
\end{proof}
\begin{theorem} \label{resuno}
  The set  $  \in   B^{(1)}_i \setminus \mathcal C_m(\lambda_i^{(n)} -  e^{-i \nu} \rho_n^{1+ \mu/4} )  $
 is contained in the resolvent set  of
 $   H^{ \theta} $ and  for all $z\in  B^{(1)}_i \setminus \mathcal C_m(\lambda_i^{(n)} -  e^{-i \nu} \rho_n^{1+ \mu/4} )  $:
\begin{align}
\norm{ \frac{1}{ H^{\theta}-z}}
\leq    8
  \boldsymbol{   C}^{n+1}  \frac{1}{ {\rm dist} ( z, \mathcal{C}_m(\lambda_i^{(n)})       )},
\end{align} 
and  
\begin{align}
\norm{ \frac{1}{ H^{\theta}-z}     -   \frac{1}{ \tilde H^{(n),\theta}-z}   }
\leq   | g | \frac{ 10^5}{\sin ( \nu/m)^2 }\bold C^{2n+2} 
   \rho_n^{\frac{3\mu}{4}}  \frac{1}{ {\rm dist} ( z, \mathcal{C}_m(\lambda_i^{(n)})       )}.
\end{align} 
\end{theorem}
\begin{proof}
The result is a consequence of Neumann series and Lemmas \ref{EL} and \ref{lemma:is-keyesttilde}  and Eq.\ \eqref{ggg}.  Notice that our assumptions on $\boldsymbol{C}$ in Definition  \ref{C} imply that  $  \bold C^{n+1}  \rho_n^{\frac{3\mu}{4}}  \leq 1  $. 
\end{proof}
\begin{theorem}\label{resdos}
 For every $n \in \mathbb{N}$ and $i\in\{0,1\}$,
\begin{align} \label{lal}
\mathcal{C}_m (  \lambda_i^{(n+ 1)} - \rho_{n+ 1}^{1+ \mu/4}  e^{-i \nu} )  \subset \mathcal{C}_m (  \lambda_i^{(n)} - \rho_n^{1+ \mu/4} e^{-i \nu} )  \subset \mathcal{C}_m (  \lambda_i - 2 \rho_n^{1+ \mu/4} e^{-i \nu}) , 
\end{align} 
and thus, $ B^{(1)}_i \setminus \mathcal{C}_m (  \lambda_i - 2 \rho_n^{1+ \mu/4} e^{-i \nu } ) $ is contained in the resolvent set of $H^{\theta}$ (see Theorem \ref{resuno}).   
  Moreover, $   B^{(1)}_i\setminus \mathcal C_m \left(\lambda_i\right)  $ is contained in the resolvent set of $H^{\theta}$. Additionally, the following estimate holds true:  
\begin{align} 
\norm{ \frac{1}{ H^{\theta}-z}}
 \leq  16
  \boldsymbol{   C}^{n+1}  \frac{1}{ {\rm dist} ( z, \mathcal{C}_m(\lambda_i )       )}  , \hspace{.5cm} \forall z  \in   \mathcal{C}_m (  \lambda_i -  2 \rho_n^{1+ \mu/4} e^{- i \nu} )  . \label{segunda}
\end{align}
\end{theorem}
\begin{proof}
It follows from Eq.\ \eqref{ter2} that
\begin{align}\label{ter2pe}
| \lambda_i^{(n+1)} -  \lambda_i^{(n)} | \leq |g|   \rho_{n}^{1+ \mu/2}
\end{align}
holds true.
We write, for $ \xi_1,  \xi_2 \in \mathbb{R} $,  
\begin{align}
\lambda_i^{(n+1)} -  \rho_{n+ 1}^{1+ \mu/4}  e^{-i \nu}   = \lambda_{i}^{(n)}  - 
\rho_{n}^{1+ \mu/4}  e^{-i \nu}  + \xi_1 e^{-i \nu} + \xi_2 i e ^{-i\nu}. 
\end{align}
Eq.\ \eqref{ter2pe} implies that 
\begin{align}\label{ter3pe}
|\xi_2| \leq |g|  \rho_{n}^{1+ \mu/2},   \hspace{2cm} \xi_1 \geq  \rho_{n}^{1+ \mu/4}   -    \rho_{n + 1}^{1+ \mu/4}   -   |g|  \rho_{n}^{1+ \mu/2}  > \frac{1}{2}  \rho_{n}^{1+ \mu/4} ,  
\end{align}
see Definition  \ref{sequence} and Definition  \ref{gzero} (or \eqref{ggg} -  notice that 
$  |g|  \rho_{n}^{1+ \mu/2}   \leq  \frac{|g|}{\rho}  \rho_{n + 1}^{1+ \mu/4} \rho_0^{\mu/4}  $). 
To prove the first  assertion in  \eqref{lal} it is enough to prove that   $  \lambda^{(n+1)}_i - \rho_{n+1}^{1+ \mu/4} e^{-i\nu}   \in \mathcal{C}_m ( \lambda^{(n)}_i  -  \rho_{n}^{1+\mu/4} e^{-i\nu}  )  $.  Note that since $ | g | \leq    \frac{1}{2}  \sin(\nu /  m) \leq  \frac{1}{2} \tan(\nu / m)    $ (which is verified by \eqref{ggg}), we have
\begin{align}\label{ter4pe}
|\xi_2| /  \xi_1 < \tan(\nu /  m).
\end{align}
 This proves the first  assertion in \eqref{lal}.

The first part of Eq.\ \eqref{lal} implies that, for all $n$,  
\begin{align}\label{ter5pe}
\mathbb{C} \setminus \mathcal{C}_m (  \lambda_i^{(n)} - \rho_{n}^{1+ \mu/4} e^{-i \nu} )  \subset \mathbb{C} \setminus \mathcal{C}_m (  \lambda_i^{(n + 1)} - \rho_{n+1}^{1+ \mu/4}e^{-i \nu}) 
\end{align}
 and
\begin{align}\label{ter6pe}
 \bigcup_n \mathbb{C} \setminus \mathcal{C}_m (  \lambda_i^{(n)} - \rho_{n}^{1+ \mu/4}e^{-i \nu} ) = \mathbb{C} \setminus \mathcal{C}_m (   \lambda_i) 
\end{align}
belongs to the resolvent set of $H^{\theta}$, see  Theorem \ref{resuno}. 

In a similar fashion  as above we prove that 
\begin{align} \label{lal1}
\mathcal{C}_m (  \lambda_i^{(n)} - \rho_{n}^{1+ \mu/4}e^{-i \nu} )  \subset \mathcal{C}_m (  \lambda_i - 2 \rho_n^{1+ \mu/4} e^{-i \nu})  ,
\end{align}
using \eqref{ter0}.  For every $ z \notin  \mathcal{C}_m (  \lambda_i - 2 \rho_n^{1+ \mu/4}e^{-i \nu})$ and $a \in  \mathcal{C}_m (  \lambda^{(n)}_i  ) $  we  know that (see \eqref{ter0}) 
\begin{align}\label{casa1} 
{\rm dist}(z,   \mathcal{C}_m (  \lambda_i   ) )\leq  &
{\rm dist}(z,   a  )   +
{\rm dist}(   a  ,   \mathcal{C}_m (  \lambda_i )  ) 
\leq      {\rm dist}(z,   a  ) +  2 |g| \rho_n^{1+ \mu/2},  
\end{align}
 and hence, we obtain (see Eq. \eqref{ggg}) 
\begin{align}\label{casa2} 
{\rm dist}(z,   \mathcal{C}_m (  \lambda_i   ) )\leq {\rm dist}(z,   \mathcal{C}_m (  \lambda^{(n)}_i  )    ) +   \sin(\nu/m) \rho_n^{1+ \mu/4}.   
\end{align}
Moreover  (see \eqref{lal1}),  
\begin{align}\label{casa3}
{\rm dist}(z,   \mathcal{C}_m (  \lambda^{(n)}_i  )    )  \geq 
{\rm dist}\Big ( \mathbb{C} \setminus   \mathcal{C}_m (  \lambda_i^{(n)} - \rho_{n}^{1+ \mu/4} e^{-i \nu} )  ,   \mathcal{C}_m (  \lambda^{(n)}_i  )    \Big ) \geq  \sin(\nu/m)  \rho_n^{1+ \mu/4} . 
\end{align}
Then, it follows that 
\begin{align}\label{casa2prima} 
\frac{{\rm dist}(z,   \mathcal{C}_m (  \lambda_i   ) )}{  {\rm dist}(z,   \mathcal{C}_m (  \lambda^{(n)}_i  )    ) }\leq  \frac{ {\rm dist}(z,   \mathcal{C}_m (  \lambda^{(n)}_i  )    )}{   {\rm dist}(z,   \mathcal{C}_m (  \lambda^{(n)}_i  )    ) } +   \frac{ \sin(\nu/m) \rho_n^{1+ \mu/4}}{  \sin(\nu/m) \rho_n^{1+ \mu/4}  } \leq 2.   
\end{align}
This and Theorem \ref{resuno} implies Eq.\ \eqref{segunda}. 
\end{proof}
We end this section with  a remark about the convergence rate of the projections 
$ P^{n}_i\otimes P_{\Omega^{(n,\infty)}}  $ that will be used in a forthcoming paper. 
\begin{remark}\label{scattering}
 It follows from Theorem \ref{thm:ind}, Property ($\mathcal P 3$), that
\begin{align}
\label{eq:P3prima}
\norm{P_i - P^{(n)}_i\otimes P_{\Omega^{(n,\infty)}}} \leq 2 \frac{|g|}{\rho}  \frac{1}{2^n} \rho_n^{\mu /2}.
\end{align}
This is  a consequence of a geometric series argument and Definition \ref{sequence}, since it implies that 
\begin{align}
\boldsymbol{C}^{2(n+1)+2} \rho_{n}^{\mu} =  (\boldsymbol{C}^8 \rho_0^{\mu})^{1/2} 
(\boldsymbol{C}^4 \rho^{\mu})^{n/2} \rho_n^{\mu /2} \leq  \frac{1}{2^n} \rho_n^{\mu /2}. 
\end{align}
\end{remark}

\subsection{Proof of Proposition~\ref{thm:res}}

\label{sec:proof-res}

We define the sequence of vectors (see Remark \ref{R})
\begin{align}
\Psi^{(n)}_{\lambda_i}:=P^{(n)}_i \varphi_i\otimes \Omega , \quad n\in\N .
\end{align}
Due to the Property ($\mathcal P 3$) in Theorem \ref{thm:ind}, we know that the sequence above converges to the non-zero limit $\Psi_{\lambda_i}:=P_i\varphi_i\otimes \Omega\neq 0$ (see the discussion above Eq.\ \eqref{eq:p1proof}). 
 Note that (see Remark \ref{R})
\begin{align}
H^{\theta}=H^{(n),\theta}+H_f^{(n,\infty),\theta}+gV^{(n,\infty),\theta}
= \tilde H^{(n),\theta} + gV^{(n,\infty),\theta}
\end{align} 
and set $z=\lambda^{(n)}_i - 10 \rho_n  e^{-i\nu} $. Then,
\begin{align}
H^{\theta}\Psi^{(n)}_{\lambda_i}=\lambda^{(n)}_i\Psi^{(n)}_{\lambda_i} +g (\lambda^{(n)}_i-z)V^{(n,\infty),\theta}\frac{1}{ \tilde H^{(n),\theta} -z}\Psi^{(n)}_{\lambda_i}. 
\end{align}
Lemma \ref{lemma:is-keyesttilde} implies that $V^{(n,\infty),\theta}\frac{1}{
\tilde H^{(n),\theta} -z}$ tends to zero as $ n$ tends to infinity.   We conclude that
\begin{align}
\lim\limits_{n\to \infty}H^\theta \Psi^{(n)}_{\lambda_i}=\lambda_i\Psi_{\lambda_i} .
\end{align}
As $H^\theta$ is a closed operator, $\Psi_{\lambda_i}$ belongs to its domain and is an eigenvector of $H^\theta$ corresponding to the eigenvalue $\lambda_i$. Furthermore, as $P_i$ is rank-one, $\Psi_{\lambda_i}$  spans its range.

\section{Analyticity}
\label{analyticity}
In this section we prove Theorems  \ref{thm:anaP} and \ref{thm:anaPg}, in particular, we show   analyticity of the projections $P_i$ and the eigenvalues $\lambda_i$ with respect to the coupling constant $g$ and the dilation parameter $\theta$.   We  only prove in detail  analyticity with respect to $\theta$,   the result for the coupling constant $g$ follows the same line of argumentation and it is actually simpler since $g$ only appears in the interaction term and the dependence is linear.  
This result is archived in two steps: At first, we prove these properties for $P_i^{(n)}$ and $\lambda_i^{(n)}$ which is straight-forward since there are spectral gaps between $\lambda_i^{(n)}$ and the rest of the spectrum of $H^{(n),\theta}$. 
For this purpose we collect several estimates leading to Lemma \ref{lemsud} where, among other things, we show that the resolvent is differentiable with respect to the dilation parameter $\theta$. This together with the properties of the Riesz projection allows us to conclude \sout{the} analyticity of $P_i^{(n)}$ and $\lambda_i^{(n)}$ with respect to  $\theta$.
Secondly, we take the uniform limit $n\to\infty$ in order to conclude  that   the statement also  holds for $P_i$ and $\lambda_i$ (see Theorem \ref{thm:anaPp} below). This is possible  because   our estimates are uniform in the parameter $\theta$.

In this section we assume that Definitions \ref{C}, \ref{sequence}, \ref{gzero}  hold true.  We recall that we use the symbol $c$ to represent any generic (indeterminate) constant
that does not depend on   $n$, $g$, $\rho$, $\rho_0$ and dilation parameters (here do not only use $\theta$, but also $\eta$ and  $\lambda$).  
\begin{lemma}\label{ana1}
For every $r > 0$ and $\lambda, \eta \in  D(0, \pi/16) $ (and every $n \in \N$):
\begin{align}\label{lis1}
\Big \|    \frac{   H_0^{(n), \lambda} + r  }{ H_0^{(n), \eta  } + r  }  \Big  \| \leq & 10.
\end{align} 
Moreover, for large enough $r$   (independent of  $n$, $g$, $\rho$, $\rho_0$ and $\eta$ and  
$\lambda$) and every $z$ in the resolvent set of $  H^{(n), \eta} $:
\begin{align}\label{lis2}
\Big \|   H_0^{(n), \lambda}   \frac{ 1 }{     H^{(n), \eta}  - z  } \Big  \|
\leq 20 + 20(|z| + r/2 ) \Big \|   \frac{1}{  H^{(n), \eta}  - z } \Big \|. 
\end{align}
\end{lemma}
\begin{proof}
Notice that for every $\eta, \lambda \in  D(0, \pi/16) $ 
\begin{align}
\Big \|    \frac{   H_0^{(n), \lambda} + r  }{ H_0^{(n), \eta  } + r  }  \Big  \| \leq & \sup_{s \geq 0, i \in {0,1}}
\Big |    \frac{   e_i +  e^{-\lambda }s + r  }{    e_i + e^{-\eta }s  + r  }  \Big  | . 
\end{align} 
For every $s \geq 0$ and $ i \in \{0, 1 \} :$
\begin{align}
\Big |    \frac{   e_i + e^{-\lambda } s + r  }{    e_i + e^{-\eta }s  + r  }  \Big  |
\leq & |e^{\eta - \lambda} | + \Big |   e^{\eta - \lambda }  \frac{    r + e_i  }{    e_i + e^{-\eta }s  + r  }  \Big  | +
\Big |    \frac{   e_i +  r  }{    e_i + e^{-\eta }s  + r  }  \Big  |   \leq   10,
\end{align}
where we use that $\eta, \lambda \in    D(0, \pi/16)  $ and $e_i \geq 0$. This implies \eqref{lis1}.

It follows from Appendix \ref{app:sa} that there is a constant $c$ that does not depend on $n$, $g$, $\rho$, $\rho_0$ and $\eta$ such that  for every $r \geq  1$:
\begin{align}
\Big \|  V^{(n), \eta}   \frac{1}{ (H_0^{(n), 0} + r)^{1/2}}  \Big \| \leq c .
\end{align}
In conclusion, 
\begin{align}
\Big \|  V^{(n), \eta}   \frac{1}{ H_0^{(n), 0} + r}  \Big \| \leq \frac{c}{r^{1/2}} 
\end{align}
holds true.
 It follows that there is a constant 
$   C_{\eqref{zee}} $ that does not depend on $n$, $g$, $\rho$, $\rho_0$ and $\eta$ such that  for every $r \geq  1$:
\begin{align} \label{zee}
\Big \|  V^{(n), \eta}   \frac{1}{ H_0^{(n), \eta} + r}  \Big \| \leq   \frac{C_{\eqref{zee}}}{r^{1/2}} . 
\end{align}
Take $ \phi $ in the domain of $ H_0^{(n),  \eta}   $, $z \in \mathbb{C}$ and  $ r \geq 4 C_{\eqref{zee}}^2 $. Then  we have 
(recall that $|g| \leq 1$): 
\begin{align}
\|   H_0^{(n), \eta}   \phi   \|  \leq & \|  ( H^{(n), \eta}  - z )  \phi   \| +
\|  V^{(n), \eta} \phi \|  +  |z| \|  \phi \|  \\ \notag  \leq & 
\|  ( H^{(n), \eta}  - z )  \phi   \| +
(1/2)\|  H_0^{(n), \eta} \phi \|  +  (|z| + r/2 ) \|  \phi \|.  
\end{align}
Then, we obtain, for $z$ in the resolvent set of $  H^{(n), \eta} $ and $s  > 0$ (we take the term 
$  (1/2)\|  H_0^{(n), \eta} \phi \|  $ in the previous equation to the other side and $ \phi $ of the form 
$   \frac{1}{  H^{(n), \eta}  - z } \psi $):
\begin{align}
 \Big \| ( H_0^{(n), \eta} + s)   \frac{ 1 }{     H^{(n), \eta}  - z  } \Big  \|
 \leq & 2 + 2 (|z| + (r+ 2 s)/2 ) \Big \|   \frac{1}{  H^{(n), \eta}  - z } \Big \|.
\end{align}
Using \eqref{lis1}, we find
\begin{align}
&\Big \|  ( H_0^{(n), \lambda} + s  )  \frac{ 1 }{     H^{(n), \eta}  - z  } \big \| \leq  
\Big \| \frac{  H_0^{(n), \lambda} + s  }{  H_0^{(n), \eta} + s  } ( H_0^{(n), \eta} + s)   \frac{ 1 }{     H^{(n), \eta}  - z  } \Big  \|
\notag \\ 
&\leq  10 \Big \| ( H_0^{(n), \eta} + s)   \frac{ 1 }{     H^{(n), \eta}  - z  } \Big  \|
\leq  20 + 2 0(|z| + (r + 2 s)/2 ) \Big \|   \frac{1}{  H^{(n), \eta}  - z } \Big \|.
\end{align}
Taking the limit $s $ to zero, we arrive at Eq.\ \eqref{lis2}. 
\end{proof}
\begin{lemma}  \label{pppp}
For every  $\lambda, \eta, \theta \in  D(0, \pi/16)  $ (and every $n \in \mathbb{N}$),
there is a constant  $c$ (independent of  $n$, $g$, $\rho$, $\rho_0$, $\eta$, $\theta$ and  
$\lambda$) such that for  every $z$ in the resolvent set of $  H^{(n), \theta} $:
\begin{align}
\label{Mtama}
\Big \|  (H^{(n), \eta} -  H^{(n), \lambda} ) \frac{1}{  H^{(n), \theta} - z   } 
\Big \| \leq c (1+ |z|) |\eta - \lambda | \Big ( \Big \|   \frac{1}{  H^{(n), \theta} - z   }  
\Big \| + 1  \Big ) .
\end{align}
\end{lemma}
\begin{proof}
We  take a large enough $r > 0$ such that the results of Lemma \ref{ana1} hold true. We calculate
\begin{align}\label{toto1}
\Big \|  (H^{(n), \eta} -  H^{(n), \lambda} ) \frac{1}{  H^{(n), \theta} - z   } 
\Big \| \leq \Big \|  (H^{(n), \eta} -  H^{(n), \lambda} ) \frac{1}{ H_0^{(n), 0} + r  }  \Big \| 
 \Big \|(H_0^{(n), 0} + r )
 \frac{1}{  H^{(n), \theta} - z   } 
\Big \|.
\end{align}
Next, we notice that 
\begin{align}\label{toto2}
 \Big \|  (H_0^{(n), \eta} -  H_0^{(n), \lambda} ) \frac{1}{ H_0^{(n), 0} + r  }  \Big \|
 = \sup_{s \geq 0, i \in \{1, 2 \} }  \Big \| (e^{- \eta } - e^{-\lambda}) s   \frac{1}{ e_i + s + r  }  \Big \| \leq |   e^{- \eta } - e^{-\lambda} |. 
\end{align}
Using Appendix \ref{app:sa}, we   find a constant 
$ c $  (independent of  $n$, $g$, $\rho$, $\rho_0$ and $\eta$ and  
$\lambda$) such that 
\begin{align}\label{toto3}
 \Big \|  (V^{(n), \eta} -  V^{(n), \lambda} ) \frac{1}{ H_0^{(n), 0} + r  }  \Big \|
 \leq c |\eta - \lambda | 
\end{align}
Eqs.\ \eqref{toto1}-\eqref{toto3}, together with Lemma \ref{ana1}, imply the desired result. 
\end{proof}
\begin{definition}
For every $ \theta \in  D(0, \pi/16)  $, we set $h^{\theta} = \frac{\partial}{\partial \theta}   f^{ \theta} $ and
\begin{align} \label{operators444}
 \frac{\partial}{ \partial \theta }   V^{(n), \theta} := \sigma_1 \otimes 
 \left(a_n(  h^{\overline \theta})+ a_n(  h^{ \theta})^*  \right) 
\end{align} 
and (see Remark \ref{R})
\begin{align}
\frac{\partial}{ \partial \theta } H^{(n), \theta} : =   - H_f^{(n), \theta} + g\frac{\partial}{ \partial \theta } V^{(n), \theta} .
\end{align}
\end{definition}
\begin{lemma} \label{penu}
For every  $\lambda, \eta, \theta \in  D(0, \pi/16)  $ (and every $n \in \mathbb{N}$),
there is a constant  $c$ (independent of  $n$, $g$, $\rho$, $\rho_0$, $\eta$, $\theta$ and  
$\lambda$) such that for  every $z$ in the resolvent set of $  H^{(n), \theta} $:
\begin{align}
\label{Mtamaba}
\Big \| \Big ( \frac{1}{\eta - \lambda} (H^{(n), \eta} -  H^{(n), \lambda})  -  \frac{\partial}{ \partial \lambda } H^{(n), \lambda}\Big )  \frac{1}{  H^{(n), \theta} - z   } 
\Big \| \leq c (1+ |z|) |\eta - \lambda | \Big (  \Big \|   \frac{1}{  H^{(n), \theta} - z   } 
\Big \|   + 1 \Big ).
\end{align}
\end{lemma}
\begin{proof}
 The proof is very similar to the proof of Lemma \ref{pppp}, and therefore, we omit it.
\end{proof}
\begin{lemma}
\label{lemsud}
For every  $\lambda, \eta, \theta \in D(0, \pi/16)  $ (and every $n \in \mathbb{N}$),
there is a constant  $c$ (independent of  $n$, $g$, $\rho$, $\rho_0$, $\eta$, $\theta$   and  
$\lambda$) such that for  every $z$ in the resolvent set of both $  H^{(n), \eta} $ and $  H^{(n), \lambda } $:
\begin{align} \label{sud0}
\Big \| \frac{1}{\eta - \lambda} \Big ( \frac{1}{H^{(n), \lambda}  - z } &  -   \frac{1}{H^{(n), \eta}  - z } \Big )    -   \frac{1}{H^{(n), \lambda}  - z }  \frac{\partial }{\partial \lambda }  H^{(n), \lambda} \frac{1}{H^{(n), \lambda }  - z }        \Big \|  
\\ \notag  & \leq c ( 1 + |z|)^2 |  \eta - \lambda | \Big ( \Big \|  \frac{1}{H^{(n), \lambda }  - z }        \Big \| + 1 \Big )^2 \Big (  \Big \|  \frac{1}{H^{(n), \eta }  - z }        \Big \| + 1 \Big ) ,
\end{align}
and 
\begin{align} \label{sud01}
\Big \|  H_0^{(n), \theta}   \Big ( \frac{1}{\eta - \lambda} \Big ( \frac{1}{H^{(n), \lambda}  - z } &  -   \frac{1}{H^{(n), \eta}  - z } \Big )    -   \frac{1}{H^{(n), \lambda}  - z }  \frac{\partial }{\partial \lambda }  H^{(n), \lambda} \frac{1}{H^{(n), \lambda }  - z }      \Big )  \Big \|  
\\ \notag  & \leq c ( 1 + |z|)^2 |  \eta - \lambda | \Big ( \Big \|  \frac{1}{H^{(n), \lambda }  - z }        \Big \| + 1 \Big )^2 \Big (  \Big \|  \frac{1}{H^{(n), \eta }  - z }        \Big \| + 1 \Big ) ,
\end{align}
\end{lemma}
\begin{proof}
First, we notice that Lemma  \ref{pppp} and the resolvent identity imply 
\begin{align}\label{sud1}
\Big \| \frac{1}{H^{(n), \lambda}  - z } & (  H^{(n), \eta}  - H^{(n), \lambda }   )        \Big ( \frac{1}{H^{(n), \eta}  - z }  -   \frac{1}{H^{(n), \lambda}  - z } \Big )       \Big \|   \\ \notag  & \leq c ( 1 + |z|)^2 |  \eta - \lambda |^2 \Big ( \Big \|  \frac{1}{H^{(n), \lambda }  - z }        \Big \| + 1 \Big )^2 \Big (  \Big \|  \frac{1}{H^{(n), \eta }  - z }        \Big \| + 1 \Big ). 
\end{align}
We use the resolvent identity again and also Eq.\ \eqref{sud1} to obtain 
\begin{align}\label{sud2}
\Big \|  \Big ( \frac{1}{H^{(n), \lambda}  - z }  & -   \frac{1}{H^{(n), \eta}  - z } \Big )  -   \Big ( \frac{1}{H^{(n), \lambda}  - z } (  H^{(n), \eta}  - H^{(n), \lambda }   )     \frac{1}{H^{(n), \lambda}  - z } \Big )    \Big     \| 
 \\ \notag  & \leq c ( 1 + |z|)^2 |  \eta - \lambda |^2 \Big ( \Big \|  \frac{1}{H^{(n), \lambda }  - z }        \Big \| + 1 \Big )^2 \Big (  \Big \|  \frac{1}{H^{(n), \eta }  - z }        \Big \| + 1 \Big ) . 
\end{align}
Then, Eq.\ \eqref{sud0} follows from \eqref{sud2} and Lemma \ref{penu}. The proof of 
\eqref{sud01} follows in a similar fashion as the one of \eqref{sud0}, using Lemma \ref{ana1}. Therefore, we omit it.
\end{proof}
\begin{proposition}\label{propana1}
For every $\eta \in  D(0, \pi/16) $, the operator valued functions 
\begin{align}
\theta \in \mathcal S \mapsto P_i^{(n)},  \hspace{2cm} \theta \in \mathcal S \mapsto    
 H_0^{(n), \eta}  P_i^{(n)}  
\end{align}
are analytic. 
\end{proposition}
\begin{proof}
The proof is an obvious consequence of Lemma \ref{lemsud} and the formula for the Riesz projections as line integrals in the complex plane.   
\end{proof}
\begin{proposition}\label{propana2}
The complex valued function 
\begin{align}
\theta \in \mathcal S \mapsto \lambda_i^{(n)}
\end{align}
is analytic. 
\end{proposition}
\begin{proof}
We use the formalism  of the proof of Proposition \ref{prop:is-p1} and make  explicit the dependence of $ P^{(n)}_{i}  $ on $\theta$, i.e.,  $ P^{(n)}_{i} \equiv  P^{(n),\theta}_{i}   $.  We define $\Psi^{(n),\theta}_{i} =   P^{(n),\theta}_{i}\varphi_i\otimes \Omega^{(n)} $ (here we use a slightly different notation from  proof of Proposition \ref{prop:is-p1}). Notice that 
\begin{align}
\lambda_i^{(n)} = \frac{\langle  \Psi^{(n), \overline \theta}_{i}, H^{(n), \theta} \Psi^{(n),\theta}_{i}  \rangle}
{        \langle  \Psi^{(n), \overline \theta}_{i}, \Psi^{(n),\theta}_{i}  \rangle         }, 
\end{align}
 and that  the denominator does not vanish (this follows as in \eqref{eq:p1proof}). 
Then,  the result is a consequence of Proposition \ref{propana1}, because it implies that the 
functions 
\begin{align}
\theta \mapsto \Psi^{(n),\theta}_{i},   \hspace{1cm} \theta \mapsto H^{(n), \theta} \Psi^{(n),\theta}_{i} =  H^{(n), \theta}   \frac{1}{ H_0^{(n), 0} + 1  } \Big ( (  H_0^{(n), 0} + 1 )  P^{(n),\theta}_{i}\varphi_i\otimes \Omega^{(n)} \Big ) 
\end{align} 
are analytic.  Notice that the function $  \theta \mapsto   H^{(n), \theta}   \frac{1}{ H_0^{(n), 0} + 1  }  $ is an operator valued analytic function (the proof of this fact is similar to the proof of Lemma \ref{penu}, but much simpler).
\end{proof}

\begin{proposition}\label{propana3}
The maps 
\begin{align}
g \in D(0, g_0 ) \mapsto P_i^{(n)},  \qquad  g \in   D(0, g_0 ) \mapsto \lambda_i^{(n)}
\end{align}
are analytic.
\end{proposition}
\begin{proof}
The proof  follows directly from the proofs of Propositions \ref{propana1} and \ref{propana2}. In this case the proof is  much simpler because the coupling constant is only present in the interaction term (and the interaction term depends linearly on the coupling constant).   
\end{proof}

\begin{theorem}
\label{thm:anaPp}
The functions 
  \begin{align}
  \mathcal S \ni & \theta \mapsto  P_i, \hspace{1cm}  \mathcal S \ni \theta \mapsto  \lambda_i \\ \notag
   D(0, g_0) \ni & g \mapsto  P_i,   \hspace{1cm}  D(0, g_0)  \ni g \mapsto  \lambda_i
  \end{align} 
  are  analytic. Moreover, this implies that $ \lambda_i (\theta) \equiv  \lambda_i$ is constant for $\theta\in\mathcal S$ (see \eqref{def:setS}).
\end{theorem}
\begin{proof}
Theorem \ref{thm:ind},  Properties $(\mathcal P1)$  and $ ( \mathcal P 3 )$  imply that 
the convergence rates of $ \lambda^{(n)}_i $ to  $ \lambda_i $ and 
 $ P^{(n)}_i\otimes P_{\Omega^{(n, \infty)}} $ to $ P_i $ do
not depend on $ \theta  $ and $g$. Then $   \lambda_i $  and $ P_i $ are uniform limits of analytic functions (see Propositions \ref{propana1}, \ref{propana2}, \ref{propana3}). Therefore, they are analytic.  That  $\lambda_i$ is constant with respect to $\theta$ follows from the fact that it does not depend of the real part of $\theta$ 
 because a change in the real part of $ H^{\theta} $ produces unitarily equivalent Hamiltonians: if $ \theta $ and $ \tilde \theta $ have the same imaginary part, then $ H^{\theta} $ and $ H^{\tilde \theta} $ are unitarily equivalent (thus, isospectral).  Both $\lambda_i(\theta)$  and $ \lambda_i(\tilde \theta) $ are distinguished points in the spectrum because they are the vertex of  the same cone (see Theorem \ref{spectralestimates}), we conclude that $\lambda_i(\theta) = \lambda_i(\tilde \theta) $.           
\end{proof}

\begin{appendix}
\section{Closedness of H and standard estimates}
\label{app:sa}
In the following we shall use the well-known standard inequality
\begin{align} 
\begin{split}
    \|a(h)\Psi\|&\leq \|h/\sqrt\omega\|_2 \, \|H_f^{1/2}\Psi\|
    \\
     \|a(h)^*\Psi\|&\leq \|h/\sqrt\omega\|_2 \, \|H_f^{1/2}\Psi\| + \|h \|_2 \, \| \Psi\| 
    \end{split}
      \label{eq:st-est}
\end{align}
which holds for all  $h , h/\sqrt\omega \in\mathfrak h$  and $\Psi\in\mathcal H$ such that the
left- and right-hand side are well-defined; see \cite[Eq. (13.70)]{spohn_dynamics_2008}. 
\begin{lemma}
\label{lemma:standardest1}
Let $h, h/\sqrt{\omega}\in \mathfrak{h}$. Then, we have the following standard estimates
\begin{align}
\norm{a(h)^* (H_f+1)^{-\frac{1}{2}}}\leq \norm{h}_2 +\norm{h/\sqrt{\omega}}_2 \quad \text{and} \quad  \norm{a(h) (H_f+1)^{-\frac{1}{2}}}\leq   \norm{h/\sqrt{\omega}}_2 .
\label{eq:standartesta}
\end{align}
\end{lemma}
\begin{proof}
Let $\Psi\in \mathcal F [\mathfrak{h}]$ such that $\|\Psi \|_{\mathcal H}=1$. Then, it follows from \eqref{eq:ccr} that 
\begin{align}
&\norm{ a(h)^*  (H_f+1)^{-\frac{1}{2}}\Psi }^2 
\leq \left\langle \Psi, (H_f+1)^{-\frac{1}{2}} a(h)a(h)^* (H_f+1)^{-\frac{1}{2}} \Psi  \right\rangle
\notag \\
&  = \left\langle \Psi,(H_f+1)^{-\frac{1}{2}} \left(\norm{h}^2_2 + a(h)^*a(h) \right) (H_f+1)^{-\frac{1}{2}} \Psi \right\rangle
\notag \\
&\leq \norm{h}^2_2  +\left\langle \Psi,(H_f+1)^{-\frac{1}{2}}  a(h)^*a(h) (H_f+1)^{-\frac{1}{2}} \Psi \right\rangle
\notag \\
&=\norm{h}^2_2  +\norm{ a(h)  (H_f+1)^{-\frac{1}{2}}\Psi}^2 . 
\end{align} 
This implies that
\begin{align}
\label{eq:standest111}
\norm{ a(h)^*  (H_f+1)^{-\frac{1}{2}} }\leq \norm{h}_2  +\norm{ a(h)  (H_f+1)^{-\frac{1}{2}}}.
\end{align}
Moreover, it follows from the Cauchy-Schwarz inequality that
\begin{align}
&\norm{ a(h)  (H_f+1)^{-\frac{1}{2}}\Psi }\leq \int\mathrm{d}^3k\, |h(k)| \norm{a(k) (H_f+1)^{-\frac{1}{2}}\Psi}
\notag \\
&\leq \left( \int\mathrm{d}^3k\, |h(k)|^2/\omega(k)\right)^{\frac{1}{2}}\left( \int\mathrm{d}^3k\, \omega(k)  \norm{a(k) (H_f+1)^{-\frac{1}{2}}\Psi}^2    \right)^{\frac{1}{2}}
\notag \\
&\leq \norm{h/\sqrt \omega}_2 \norm{H_f(H_f+1)^{-1}}\leq \norm{h/\sqrt \omega}_2 .
\end{align}
This proves the second estimate, and the first one follows together with \eqref{eq:standest111}.
\end{proof}
\begin{lemma}
\label{lemma:standardest2}
\begin{align}
\norm{ V \left( H_0 +1  \right)^{-\frac{1}{2}}}&\leq 
  \norm{f}_2 +2\norm{f/\sqrt{\omega}}_2 .
\end{align}
\end{lemma}
\begin{proof}
 By definition in \eqref{eq:f}, we have $f, f/\sqrt{\omega}\in \mathfrak{h}$. Then, it follows from Lemma \ref{lemma:standardest1} that
\begin{align}
&\norm{ V \left( H_f +1  \right)^{-\frac{1}{2}}} 
 \leq \norm{ \sigma_1  \otimes a(f)\left( H_f +1  \right)^{-\frac{1}{2}}  }+ \norm{\sigma_1  \otimes a(f)^*\left( H_f +1  \right)^{-\frac{1}{2}}  }
\notag \\
& \leq  \norm{ a(f) \left( H_f +1  \right)^{-\frac{1}{2}} } +\norm{ a(f)^*\left( H_f +1  \right)^{-\frac{1}{2}}  }
 \leq \norm{f}_2 +2\norm{f/\sqrt{\omega}}_2 ,
\end{align}
 and furthermore, we obtain by the functional calculus together with $|e_i+r+1|\geq |r+1|$ that
\begin{align}
\norm{  \frac{(H_f +1)^{\frac{1}{2}} }{(H_0 +1)^{\frac{1}{2}} }}\leq \norm{  \frac{(H_f +1)^{1/2} }{(H_0 +1)^{1/2} }} &=\sup_{r\in [0,\infty), i=0,1} \frac{(r+1)^{1/2}}{(e_i+r+1)^{1/2}}
\leq 1 .
\end{align}
Then, we conclude 
\begin{align}
\norm{ V \left( H_0 +1  \right)^{-\frac{1}{2}}}
 \leq \norm{ V \left( H_f +1  \right)^{-\frac{1}{2}}} \norm{   \frac{(H_f +1)^{1/2} }{(H_0 +1)^{1/2} }} \leq \left( \norm{f}_2 +2\norm{f/\sqrt{\omega}}_2\right) .
\end{align}
This completes the proof.
\end{proof}
\begin{proof}[Proof of Proposition \ref{thm:Hsa}]
Lemma \ref{lemma:standardest2} implies that  $\| \lim_{s \to \infty} V \frac{1}{H_0 + s}  \| = 0  $ and, therefore, $V$ is bounded relatively to $H_0$, with infinitesimal bound. 
\end{proof}
\end{appendix}

\section*{Acknowledgement}
D.\ -A.\ Deckert and F.\ H\"anle would like to thank the IIMAS at UNAM and M.\
Ballesteros  the Mathematisches Institut at LMU Munich  for their hospitality. This
project was partially funded by the DFG Grant  DE 1474/3-1,  the grants PAPIIT-DGAPA
UNAM  IN108818, SEP-CONACYT 254062, and the junior research group ``Interaction
between Light and Matter'' of the Elite Network Bavaria. M.\  B.\  is a
Fellow of the Sistema Nacional de Investigadores (SNI).
F.\ H.\ gratefully acknowledges financial support by the ``Studienstiftung des deutschen Volkes''.
Moreover, the authors express their gratitude for the fruitful discussions with
 V.\ Bach, J.\ Faupin, J.\ S.\ M\o ller, A.\ Pizzo and W.\ De Roeck.

\bibliographystyle{amsplain}
\bibliography{ref}

\begin{thebibliography}{10}

\bibitem{ah}
A.~Abdesselam and D.~Hasler.
\newblock Analyticity of the ground state energy for massless nelson models.
\newblock {\em Comm. Math. Phys.}, 310(2):511--536, 2012.

\bibitem{bbf}
V.~Bach, M.~Ballesteros, and J.~Fr\"ohlich.
\newblock Continuous renormalization group analysis of spectral problems in
  quantum field theory.
\newblock {\em J. Funct. Anal.}, 268(5):749--823, 2015.

\bibitem{bbkm}
V.~Bach, M.~Ballesteros, M.~K{\"o}nenberg, and L.~Menrath.
\newblock Existence of ground state eigenvalues for the spin-boson model with
  critical infrared divergence and multiscale analysis.
\newblock {\em J. Math. Anal. and Appl.}, 453(2):773--797, 2017.

\bibitem{bach}
V.~Bach, M.~Ballesteros, and A.~Pizzo.
\newblock Existence and construction of resonances for atoms coupled to the
  quantized radiation field.
\newblock {\em ArXiv perprint: arXiv:1302.2829}, 2013.

\bibitem{bbp}
V.~Bach, M.~Ballesteros, and A.~Pizzo.
\newblock Existence and construction of resonances for atoms coupled to the
  quantized radiation field.
\newblock {\em Adv. Math.}, 314:540--572, 2017.

\bibitem{bcfs}
V.~Bach, T.~Chen, J.~Fr\"ohlich, and I.~M. Sigal.
\newblock Smooth {F}eshbach map and operator-theoretic renormalization group
  methods.
\newblock {\em J. Funct. Anal.}, 203:44--92, 2003.

\bibitem{bfs1}
V.~Bach, J.~Fr\"ohlich, and I.~M. Sigal.
\newblock Mathematical theory of nonrelativistic matter and radiation.
\newblock {\em Lett. Math. Phys.}, 34(3):183--201, 1995.

\bibitem{bfs3}
V.~Bach, J.~Fr\"ohlich, and I.~M. Sigal.
\newblock Quantum electrodynamics of confined nonrelativistic particles.
\newblock {\em Adv. Math.}, 137(2):299--395, 1998.

\bibitem{bfs2}
V.~Bach, J.~Fr\"ohlich, and I.~M. Sigal.
\newblock Renormalization group analysis of spectral problems in quantum field
  theory.
\newblock {\em Adv. Math.}, 137(2):205--298, 1998.

\bibitem{bfs100}
V.~Bach, J.~Fr\"ohlich, and I.~M. Sigal.
\newblock Spectral analysis for systems of atoms and molecules coupled to the
  quantized radiation field.
\newblock {\em Comm. Math. Phys.}, 207(2):249--290, 1999.

\bibitem{bmw}
V.~Bach, J.~S. M{\o}ller, and M.~C. Westrich.
\newblock Beyond the van hove timescale.
\newblock {\em preprint in preperation}.

\bibitem{bffs}
M.~Ballesteros, J.~Faupin, J.~Fr\"ohlich, and B.~Schubnel.
\newblock Quantum electrodynamics of atomic resonances.
\newblock {\em Comm. Math. Phys.}, 337(2):633--680, 2015.

\bibitem{bf}
J.-F. Bony and J.~Faupin.
\newblock Resolvent smoothness and local decay at low energies for the standard
  model of non-relativistic qed.
\newblock {\em J. Funct. Anal.}, 262:850--888, 2012.

\bibitem{f}
J.~Faupin.
\newblock Resonances of the confined hydrogen atom and the lamb-dicke effect in
  non-relativistic qed.
\newblock {\em Ann. Henri Poincar\'e}, 9:743--773, 2008.

\bibitem{fgs100}
J.~Fr\"ohlich, M.~Griesemer, and I.~M. Sigal.
\newblock Spectral renormalization group.
\newblock {\em Rev. in Math. Phys.}, 21:511--548, 2009.

\bibitem{feshbach}
M.~Griesemer and D.~Hasler.
\newblock On the smooth {F}eshbach-{S}chur map.
\newblock {\em J. Funct. Anal.}, 254(9):2329--2335, 2008.

\bibitem{gh}
M.~Griesemer and D.~Hasler.
\newblock Analytic perturbation theory and renormalization analysis of matter
  coupled to quantized radiation.
\newblock {\em Ann. Henri Poincar\'e}, 10(3):577--621, 2009.

\bibitem{hh101}
D.~Hasler and I.~Herbst.
\newblock Convergent expansions in non-relativistic qed: Analyticity of the
  ground state.
\newblock {\em J. Funct. Anal.}, 61(11):3119--3154, 2011.

\bibitem{hasler1}
D.~Hasler and I.~Herbst.
\newblock Ground states in the spin boson model.
\newblock {\em Ann. Henri Poincar\'e}, 12(4):621--677, 2011.

\bibitem{spohnspin}
M.~H\"ubner and H.~Spohn.
\newblock Spectral properties of the {S}pin-{B}oson {H}amiltonian.
\newblock {\em Ann. d'I.H.P Section A}, 64(2):289--323, 1995.

\bibitem{pizzo1}
A.~Pizzo.
\newblock One-particle (improper) states in nelson's massless model.
\newblock {\em Ann. Henri Poincar\'e}, 4:439-- 86, 2003.

\bibitem{pizzo2}
A.~Pizzo.
\newblock Scattering of an infraparticle: The one particle sector in nelson's
  massless model.
\newblock {\em Ann. Henri Poincar\'e}, 6:553--606, 2005.

\bibitem{reedsimon1}
M.~Reed and B.~Simon.
\newblock {\em Methods of modern mathematical physics I: Analysis of
  Operators}.
\newblock Academic Press, 1978.

\bibitem{reedsimon2}
M.~Reed and B.~Simon.
\newblock {\em Methods of modern mathematical physics II: Fourier Analysis,
  Self-adjointness}.
\newblock Academic Press, 1978.

\bibitem{s}
I.~M. Sigal.
\newblock Ground state and resonances in the standard model of the
  non-relativistic {QED}.
\newblock {\em J. Stat. Phys.}, 134(5-6):899--939, 2009.

\bibitem{spohn_dynamics_2008}
H.~Spohn.
\newblock {\em Dynamics of {Charged} {Particles} and their {Radiation}
  {Field}}.
\newblock Cambridge University Press, Cambridge, 1 edition, 2008.

\end{thebibliography}
\end{document}